\documentclass[a4paper,11pt]{article}

\usepackage{jheppub} 

\usepackage{amsfonts,bm,empheq}
\usepackage{amscd}
\usepackage{multirow}
\usepackage{amsthm}
\usepackage{color}

\allowdisplaybreaks

\newtheorem{thm}{Theorem}[section]

\newtheorem{prop}[thm]{Proposition}

\theoremstyle{definition}
\newtheorem{Def}[thm]{Definition}
\newtheorem{rem}[thm]{Remark}

\newcommand{\e}{\mathrm{e}}

\def\IN{\mathbb {N}}

\def\IC{\mathbb {C}}
\def\IP{\mathbb {P}}

\def\sff{\mathsf{f}}

\def\sfK{\mathsf{K}}
\def\sfF{\mathsf{F}}
\def\sfDel{\mathsf{\Delta}}

\def\nn{n}
\def\wkp{\widetilde{\kappa}}
\def\Psd{\Psi^{\dagger}}

\def\Psdw{\widetilde{\Psi}^{\dagger}}
\def\vPsdw{\widetilde{\Psi}^{\dagger}_{2}}

\makeatletter
\def\eqnarray{%
 \stepcounter{equation}%
 \let\@currentlabel=\theequation
 \global\@eqnswtrue
 \global\@eqcnt\z@
 \tabskip\@centering
 \let\\=\@eqncr
 $$\halign to \displaywidth\bgroup\@eqnsel\hskip\@centering
 $\displaystyle\tabskip\z@{##}$&\global\@eqcnt\@ne
 \hfil$\displaystyle{{}##{}}$\hfil
 &\global\@eqcnt\tw@$\displaystyle\tabskip\z@{##}$\hfil
 \tabskip\@centering&\llap{##}\tabskip\z@\cr}
\makeatother

\title{A Hamiltonian Formalism for Topological Recursion}

\author[a]{Hiroyuki Fuji,}
\author[b]{Masahide Manabe,}
\author[c]{and Yoshiyuki Watabiki}

\affiliation[a]{Center for Mathematical and Data Sciences and Department of Mathematics, \\
Kobe University,
Rokko, Kobe 657-8501, Japan}
\affiliation[b]{Liberal Arts \textnormal{\&} Data Science Unit, Tottori University,\\
4-101 Koyama-cho Minami, Tottori, 680-8550, Japan}
\affiliation[c]{Department of Physics,
Institute of Science Tokyo,\\
Oh-okayama 2-12-1, Meguro-ku, Tokyo 152-8551, Japan}

\emailAdd{hfuji@math.kobe-u.ac.jp}
\emailAdd{masahidemanabe@gmail.com}
\emailAdd{yoshiyuki.watabiki@gmail.com}

\abstract{
We propose a string field Hamiltonian formalism that associates a class of spectral curves and provides their quantization through the Chekhov-Eynard-Orantin topological recursion.
As illustrative examples, we present Hamiltonians for the $(2,2m-1)$ minimal discrete and continuum dynamical triangulation (DT) models, the supersymmetric analogue of minimal continuum DT models, the Penner model, and 4D $\mathcal{N}=2$ $SU(2)$ gauge theories in the self-dual $\Omega$-background.
}

\begin{document}
\maketitle
\flushbottom

\section{Introduction}\label{sect:introduction}

Topological recursion originally proposed by Eynard and Orantin in 2007 \cite{Eynard:2007kz} (see e.g., \cite{Bouchard:2024fih} for a review including generalizations and applications) provides a formalism to quantize spectral curves, 
where the quantization is induced by coupling to two-dimensional 
(string world-sheet) gravity. 
It was formulated based on the prototype of the recursion that was derived, 
from the Schwinger-Dyson (SD) equations, in matrix models by Chekhov, Eynard and Orantin \cite{Eynard:2004mh,Chekhov:2006vd} to perturbatively compute matrix model resolvents (correlation functions) in the $1/N$-expansion, where $N$ is the size of matrices.
Let $(\Sigma; x,y,B)$ be a spectral curve data consisting of a compact Riemann surface $\Sigma$, meromorphic functions $x, y: \Sigma \to {\IP}^1$ such that the zeros of $dx$ ($=$ branch points of spectral curve) do not coincide with the zeros of $dy$, and a meromorphic symmetric bi-differential $B$ on $\Sigma^{2}$, 
where the zeros of $dx$ are further assumed to be simple as in \cite{Eynard:2007kz}.
Then, the Chekhov-Eynard-Orantin (CEO) topological recursion in \cite{Chekhov:2006vd, Eynard:2007kz} recursively defines symmetric multi-differentials $\omega_{\nn}^{(g)}(z_1,\ldots,z_{\nn})$ on $\Sigma^{\nn}$ labeled by two integers $g \ge 0$ and $\nn \ge 1$ from the spectral curve data: $\omega_{1}^{(0)}(z)=y(z)dx(z)$ and $\omega_{2}^{(0)}(z_1,z_2)=B(z_1,z_2)$, 
and $g$ and $\nn$ correspond to the genus and the number of marked points of string world-sheet, respectively. In matrix models, the multi-differentials $\omega_{\nn}^{(g)}(z_1,\ldots,z_{\nn})$ give the expansion coefficients of 
$n$-point resolvents in the $1/N$-expansion.

In this paper, we particularly focus on a class of spectral curves described by
\begin{align}
\begin{split}
&
p(x)^2\, y^2 = 
q(x)^2\, \sigma(x)\,,
\qquad
\sigma(x):=\prod_{k=1}^{\mathfrak{b}}(x-\alpha_k)\,,
\\
&
p(x):=\sum_{k=\mu}^{r} p_k\, x^{k}\,,
\quad
r \ge \mu \ge 0,
\qquad
q(x):=\sum_{k=0}^{s-h} q_k\, x^k\,,
\quad s \ge h:=\left\lfloor \frac{\mathfrak{b}-1}{2}\right\rfloor,
\label{sp_curve_int}
\end{split}
\end{align}
where $\alpha_k$ ($k=1,2,\ldots,\mathfrak{b}$, $\alpha_i \ne \alpha_j$ for $i\ne j$) 
are branch points of spectral curves, and $h$ is the genus of $\Sigma$.
In Sections \ref{sec:ham_tr} and \ref{sec:ham_ai}, we separately examine the cases $\mathfrak{b}=2h+2$ and $\mathfrak{b}=2h+1$, under the respective conditions $s+2 \ge r$ and $s+1 \ge r$.
For the above class of spectral curves, the bi-differential $B$ in this paper is characterized by the conditions:
\begin{align}
&
\bullet\ \
\textrm{the pole is only at the diagonal $z_1=z_2$ as}\ 
B(z_1,z_2)\mathop{\sim}_{z_1 \to z_2} \frac{dz_1dz_2}{(z_1-z_2)^2}+\textrm{regular}\,,
\nonumber\\
&
\bullet\ \
\oint_{[\alpha_{2i-1}, \alpha_{2i}]}B(z_1,\cdot)=0\,, \quad i=1,\ldots,h\,.
\nonumber
\end{align}
For example, when $h=0$ ($\Sigma={\IP}^1$), 
\begin{align}
B(z_1,z_2)=
\frac{dz_1 dz_2}{\left(z_1-z_2\right)^2}
=
\begin{cases}
\frac{dx_1 dx_2}{2\left(x_1-x_2\right)^2}
\left(
\frac{\left(x_1+x_2\right)/2-\alpha_1}{\sqrt{\sigma(x_1) \sigma(x_2)}}+1\right)\ &\textrm{for}\ \ \mathfrak{b}=1\,,
\\
\frac{dx_1 dx_2}{2\left(x_1-x_2\right)^2}
\left(
\frac{x_1x_2-\left(\alpha_1+\alpha_2\right)\left(x_1+x_2\right)/2+\alpha_1\alpha_2}{\sqrt{\sigma(x_1) \sigma(x_2)}}+1\right)\ 
&\textrm{for}\ \ \mathfrak{b}=2\,,
\end{cases}
\nonumber
\end{align}
where $z_1, z_2 \in {\IP}^1$ can be introduced, for example, by 
$x(z)=z^2+\alpha_1$ for $\mathfrak{b}=1$, and 
$x(z)=(\alpha_1+\alpha_2)/2-(\alpha_1-\alpha_2)(z+z^{-1})/4$ 
(the Zhukovsky variable) for $\mathfrak{b}=2$.

On the other hand, a string field Hamiltonian formalism in two-dimensional (2D) quantum gravity was formulated in \cite{Ishibashi:1993nq} at the continuum level, based on an earlier work \cite{Kawai:1993cj}, and in \cite{Watabiki:1993ym,Ambjorn:1996ne} at the discrete level.
In \cite{FMW2025a,FMW2025b} we showed that the perturbative amplitudes 
with respect to the string coupling constant 
for a class of dynamical triangulation (DT) models---described by the Hamiltonians in \cite{Watabiki:1993ym,Ambjorn:1996ne}---can be obtained via the CEO topological recursion.
It is then natural and intriguing to ask what form of string field Hamiltonian formalism is captured by the topological recursion formalism, and vice versa:
$$
\textrm{Hamiltonian formalism}\qquad
\stackrel{?}{\longleftrightarrow}\qquad
\textrm{topological recursion formalism}
$$
In this paper, we propose a string field Hamiltonian formalism which leads to the topological recursion for the class of spectral curves in \eqref{sp_curve_int}.%
\footnote{
In the context of the topological recursion formalism, 
a Hamiltonian representation of 
isomonodromic deformations of rational connections on $\mathfrak{gl}_2(\IC)$, 
which provides a generalization of the Painlev\'e Hamiltonian systems in 
\cite{Okamoto:P1,Okamoto:P2,Okamoto:P3,Okamoto:P4,OKSO:P5},
was presented in \cite{Marchal:2022hyz}.
Note that their Hamiltonian formalism, formulated in terms of Darboux coordinates, is 
different from our string field Hamiltonian formalism.}

This paper is organized as follows.
In Section \ref{sec:string_op_ex}, after introducing the string field Hamiltonian formalism, we briefly review the frameworks presented in \cite{Watabiki:1993ym,Ambjorn:1996ne}, along with our previous works \cite{FMW2025a,FMW2025b}, and motivate the proposal put forward in this paper.
In Sections \ref{sec:ham_tr} and \ref{sec:ham_ai}, we propose the Hamiltonians \eqref{gen_hamiltonian} and \eqref{gen_hamiltonian_ai} as working definitions 
for the class of spectral curves in \eqref{sp_curve_int}
---with an even and an odd number of branch points, respectively---
and show that the amplitudes defined by these Hamiltonians can be obtained perturbatively in the string coupling constant via the CEO topological recursion. 
As examples, we illustrate the $(2,2m-1)$ minimal discrete DT model \cite{Watabiki:1993ym}, the Penner model \cite{Penner1988}, and the 4D $\mathcal{N}=2$ $SU(2)$ gauge theory with $N_f=4$ hypermultiplets in Section \ref{subsec:ex_even}. 
We also illustrate, in Section \ref{subsec:ex_odd}, the $(2,2m-1)$ minimal continuum DT model \cite{Gubser:1993vx,Watabiki:1993ym,Ambjorn:1996ne}, its supersymmetric analogue, and the 4D $\mathcal{N}=2$ pure $SU(2)$ gauge theory.
In Appendix \ref{app:sp_to_sd}, we discuss a reverse construction of the Hamiltonians presented in Sections \ref{sec:ham_tr} and \ref{sec:ham_ai}, starting from the class of spectral curves given in \eqref{sp_curve_int}.
In Appendix \ref{app:decoupling}, as further illustrations of the Hamiltonians, we consider 4D $\mathcal{N}=2$ $SU(2)$ gauge theories with $N_f=0,1,2,3,4$ hypermultiplets from the viewpoint of hypermultiplet decoupling.

\section{Hamiltonian formalism}
\label{sec:string_op_ex}

In this section, we introduce the string field Hamiltonian formalism, and 
review a pure DT model as an example which leads to the topological recursion formalism.

\subsection{String operators and SD equation}
\label{subsec:string_op_sd}

We first define string operators \cite{Watabiki:1993ym,Ambjorn:1996ne}.

\begin{Def}[String operators]\label{def:string_op}
The string annihilation and string creation operators 
$\Psi(\ell)$ and $\Psd(\ell)$, $\ell \ge 1$, satisfy 
the commutation relations
\begin{align}
\left[\Psi(\ell), \Psd(\ell')\right]=\delta_{\ell,\ell'}\,,
\quad
\left[\Psi(\ell), \Psi(\ell')\right]
=\left[\Psd(\ell), \Psd(\ell')\right]=0\,,
\qquad \ell,\ell' \in {\IN}\,.
\label{string_com_rel}
\end{align}
For convenience we also introduce 
$\Psi(\ell)=\Psd(\ell) \equiv 0$ for $\ell \le 0$.
A vacuum $|\mathrm{vac}\rangle$ and its conjugate $\langle\mathrm{vac}|$ with pairing $\langle\mathrm{vac}|\mathrm{vac}\rangle=1$ are defined, in the bra-ket notation, by
\begin{align}
\Psi(\ell)\, |\mathrm{vac}\rangle=0\,,
\quad
\langle\mathrm{vac}|\, \Psd(\ell)=0\,.
\label{vac_gen}
\end{align}
Here $\Psi(\ell)$ (resp. $\Psd(\ell)$) is considered to annihilate (resp. create) a closed loop $S^1$ of circumference $\ell$.
\end{Def}

In this paper, we will introduce a $W^{(3)}$-type Hamiltonian $\mathcal{H}$ in Section \ref{sec:ham_tr} and a two-reduced $W^{(3)}$-type Hamiltonian $\mathcal{H}^{red}$ in Section \ref{sec:ham_ai}, composed of string operators up to three-string interaction terms, for two classes of spectral curves with an even and an odd number of branch points, respectively.%
\footnote{
We refer to the Hamiltonians \eqref{gen_hamiltonian} and \eqref{gen_hamiltonian_ai}, which include three-string interaction terms, as the $W^{(3)}$-type Hamiltonian and the two-reduced $W^{(3)}$-type Hamiltonian, respectively, following the terminology discussed in \cite{Ambjorn:1996ne} in the study of 2D quantum gravity (see also \cite{Dijkgraaf:1990rs,Fukuma:1990jw,Fukuma:1990yk}).}%
\footnote{
The $W^{(3)}$ operators are well known to appear in matrix models as cut-and-join operators that generate their partition functions from the vacuum states (see e.g., \cite{Morozov:2009xk,Morozov:2019gbt}).
At present, we do not find a clear relation between the formalism of cut-and-join operators and our string field Hamiltonian formalism, but it would be interesting to establish a precise mapping between them.}
A motivated example of $\mathcal{H}$ (resp. $\mathcal{H}^{red}$) is shown in Section \ref{subsubsec:discrete_pureDT} (resp. Section \ref{subsubsec:conti_pureDT}).
Our main objective in this paper is to show that 
the Hamiltonians $\mathcal{H}$ and $\mathcal{H}^{red}$ generate solutions determined by the CEO topological recursion.

For the first class of Hamiltonians $\mathcal{H}$, we introduce a discrete Laplace transform of string creation operators,
\begin{align}
\Psdw(x):=
\sum_{\ell \ge 1} x^{-\ell-1}\, \Psd(\ell)\,,
\label{psdw}
\end{align}
and the time evolution by $\mathcal{H}$ defines $\nn$-point amplitude
\begin{align}
f_{\nn}(\bm{x}_{I})=
\lim_{T \to \infty} \left\langle\mathrm{vac}\Big| \e^{-T \mathcal{H}}\,
\widehat{f}_{\nn}(\bm{x}_{I}) \Big|\mathrm{vac}\right\rangle,
\quad
\widehat{f}_{\nn}(\bm{x}_{I}):=
\Psdw(x_1) \Psdw(x_2) \cdots \Psdw(x_{\nn})\,,
\label{n_pt_amp}
\end{align}
where ${I}=\{1, \ldots, \nn\}$ and $\bm{x}_{I}=\{x_1, \ldots, x_{\nn}\}$.%
\footnote{
A certain boundary condition at $T=\infty$ should be required to determine the amplitude \eqref{n_pt_amp} uniquely (see also Remark \ref{rem:sp_curve_vac} in Section \ref{subsec:n1_sd}).
The well-definiteness of the amplitude at $T=\infty$ leads to the SD equation \eqref{sd_eq_gen}.
}
The SD equation for the $\nn$-point amplitude in string field Hamiltonian formalism is \cite{Ishibashi:1993nq},
\begin{align}
0&=\lim_{T \to \infty}\frac{\partial}{\partial T}
\left\langle\mathrm{vac}\Big| \e^{-T \mathcal{H}}\, 
\widehat{f}_{\nn}(\bm{x}_{I}) \Big|\mathrm{vac}\right\rangle
=
\lim_{T \to \infty}
\left\langle\mathrm{vac}\Big| \e^{-T \mathcal{H}}
\left[-\mathcal{H}, \widehat{f}_{\nn}(\bm{x}_{I})\right]
\Big|\mathrm{vac}\right\rangle
\nonumber\\
&=
\lim_{T \to \infty}
\sum_{i=1}^{\nn}
\left\langle\mathrm{vac}\Big| \e^{-T \mathcal{H}}\, 
\Psdw(x_1) \cdots \Psdw(x_{i-1})
\left[-\mathcal{H}, \Psdw(x_i)\right]
\Psdw(x_{i+1})\cdots \Psdw(x_{\nn})
\Big|\mathrm{vac}\right\rangle,
\label{sd_eq_gen}
\end{align}
where a condition, which is referred to as the ``no big-bang condition'',
\begin{align}
\mathcal{H} |\mathrm{vac}\rangle=0\,,
\label{no_bigbang}
\end{align}
is assumed and used in the second equality.
We can understand that the condition \eqref{no_bigbang} is imposed to avoid ``overcounting'' from the view point of dynamical triangulation (see e.g., Section 2.1.2 of \cite{FMW2025a}).

For the second class of Hamiltonians $\mathcal{H}^{red}$, we, instead, introduce 
\begin{align}
\vPsdw(x):=
\sum_{\ell \ge 1} x^{-\ell/2-1}\, \Psd(\ell)\,,
\label{psdw_v}
\end{align}
and $\nn$-point amplitude,
\begin{align}
\sff_{\nn}(\bm{x}_{I})=
\lim_{T \to \infty} \left\langle\mathrm{vac}\Big| \e^{-T \mathcal{H}^{red}}\,
\vPsdw(x_1) \vPsdw(x_2) \cdots \vPsdw(x_{\nn}) \Big|\mathrm{vac}\right\rangle.
\label{n_pt_amp_v}
\end{align}
As above, the SD equation for the $\nn$-point amplitude is
\begin{align}
0=
\lim_{T \to \infty}
\sum_{i=1}^{\nn}
\left\langle\mathrm{vac}\Big| \e^{-T \mathcal{H}^{red}}\, 
\vPsdw(x_1) \cdots \vPsdw(x_{i-1})
\left[-\mathcal{H}^{red}, \vPsdw(x_i)\right]
\vPsdw(x_{i+1})\cdots \vPsdw(x_{\nn})
\Big|\mathrm{vac}\right\rangle,
\label{sd_eq_gen_red}
\end{align}
where the ``no big-bang condition'',
\begin{align}
\mathcal{H}^{red} |\mathrm{vac}\rangle=0\,,
\label{no_bigbang_red}
\end{align}
is also assumed.

\subsection{Pure DT models}
\label{subsec:dt_rev}

In the following, we review the pure DT models \cite{Watabiki:1993ym,Ambjorn:1996ne} for 2D pure gravity as a concrete example of the string field Hamiltonian formalism summarized above. 
We present some computations of disk amplitudes to illustrate the connection with the CEO topological recursion in later sections.

\subsubsection{Discrete DT model for the pure gravity}\label{subsubsec:discrete_pureDT}

Here we briefly review the pure DT model consisting only of triangles (the basic-type discrete DT model), as presented in \cite{Watabiki:1993ym}, 
in order to motivate the construction of Hamiltonians intended for topological recursion.
The Hamiltonian $\mathcal{H}=\mathcal{H}_{\mathrm{basic}}$ of the basic-type discrete DT model, which belongs to a class of $W^{(3)}$-type Hamiltonians in this paper, is given by \cite{Watabiki:1993ym},%
\footnote{
Remark that the standard notation in DT models is found by $g_s \to G^{1/2}$, 
$\Psd(\ell) \to G^{-1/2} \Psd(\ell)$, $\Psi(\ell) \to G^{1/2} \Psi(\ell)$.\label{ft:dt}}
\begin{align}
\mathcal{H}_{\mathrm{basic}}&=
\sum_{\ell \ge 1} \ell 
\left(\Psd(\ell)-\kappa \Psd(\ell+1)-2\kappa \Psd(\ell-1)
\right)\Psi(\ell)
-g_s^{-1} \kappa\, \Psi(1) - 3g_s^{-1} \kappa\, \Psi(3)
\nonumber\\
&\ \
- g_s \kappa \sum_{\ell,\ell' \ge 1}
\left(\ell+\ell'-1\right) \Psd(\ell)\Psd(\ell')\Psi(\ell+\ell'-1)
- g_s \kappa \sum_{\ell,\ell' \ge 1} 
\ell \ell'\, \Psd(\ell+\ell'+1)\Psi(\ell)\Psi(\ell')\,,
\label{basic_h}
\end{align}
where $g_s$ is the string coupling constant and $\kappa \ge 0$ is the discrete cosmological constant, and $\mathcal{H}_{\mathrm{basic}}$ satisfies the no big-bang condition \eqref{no_bigbang} following by $\Psi(\ell)|\mathrm{vac}\rangle=0$.
Then, it is shown that the SD equation \eqref{sd_eq_gen} leads to
\begin{align}
0&=
g_s \kappa\, x_1^3\, f_{\nn+1}(x_1, \bm{x}_{I})
+  \left(\kappa\, x_1^2-x_1+2\kappa \right)f_{\nn}(\bm{x}_{I})
+g_s^{-1} \kappa \left(x_1^{-1}+x_1^{-3}\right) 
f_{\nn-1}(\bm{x}_{I \backslash \{1\}})
\nonumber\\
&\ \
+g_s \kappa \sum_{i=2}^{\nn} \partial_{x_i}
\frac{x_1^3 f_{\nn-1}(\bm{x}_{I \backslash \{i\}}) 
- x_i^3 f_{\nn-1}(\bm{x}_{I \backslash \{1\}})}{x_1-x_i}
+C_{\nn-1}(\bm{x}_{I \backslash \{1\}})\,,
\label{sp_sd_N_dt_d_sim}
\end{align}
where $C_{\nn-1}(\bm{x}_{I \backslash \{1\}})$ is 
a function of $\bm{x}_{I \backslash \{1\}}=\{x_2, \ldots, x_{\nn}\}$.
Note that the connected part $f_{\nn}^{\mathrm{con}}(\bm{x}_{I})$ of the $\nn$-point amplitude $f_{\nn}(\bm{x}_{I})$ for the basic-type discrete DT model is expanded around $x=\infty$, $\kappa=0$, $g_s=0$ as
\begin{align}
f_{\nn}^{\mathrm{con}}(\bm{x}_{I})=
\sum_{g \ge 0} g_s^{2g-2+n}
\sum_{\ell_1, \ldots, \ell_{\nn} \ge 1} x_1^{-\ell_1-1} \cdots x_{\nn}^{-\ell_{\nn}-1}
\sum_{N \ge 1} \sum_{S \in \mathcal{T}_{\nn}^{(g)}(\ell_1,\ldots,\ell_{\nn};N)} \kappa^{N}\,.
\label{dt_count}
\end{align}
Here $\mathcal{T}_{\nn}^{(g)}(\ell_1,\ldots,\ell_{\nn};N)$ is the set of all triangulated, oriented, and connected surfaces of genus $g$ with $\nn$ boundary loops of circumferences $\ell_1, \ldots, \ell_{\nn}$, and 
$N$ equilateral triangles of the same size.

For $n=1$ of \eqref{sp_sd_N_dt_d_sim}, in particular, 
\begin{align}
\begin{split}
0&=
g_s \kappa\, x_1^3\, f_{2}(x_1, x_1)
+  \left(\kappa\, x_1^2-x_1+2\kappa \right)f_{1}(x_1)
+g_s^{-1} \kappa \left(x_1^{-1}+x_1^{-3}\right) 
+C_{0}\,,
\\
C_{0}&=
-\kappa \lim_{T \to \infty} \left\langle\mathrm{vac}\Big| \e^{-T \mathcal{H}}\,
\Psd(1) \Big|\mathrm{vac}\right\rangle,
\label{sp_sd_sim_n1}
\end{split}
\end{align}
is obtained, and the asymptotic behavior of amplitudes under $g_s \to 0$:
\begin{align}
f_{1}(x) = g_s^{-1}\, f_{1}^{(0)}(x) + O(g_s^0)\,,\quad
f_{2}(x,x) = g_s^{-2}\, f_{1}^{(0)}(x)^2 + O(g_s^{-1})\,,
\end{align}
gives a planar SD equation for the disk amplitude $f_{1}^{(0)}(x)$:
\begin{align}
0&=
\kappa\, x^3\, f_{1}^{(0)}(x)^2
+  \left(\kappa\, x^2-x+2\kappa \right)f_{1}^{(0)}(x)
+ \kappa \left(x^{-1}+x^{-3}\right) 
+ C_{0}^{(0)}\,,
\nonumber\\
\Longleftrightarrow\ \
0&=
\left(f_{1}^{(0)}(x) 
+ \frac{1}{\kappa x^3}\left(\frac{\kappa}{2} x^2-\frac12 x+\kappa \right)\right)^2
\nonumber\\
&\qquad
-\frac{1}{4 \kappa^2 x^5}
\left(\kappa^2\, x^3 - 2 \kappa \left(1 + 2 C_{0}^{(0)}\right)x^2
+ x - 4\kappa\right),
\label{sp_sd_sim_n1_planar}
\end{align}
where $C_{0}^{(0)}=\lim_{g_s \to 0} g_s C_{0}$.
This planar SD equation defines a spectral curve of the basic-type discrete DT model:
\begin{align}
\begin{split}
&
y^2 = \frac{1}{4 \kappa^2 x^6}
\left(\kappa^2\, x^4 - 2 \kappa \left(1 + 2 C_{0}^{(0)}\right)x^3
+ x^2 - 4\kappa\, x\right),
\\
&
y := f_{1}^{(0)}(x) 
+ \frac{1}{\kappa x^3}\left(\frac{\kappa}{2} x^2-\frac12 x+\kappa \right).
\label{sp_curve_basic}
\end{split}
\end{align}
Here, by the one-cut ansatz of the spectral curve
\begin{align}
y=\frac{q(x)}{p(x)} \sqrt{\sigma(x)}\,,
\quad
p(x)=\kappa\, x^3\,,
\quad
q(x)= \frac{\kappa}{2} \left(x - \gamma \right),
\quad
\sigma(x)=x \left(x - \alpha\right),
\label{discrete_DT_sim_sp_curve}
\end{align}
the comparison between \eqref{sp_curve_basic} and \eqref{discrete_DT_sim_sp_curve} gives relations
\begin{align}
2\alpha \gamma+\gamma^2=\frac{1}{\kappa^2}\,,\quad
\alpha \gamma^2=\frac{4}{\kappa}\,,\quad
\alpha + 2\gamma = \frac{2}{\kappa}\left(1+2C_0^{(0)}\right),
\label{rel_basic_dt}
\end{align}
which determine $\alpha$ and $\gamma$ in terms of $\kappa$, and the constant $C_{0}^{(0)}$ is also determined as well.

One can show that the CEO topological recursion by Chekhov, Eynard and Orantin \cite{Chekhov:2006vd,Eynard:2007kz}, 
with the spectral curve \eqref{sp_curve_basic} as input, 
provides perturbative solutions around $g_s=0$ to 
the SD equation \eqref{sp_sd_N_dt_d_sim} for general $n$ \cite{FMW2025a}.

\begin{rem}\label{rem:critical_kappa}
The first two equations in \eqref{rel_basic_dt} yield the equation 
$\kappa^2 \gamma^3-\gamma+8\kappa=0$ for $\gamma$, 
and a solution to this equation is chosen so that the disk amplitude $f_{1}^{(0)}(x)$ admits the expansion given in \eqref{dt_count}. 
In particular, the behavior of the leading disk amplitude
\begin{align}
\begin{split}
\lim_{g_s \to 0} g_s
\lim_{T \to \infty} \left\langle\mathrm{vac}\Big| \e^{-T \mathcal{H}}\,
\Psd(1) \Big|\mathrm{vac}\right\rangle
=
-\frac{1}{\kappa}C_0^{(0)}
&=\frac{1}{2\kappa}-\frac{1}{\kappa\, \gamma^2}-\frac{\gamma}{2}
\\
&=\sum_{N \ge 1} \sum_{S \in \mathcal{T}_{1}^{(0)}(1;N)} \kappa^{N}\,,
\label{lead_disk_basic_dt}
\end{split}
\end{align}
around $\kappa=0$ determines the appropriate solution which behaves as $\gamma=\kappa^{-1}-4\kappa+O(\kappa^3)$ as $\kappa \to 0$, and actually 
the leading disk amplitude \eqref{lead_disk_basic_dt} behaves as $\kappa + O(\kappa^3)$ only for this solution.%
\footnote{
The other two solutions, which are inappropriate, behave as 
$\gamma=-\kappa^{-1}-4\kappa+O(\kappa^3)$ and 
$\gamma=8\kappa+O(\kappa^3)$ as $\kappa \to 0$, and 
they do not reproduce the leading behavior of the disk amplitude \eqref{lead_disk_basic_dt}, which is $\kappa^{-1} + O(\kappa)$ and $-\kappa^{-3}/64 + O(\kappa^{-1})$, respectively.}
Furthermore, by requiring that the solution yields a multiple root of the equation $\kappa^2 \gamma^3-\gamma+8\kappa=0$, 
one finds the critical value
\begin{align}
\kappa_c=\frac{1}{2 \cdot 3^{3/4}}\,,
\label{crit_kappa}
\end{align}
which is the reasonable maximal value of $\kappa$, as well as the critical values 
$\gamma_c=1/(\sqrt{3}\kappa_c)=2\cdot 3^{1/4}$ of $\gamma$ and 
$\alpha_c=4/(\kappa_c \gamma_c^2)=2\cdot 3^{1/4}=\gamma_c$ of $\alpha$.
In addition, for the critical value \eqref{crit_kappa}, the critical value $x_c$ of $x$ is given by the value that coincides with the branch point $\alpha_c$ of the spectral curve \eqref{discrete_DT_sim_sp_curve}, namely,
\begin{align}
x_c=\alpha_c=2 \cdot 3^{1/4}\,.
\label{crit_x}
\end{align}
\end{rem}

\subsubsection{Continuous DT model for the pure gravity}\label{subsubsec:conti_pureDT}

The Hamiltonian $\mathcal{H}^{red}=\mathcal{H}_{\mathrm{conti.DT}}^{\mathrm{pure}}$ of the continuum pure DT model, which belongs to a class of two-reduced $W^{(3)}$-type Hamiltonians in this paper, is given by \cite{Ambjorn:1996ne},%
\footnote{
Remark that the standard notation in continuum DT models is found by $g_s \to (2\mathcal{G})^{1/2}$, $\Psd(\ell) \to (2\mathcal{G})^{-1/2} \phi_{\ell}^{\dagger}$, $\Psi(\ell) \to (2\mathcal{G})^{1/2} \phi_{\ell}$.
Here, $\phi_{\ell}^{\dagger}$ and $\phi_{\ell}$ are distinct from the string operators in the literature; they actually denote the coefficients in the mode expansions of the Laplace-transformed string operators.
\label{ft:cont_dt}}
\begin{align}
\mathcal{H}_{\mathrm{conti.DT}}^{\mathrm{pure}}&=
-\frac{g_s}{4} \Psi(4)
- \frac{g_s}{8} 
\left(\Psi(1)-\frac{3\mu}{2}g_s^{-1}\right)^2\Psi(2)
- 2 \sum_{\ell \ge 1} \ell\, \Psd(\ell+1)\Psi(\ell)
\nonumber\\
&\ \
+ \frac{3\mu}{4}\sum_{\ell \ge 4} \ell\, \Psd(\ell-3)\Psi(\ell)
- g_s \sum_{\ell, \ell' \ge 1}
\left(\ell+\ell'+4\right) \Psd(\ell)\Psd(\ell')\Psi(\ell+\ell'+4)
\nonumber\\
&\ \
-\frac{g_s}{4}
\sum_{\ell, \ell' \ge 1}
\ell \ell'\, \Psd(\ell+\ell'-4) \Psi(\ell) \Psi(\ell')\,,
\label{cont_pure_dt_ham}
\end{align}
where $\mu \ge 0$ is the cosmological constant, and $\mathcal{H}_{\mathrm{conti.DT}}^{\mathrm{pure}}$ satisfies the no big-bang condition \eqref{no_bigbang_red}. 
By a similar way as in Section \ref{subsubsec:discrete_pureDT}, one obtains the one-cut spectral curve associated to the Hamiltonian $\mathcal{H}_{\mathrm{conti.DT}}^{\mathrm{pure}}$ as
\begin{align}
y:=
\sff_1^{(0)}(x) + x^{3/2}-\frac{3\mu}{8}x^{-1/2}
=\left(x -\frac{\sqrt{\mu}}{2}\right) \sqrt{x + \sqrt{\mu}}\,,
\label{cont_DT_sp_curve}
\end{align}
where $\sff_{1}(x) = g_s^{-1}\, \sff_{1}^{(0)}(x) + O(g_s^0)$ as $g_s \to 0$, 
and finds that the perturbative $n$-point amplitudes around $g_s=0$ are obtained \cite{FMW2025a} via the CEO topological recursion.

\begin{rem}\label{rem:cont_limit_basic}
The continuum pure DT model is obtained, by a continuum limit, from the discrete pure DT model in Section \ref{subsubsec:discrete_pureDT}. 
The continuum limit is achieved by setting
\begin{align}
x = x_c\, \e^{\epsilon\, \xi}\,,
\qquad
\kappa = \kappa_c\, \e^{-\frac{3}{16} \epsilon^2\, \mu}\,,
\qquad
g_s = \frac{\epsilon^{5/2}}{2}\, g_s'\,,
\label{cont_v_pure_dt}
\end{align}
and zooming in to the branch point $x=\alpha$, of the spectral curve \eqref{discrete_DT_sim_sp_curve} of the basic-type discrete DT model, 
for the critical values $\kappa=\kappa_c$ in \eqref{crit_kappa}.
Actually, the limit $\epsilon \to 0_{+}$ gives the spectral curve \eqref{cont_DT_sp_curve} of the continuum pure DT model:
\begin{align}
2x_c\, \epsilon^{-3/2}\, y = \left(\xi -\frac{\sqrt{\mu}}{2}\right) \sqrt{\xi + \sqrt{\mu}} + O(\epsilon)\,,
\label{eq:cont_pure_DT_spectral_curve}
\end{align}
up to an overall factor (see e.g., \cite{FMW2025a} in detail).
\end{rem}

\section{Hamiltonians for spectral curves with even number of branch points}
\label{sec:ham_tr}

In this section, we focus on a class of spectral curves with $(h+1)$-cut ($h \ge 0$) and even number of branch points, described by
\begin{align}
y=M(x) \sqrt{\sigma(x)}\,,
\quad M(x)=\frac{q(x)}{p(x)}\,,
\quad \sigma(x)=\prod_{k=1}^{2h+2} (x - \alpha_k)\,,
\label{sp_curve_g}
\end{align}
where $\alpha_k$ ($k=1, 2, \ldots, 2h+2$, $\alpha_i \ne \alpha_j$ for $i\ne j$) are branch points, and 
$p(x)$ and $q(x)$ are polynomials of $x$ of degrees 
$r$ and $s-h \ge 0$, respectively, as
\begin{align}
\begin{split}
p(x)&=\sum_{k=\mu}^{r} p_k\, x^{k}
\quad 
(p_{\mu} \ne 0, p_r \ne 0, r \ge \mu \ge 0)\,,
\\
q(x)&=\sum_{k=0}^{s-h} q_k\, x^k
\quad (q_{s-h} \ne 0, s \ge h)\,.
\label{pq_def}
\end{split}
\end{align}
We will associate a $W^{(3)}$-type Hamiltonian with the spectral curve \eqref{sp_curve_g}.

\subsection{\texorpdfstring{$W^{(3)}$}{W{(3)}}-type Hamiltonian and amplitudes}
\label{subsec:w3_h}

The spectral curve \eqref{discrete_DT_sim_sp_curve} of the basic-type discrete DT model is an example of the above class of spectral curves with $r=\mu=3$, $s=1$ and $h=0$. 
Remarkably, we see that the parameter $p_3=\kappa$, which determines $p(x)$, 
barely appears in the three-string interaction terms of the Hamiltonian \eqref{basic_h}, whereas the parameters $\alpha$ and $\gamma$, which determine $\sigma(x)$ and $q(x)$, do not barely appear in the Hamiltonian and are rather determined from the relations \eqref{rel_basic_dt}.
This is a key piece to construct a generalized Hamiltonian $\mathcal{H}$ such that the spectral curve \eqref{sp_curve_g} is associated and the amplitudes defined by $\mathcal{H}$ obey the CEO topological recursion, i.e., 
it is reasonable to assume that 
$p_k$ ($k=\mu, \mu+1, \ldots, r$) barely appear in three-string interaction terms of Hamiltonian, 
whereas $\alpha_k$ ($k=1, 2, \ldots, 2h+2$) and $q_k$ ($k=0, 1, \ldots, s-h$) do not barely appear in Hamiltonian and are rather determined from certain relations.
In Appendix \ref{app:rev_sd}, we construct a generalized Hamiltonian, reversely starting from the spectral curve \eqref{sp_curve_g}, by consulting the derivation of the spectral curve \eqref{sp_curve_basic} of the basic-type discrete DT model from the planar SD equation \eqref{sp_sd_sim_n1_planar} for the Hamiltonian \eqref{basic_h}. 
As a result, we find that the following $W^{(3)}$-type Hamiltonian is associated with the spectral curve \eqref{sp_curve_g}, and we adopt it as our working definition in this paper.

\begin{Def}\label{def:hamiltonian}
A $W^{(3)}$-type Hamiltonian $\mathcal{H}$ (see Fig. \ref{fig:hamiltonian}) is
\begin{align}
\begin{split}
-\mathcal{H}&=
g_s^{-1}
\sum_{k=1}^{\mu+2} k \, \kappa_{k}\, \Psi(k)
+
\tau_{0} \Psi(1)^2
+2\sum_{k=0}^{s+2} \tau_k
\sum_{\ell \ge 1}
\ell\, \Psd(\ell+k-2)\Psi(\ell)
\\
&\ \
+g_s \sum_{k=\mu}^{r} p_k
\sum_{\ell, \ell' \ge 1}
\left(\ell+\ell'-k+2\right)
\Psd(\ell)\Psd(\ell')\Psi(\ell+\ell'-k+2)
\\
&\ \
+g_s \sum_{k=\mu}^{r} p_k
\sum_{\ell, \ell' \ge 1}
\ell \ell'\, 
\Psd(\ell+\ell'+k-2) \Psi(\ell) \Psi(\ell')\,,
\label{gen_hamiltonian}
\end{split}
\end{align}
where $g_s$ is the string coupling constant.
\begin{figure}[t]
\centering
\includegraphics[width=140mm]{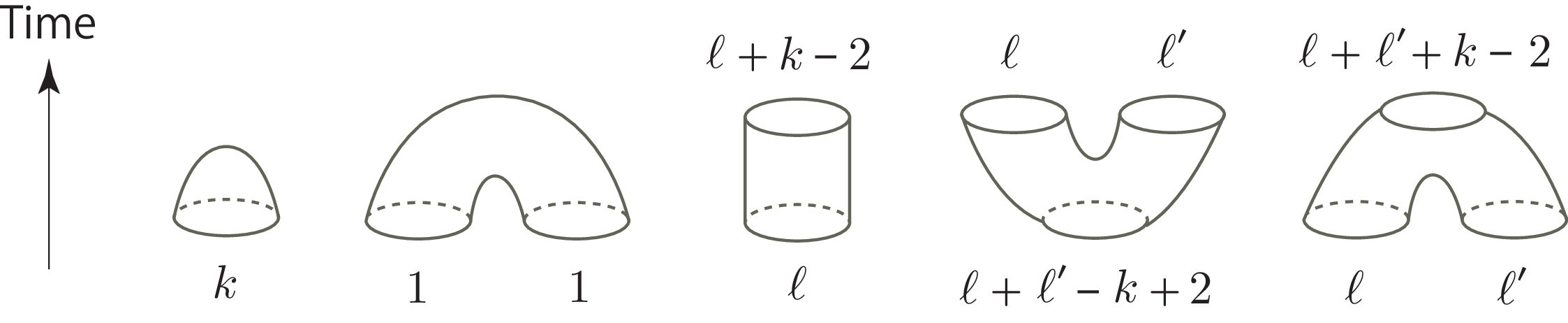}
\caption{Graphical representation of building blocks of the Hamiltonian \eqref{gen_hamiltonian} which generates the time evolution from bottom to top.}
\label{fig:hamiltonian}
\end{figure}
The parameters $\kappa_k$ ($k=1, 2, \ldots, \mu$), 
$\tau_k$ ($k=0, 1, \ldots, s+2$) and 
$p_k$ ($k=\mu, \mu+1, \ldots, r$, $p_{\mu} \ne 0$, $p_r \ne 0$) are independent parameters, 
whereas $\kappa_k$ ($k=\mu+1, \mu+2$) are not independent parameters determined 
by the planar SD equation in Section \ref{subsec:n1_sd} 
(see also \eqref{condition_coeff_cc_1}) as
\begin{align}
\kappa_{\mu+1}=
\frac{2 \tau_{0} \tau_{1}}{p_{\mu}}-\frac{\tau_{0}^2p_{\mu+1}}{p_{\mu}^2}\,,\ \
\kappa_{\mu+2}=\frac{\tau_{0}^2}{p_{\mu}}\,.
\label{condition_coeff_cc_2}
\end{align}
Here we impose a condition
\begin{align}
s+2 \ge r\,,
\label{condition_param_gen}
\end{align}
for obtaining the spectral curve \eqref{sp_curve_g} from the planar SD equation in Section \ref{subsec:n1_sd} (see also \eqref{condition_param}).
\end{Def}

\begin{rem}\label{rem:hamiltonian_param}
The input data for the Hamiltonian \eqref{gen_hamiltonian} are interpreted as
\begin{itemize}
\item
$\mu$ parameters $\kappa_k$ ($k=1, 2, \ldots, \mu$) are associated with $S^1 \to S^0$ (tadpole) processes (the first one of Fig. \ref{fig:hamiltonian}),
\item
$s+3$ parameters $\tau_k$ ($k=0, 1, \ldots, s+2$) are associated with $S^1 \to S^1$ (propagator) processes (the third one of Fig. \ref{fig:hamiltonian}),
\item 
$r-\mu+1$ parameters $p_k$ ($k=\mu, \mu+1, \ldots, r$) are associated with $S^1 \to S^1\times S^1$ and $S^1 \times S^1 \to S^1$ (three-string interaction) processes (the fourth and fifth ones of Fig. \ref{fig:hamiltonian}).
\end{itemize}
The Hamiltonian also has
\begin{itemize}
\item
two parameters $\kappa_{\mu+1}$ and $\kappa_{\mu+2}$ determined by \eqref{condition_coeff_cc_2}, and note that, when $\tau_{0}=0$, 
$\kappa_{\mu+1}=\kappa_{\mu+2}=0$.
\end{itemize}
The derivation of the spectral curve \eqref{sp_curve_g} as a planar SD equation will be discussed in Section \ref{subsec:n1_sd}, where the constrained parameters in \eqref{condition_coeff_cc_2} and the condition \eqref{condition_param_gen} will also be clarified.
The $r + s + 4 + h -\mu$ parameters of the spectral curve \eqref{sp_curve_g} are determined, in terms of the $r + s + 4$ independent parameters of the Hamiltonian \eqref{gen_hamiltonian}, by the planar SD equation and 
the $h$ $A$-periods in \eqref{periods_c}.
As a result, we see that the Hamiltonian contains $\mu$ redundant parameters with respect to the determination of the spectral curve.
\end{rem}

We now consider the $\nn$-point amplitude $f_{\nn}(\bm{x}_{I})$ in \eqref{n_pt_amp} determined by the Hamiltonian \eqref{gen_hamiltonian}. 
Note that it satisfies the no big-bang condition $\mathcal{H} |\mathrm{vac}\rangle=0$ in \eqref{no_bigbang}, due to $\Psi(\ell)|\mathrm{vac}\rangle=0$. Then, the following proposition is established.

\begin{prop}\label{prop:sd_eq}
The SD equation \eqref{sd_eq_gen} yields
\begin{align}
0&=
g_s \sum_{i=1}^{\nn} \partial_{x_i} \left(
\left(p(x_i)\, 
f_{\nn+1}(x_i, \bm{x}_{I})\right)_{\mathrm{irreg}(x_i)}
\right)
\nonumber\\
&\ \
+ \sum_{i=1}^{\nn} \partial_{x_i} \left( 
\left(2 \Delta(x_i)\,  f_{\nn}(\bm{x}_{I})\right)_{\mathrm{irreg}(x_i)}\right)
+ g_s^{-1} \sum_{i=1}^{\nn} \partial_{x_i} \left(
K_{-}(x_i)\, f_{\nn-1}(\bm{x}_{I \backslash \{i\}})\right)
\nonumber\\
&\ \
+ 2g_s \sum_{1 \le i < j \le \nn} \partial_{x_i}\partial_{x_j}
\frac{\left(p(x_i)\, f_{\nn-1}(\bm{x}_{I \backslash \{j\}})\right)_{\mathrm{irreg}(x_i)}
-\left(p(x_j)\, f_{\nn-1}(\bm{x}_{I \backslash \{i\}})\right)_{\mathrm{irreg}(x_j)}}{x_i-x_j}
\nonumber\\
&\ \
+ 2 \sum_{1 \le i < j \le \nn} \partial_{x_i}\partial_{x_j}
\frac{\Delta(x_i)_{\mathrm{irreg}(x_i)}
-\Delta(x_j)_{\mathrm{irreg}(x_j)}}{x_i-x_j}
\, f_{\nn-2}(\bm{x}_{I \backslash \{i,j\}})\,,
\label{key_sd}
\end{align}
where $F(x)_{\mathrm{irreg}(x)}:=F(x)-F(x)_{\mathrm{reg}(x)}$ for a function $F(x)$ of $x$, and 
$\mathrm{reg}(x)$ denotes the regular part of the Laurent series of $x$ at $x=0$.
Here $p(x)$ is defined in \eqref{pq_def}, and
\begin{align}
\Delta(x):=\sum_{k=-1}^{s+1}\tau_{k+1}\, x^k\,,
\qquad
K_{-}(x):=\sum_{k=1}^{\mu+2} \kappa_k\, x^{-k}\,.
\label{delta_k_def}
\end{align}
\end{prop}
\begin{proof}
By
\begin{align}
&
\left[-\mathcal{H}, \Psd(\ell)\right]
\nonumber\\
&=
g_s\, \ell \sum_{k=\mu}^{r} p_k\, 
\sum_{\substack{\ell',\ell''\ge 1 \\ \ell'+\ell''=\ell+k-2}}
\Psd(\ell')\Psd(\ell'')
+ 2\ell \sum_{k=-1}^{s+1}\tau_{k+1}\, \Psd(\ell+k-1)
+ g_s^{-1}\, \ell\, \kappa_{\ell} \sum_{k=1}^{\mu+2} \delta_{k, \ell}
\nonumber\\
&\ \
+2g_s \sum_{k=\mu}^{r} p_k
\sum_{\ell' \ge 1} \ell \ell'\,
\Psd(\ell+\ell'+k-2) \Psi(\ell')
+2 \tau_{0}\, \Psi(1)\, \delta_{\ell,1}\,,
\label{com_H1}
\end{align}
we obtain
\begin{align}
&
\Psd(\ell_1) \cdots \Psd(\ell_{i-1})
\left[-\mathcal{H}, \Psd(\ell_i)\right]
\Psd(\ell_{i+1})\cdots \Psd(\ell_{\nn})
\nonumber
\\
&=
\Psd(\ell_1)\cdots \breve{\Psi}^{\dagger}(\ell_i)\cdots \Psd(\ell_{\nn})
\nonumber\\
&\
\times
\Biggl(
g_s\, \ell_i \sum_{k=\mu}^{r} p_k\, 
\sum_{\substack{\ell,\ell'\ge 1 \\ \ell+\ell'=\ell_i+k-2}}
\Psd(\ell)\Psd(\ell')
+ 2 \ell_i \sum_{k=-1}^{s+1}\tau_{k+1}\, \Psd(\ell_i+k-1)
+ g_s^{-1}\, \ell_i\, \kappa_{\ell_i} \sum_{k=1}^{\mu+2} \delta_{k, \ell_i}
\nonumber\\
&\hspace{2.5em}
+2g_s \sum_{k=\mu}^{r} p_k
\sum_{\ell \ge 1} \ell_i \ell\,
\Psd(\ell_i+\ell+k-2)\Psi(\ell)
+2 \tau_{0}\, \Psi(1)\, \delta_{\ell_i,1}
\Biggr)
\nonumber\\
&\
+\sum_{j=i+1}^{\nn} \Biggl(
2g_s \sum_{k=\mu}^{r} p_k \ell_i \ell_j\,
\Psd(\ell_i+\ell_j+k-2)
+2 \tau_{0}\, \delta_{\ell_i, 1}\, \delta_{\ell_j, 1}
\Biggr)
\nonumber\\
&\hspace{18em}
\times
\Psd(\ell_1)\cdots \breve{\Psi}^{\dagger}(\ell_i)\cdots 
\breve{\Psi}^{\dagger}(\ell_j) \cdots \Psd(\ell_{\nn})\,,
\end{align}
where $\breve{\Psi}^{\dagger}(\ell)$ indicates the exclusion of $\Psd(\ell)$.
As a result, the SD equation \eqref{sd_eq_gen} gives \eqref{key_sd}.
\end{proof}

Proposition \ref{prop:sd_eq} leads to the following proposition.

\begin{prop}\label{prop:sd_eq_sep}
Symmetric solutions $f_{\nn}(\bm{x}_{I})$ in the variables $\bm{x}_{I}$ of the separated SD equation
\begin{align}
0&=
g_s \left(p(x_1)\, f_{\nn+1}(x_1, \bm{x}_{I})\right)_{\mathrm{irreg}(x_1)}
\nonumber\\
&\ \
+ \left(2 \Delta(x_1)\,  f_{\nn}(\bm{x}_{I})\right)_{\mathrm{irreg}(x_1)}
+ g_s^{-1} K_{-}(x_1)\, f_{\nn-1}(\bm{x}_{I \backslash \{1\}})
\nonumber\\
&\ \
+g_s \sum_{i=2}^{\nn} \partial_{x_i}
\frac{
\left(p(x_1)\, f_{\nn-1}(\bm{x}_{I \backslash \{i\}})\right)_{\mathrm{irreg}(x_1)}
-\left(p(x_i)\, f_{\nn-1}(\bm{x}_{I \backslash \{1\}})\right)_{\mathrm{irreg}(x_i)}}{x_1-x_i}
\nonumber\\
&\ \
+ \sum_{i=2}^{\nn} \partial_{x_i}
\frac{\Delta(x_1)_{\mathrm{irreg}(x_1)}
-\Delta(x_i)_{\mathrm{irreg}(x_i)}}{x_1-x_i}
\, f_{\nn-2}(\bm{x}_{I \backslash \{1,i\}})
+ C_{\nn-1}(\bm{x}_{I \backslash \{1\}})\,,
\label{key_sp_sd}
\end{align}
are also solutions to the equation \eqref{key_sd}, 
where $C_{\nn-1}(\bm{x}_{I \backslash \{1\}})$ is 
a function of $\bm{x}_{I \backslash \{1\}}=\{x_2, \ldots, x_{\nn}\}$.
\end{prop}

\subsection{Topological recursion}\label{subsec:tr}

In this section, we show that the separated SD equation \eqref{key_sp_sd} leads to the spectral curve \eqref{sp_curve_g}, and 
admits solutions that obey the CEO topological recursion for the (connected) amplitudes defined below.

\begin{Def}\label{def:free_en}
(Connected) amplitudes $F_{\nn}(\bm{x}_{I})$, $\nn \ge 1$, are
\begin{align}
F_1(x) = f_1(x) + g_s^{-1} \frac{\Delta(x)}{p(x)}\,,
\quad
F_{\nn}(\bm{x}_{I})=f_{\nn}^{\mathrm{con}}(\bm{x}_{I})
\ \textrm{for}\ \nn \ge 2\,,
\label{free_en_def}
\end{align}
where the connected part of the $\nn$-point amplitude $f_{\nn}(\bm{x}_{I})$ is denoted as $f_{\nn}^{\mathrm{con}}(\bm{x}_{I})$, e.g.,
\begin{align}
\begin{split}
f_2(x_1,x_2)&=f_2^{\mathrm{con}}(x_1,x_2)+f_1(x_1)f_1(x_2)\,,
\\
f_3(x_1,x_2,x_3)&=f_3^{\mathrm{con}}(x_1,x_2,x_3)
+f_2^{\mathrm{con}}(x_1,x_2)f_1(x_3)
+f_2^{\mathrm{con}}(x_1,x_3)f_1(x_2)
\\
&\ \
+f_2^{\mathrm{con}}(x_2,x_3)f_1(x_1)
+f_1(x_1)f_1(x_2)f_1(x_3)\,.
\end{split}
\end{align}
We consider perturbative expansions of the connected amplitudes in the string coupling constant $g_s$ as follows:
\begin{align}
f_{\nn}^{\mathrm{con}}(\bm{x}_{I})
=\sum_{g \ge 0} g_s^{2g-2+\nn} f_{\nn}^{(g)}(\bm{x}_{I})\,,
\quad
F_{\nn}(\bm{x}_{I})=\sum_{g \ge 0} g_s^{2g-2+\nn} F_{\nn}^{(g)}(\bm{x}_{I})\,.
\label{pert_exp_f}
\end{align}
\end{Def}

\subsubsection{\texorpdfstring{$\nn=1$}{n=1} separated SD equation and disk amplitude}\label{subsec:n1_sd}

When $\nn=1$, the separated SD equation \eqref{key_sp_sd} is
\begin{align}
0&=
g_s \left(p(x)\, f_{2}(x, x)\right)_{\mathrm{irreg}(x)}
+\left(2 \Delta(x)\,  f_{1}(x)\right)_{\mathrm{irreg}(x)} 
+ g_s^{-1} K_{-}(x) + C_0
\nonumber\\
&=
\left(g_s p(x)\, f_{2}^{\mathrm{con}}(x, x) 
+ g_s p(x)\, f_{1}(x)^2+2 \Delta(x)\,  f_{1}(x)\right)_{\mathrm{irreg}(x)}
+ g_s^{-1} K_{-}(x)\,,
\label{sd_n1}
\end{align}
where, in the second equality, $C_0=0$ is used, 
following from the asymptotic behavior as $x \to \infty$.
By applying the perturbative expansion \eqref{pert_exp_f}, 
the equation \eqref{sd_n1} yields
\begin{align}
0&=
\left(f_1^{(0)}(x) + \frac{\Delta(x)}{p(x)}\right)^2
- \frac{q(x)^2}{p(x)^2}\, \sigma(x)
=
F_1^{(0)}(x)^2 
- \frac{q(x)^2}{p(x)^2}\, \sigma(x)\,,
\label{sd_pert_n10}
\\
0&=
\left(p(x)\, F_{2}^{(g-1)}(x,x)
+p(x) \sum_{g_1+g_2=g} F_1^{(g_1)}(x)\, F_1^{(g_2)}(x)
\right)_{\mathrm{irreg}(x)}
\quad \textrm{for}\ g \ge 1\,,
\label{sd_pert_n1h}
\end{align}
where
\begin{align}
q(x)^2 \sigma(x)=
\Delta(x)^2 
+p(x)\left(2 \Delta(x)\,  f_{1}^{(0)}(x)
+p(x)\, f_1^{(0)}(x)^2\right)_{\mathrm{reg}(x)} 
- p(x)\, K_{-}(x)\,.
\label{sp_curve_determine}
\end{align}
As a result, we obtain the disk amplitude $F_1^{(0)}(x)$ which defines 
the spectral curve $y=F_1^{(0)}(x)$ in \eqref{sp_curve_g} 
as an $(h+1)$-cut solution ($0 \le h \le s$) to the equation \eqref{sd_pert_n10}:
\begin{align}
F_1^{(0)}(x)=M(x) \sqrt{\sigma(x)}\,,\ \
M(x)=\frac{q(x)}{p(x)}=
\frac{\sum_{k=0}^{s-h} q_k\, x^k}{\sum_{k=\mu}^{r} p_k\, x^{k}}\,,\ \
\sigma(x)=\prod_{k=1}^{2h+2} (x - \alpha_k)\,.
\label{planar_f_g}
\end{align}
By
\begin{align*}
&
\Delta(x)^2 = \left(\sum_{k=-1}^{s+1}\tau_{k+1}\, x^k\right)^2
= \tau_{s+2}^2\, x^{2s+2} + \cdots 
+ 2 \tau_0 \tau_1\, x^{-1} + \tau_0^2\, x^{-2}\,,
\\
&
p(x)\Bigl(2 \Delta(x)\,  f_{1}^{(0)}(x)
+p(x)\, f_1^{(0)}(x)^2\Bigr)_{\mathrm{reg}(x)} = 
\left(\sum_{k=\mu}^{r} p_k\, x^{k}\right)
\left(\sum_{k=0}^{s-1} c_k\, x^k\right)
\\
&\hspace{17.7em}
= p_r c_{s-1} x^{r+s-1} + \cdots,
\\
&
p(x) K_{-}(x) = 
\left(\sum_{k=\mu}^{r} p_k\, x^{k}\right)
\left(\sum_{k=1}^{\mu+2} \kappa_k\, x^{-k}\right)
\\
&\hspace{4.7em}
= p_r \kappa_1 x^{r-1} + \cdots 
+ \left(p_{\mu} \kappa_{\mu+1} + p_{\mu+1} \kappa_{\mu+2}\right)x^{-1}
+ p_{\mu} \kappa_{\mu+2}\, x^{-2}\,,
\end{align*}
where $c_k$ $(k=0,1, \ldots, s-1)$ are certain constants, 
the equation \eqref{sp_curve_determine} yields $2s+5$ conditions on 
the coefficients of powers of $x$, 
which determine 
$\kappa_{\mu+1}, \kappa_{\mu+2}$ as in \eqref{condition_coeff_cc_2}, 
the $s$ constants $c_k$, and $s+3$ parameters $q_k$ and $\alpha_k$ out of the total $s+h+3$ parameters. 
Here $2s+2 > r+s-1$, i.e., the condition $s+2 \ge r$ in \eqref{condition_param_gen} is imposed to determine the spectral curve,  
and in particular, 
\begin{align}
q_{s-h}=\tau_{s+2}\,,
\end{align}
is determined. 
The remaining $h$ parameters can be determined by $h$ $A$-periods
\begin{align}
P_i=
\frac{1}{2 \pi \mathsf{i}} \oint_{[\alpha_{2i-1}, \alpha_{2i}]} 
M(x)\sqrt{\sigma(x)}\, dx\,,
\qquad
i=1,2,\ldots, h\,.
\label{periods_c}
\end{align}

\begin{rem}\label{rem:sp_curve_vac}
The $(h+1)$-cut solution \eqref{planar_f_g} should be considered as the solution obtained by imposing a boundary condition at $T=\infty$ for the amplitude \eqref{n_pt_amp}.
\end{rem}

\subsubsection{\texorpdfstring{$\nn=2$}{n=2} separated SD equation and annulus amplitude}

When $\nn=2$, the separated SD equation \eqref{key_sp_sd} is
\begin{align}
0&=
g_s \left(p(x_1)\, f_{3}(x_1, x_1, x_2)\right)_{\mathrm{irreg}(x_1)}
+ \left(2 \Delta(x_1)\,  f_{2}(x_1,x_2)\right)_{\mathrm{irreg}(x_1)}
+g_s^{-1} K_{-}(x_1)\, f_1(x_2)
\nonumber\\
&\ \
+g_s \partial_{x_2}
\frac{\left(p(x_1)\, f_{1}(x_1)\right)_{\mathrm{irreg}(x_1)}
-\left(p(x_2)\, f_{1}(x_2)\right)_{\mathrm{irreg}(x_2)}}{x_1-x_2}
\nonumber\\
&\ \
+ \partial_{x_2}
\frac{\Delta(x_1)_{\mathrm{irreg}(x_1)}-\Delta(x_2)_{\mathrm{irreg}(x_2)}}{x_1-x_2}
+C_1(x_2)
\nonumber\\
&=
g_s \left(p(x_1)\left(f_3^{\mathrm{con}}(x_1,x_1,x_2) + 2 f_{2}^{\mathrm{con}}(x_1, x_2) f_1(x_1) 
+ f_{2}(x_1, x_1) f_1(x_2)\right)\right)_{\mathrm{irreg}(x_1)}
\nonumber\\
&\ \
+\left(2 \Delta(x_1)\left(f_{2}^{\mathrm{con}}(x_1,x_2) 
+ f_1(x_1) f_1(x_2)\right)\right)_{\mathrm{irreg}(x_1)}
+ g_s^{-1} K_{-}(x_1)\, f_1(x_2)
\nonumber\\
&\ \
+g_s \partial_{x_2}
\frac{\left(p(x_1)\, F_{1}(x_1)\right)_{\mathrm{irreg}(x_1)}
-\left(p(x_2)\, F_{1}(x_2)\right)_{\mathrm{irreg}(x_2)}}{x_1-x_2}
+C_1(x_2)\,,
\label{sd_n2}
\end{align}
and, by the $\nn=1$ separated SD equation \eqref{sd_n1}, yields
\begin{align}
0&=
g_s \left(p(x_1)\left(F_3(x_1,x_1,x_2) 
+ 2 F_{2}(x_1, x_2) F_1(x_1)\right)\right)_{\mathrm{irreg}(x_1)}
\nonumber\\
&\ \
+g_s \partial_{x_2}
\frac{\left(p(x_1)\, F_{1}(x_1)\right)_{\mathrm{irreg}(x_1)}
-\left(p(x_2)\, F_{1}(x_2)\right)_{\mathrm{irreg}(x_2)}}{x_1-x_2}
+C_1(x_2)\,.
\label{sd_n2_1}
\end{align}
By applying the expansion \eqref{pert_exp_f}, the leading part of the equation \eqref{sd_n2_1} is
\begin{align}
0&=
\left(2p(x_1) F_1^{(0)}(x_1) F_{2}^{(0)}(x_1,x_2)\right)_{\mathrm{irreg}(x_1)}
\nonumber\\
&\ \
+\partial_{x_2}
\frac{\left(p(x_1)\, F_{1}^{(0)}(x_1)\right)_{\mathrm{irreg}(x_1)}
-\left(p(x_2)\, F_{1}^{(0)}(x_2)\right)_{\mathrm{irreg}(x_2)}}{x_1-x_2}
+C_1^{(0)}(x_2)\,,
\label{sd_n1_1_0}
\end{align}
where $C_1^{(0)}(x_2)$ is a function of $x_2$.

For the $(h+1)$-cut spectral curve 
$y=F_1^{(0)}(x)=q(x) \sqrt{\sigma(x)}/p(x)$ in \eqref{planar_f_g}, 
the equation \eqref{sd_n1_1_0} yields
\begin{align}
0=
2\sqrt{\sigma(x_1)} F_{2}^{(0)}(x_1,x_2) 
+ \partial_{x_2} \frac{\sqrt{\sigma(x_1)}-\sqrt{\sigma(x_2)}}{x_1-x_2}
+ \frac{R_2^{(0)}(x_1, x_2)}{q(x_1)}\,,
\label{sd_n1_1_01}
\end{align}
where
\begin{align}
\begin{split}
R_2^{(0)}(x_1, x_2)&=
\partial_{x_2}
\frac{\left(q(x_1)-q(x_2)\right)\sqrt{\sigma(x_2)}}{x_1-x_2}
-\left(2 q(x_1) \sqrt{\sigma(x_1)} F_{2}^{(0)}(x_1,x_2)\right)_{\mathrm{reg}(x_1)}
\\
&\ \
-\partial_{x_2}
\frac{\left(q(x_1) \sqrt{\sigma(x_1)}\right)_{\mathrm{reg}(x_1)}
-\left(q(x_2) \sqrt{\sigma(x_2)}\right)_{\mathrm{reg}(x_2)}}{x_1-x_2}
+C_1^{(0)}(x_2)\,,
\end{split}
\end{align}
is a polynomial in $x_1$ of degree $s-1$. 
We now follow the argument by Eynard \cite{Eynard:2004mh}, and 
find a solution to the equation \eqref{sd_n1_1_01} such that 
\begin{itemize}
\item
the amplitude $F_{2}^{(0)}(x_1,x_2)$ has no poles in $x_1$ away from 
the branch cuts of the spectral curve, which lie 
along $[\alpha_1, \alpha_2] \cup \cdots \cup 
[\alpha_{2h+1}, \alpha_{2h+2}]$,
\end{itemize}
as follows.
Because $q(x_1)$ is the polynomial in $x_1$ of degree $s-h$, 
$R_2^{(0)}(x_1, x_2)/q(x_1)$ should be a polynomial in $x_1$ of degree $(s-1)-(s-h)=h-1$. 
In general, as a basis $\{L_1(x), \ldots, L_{h}(x)\}$ 
for degree $h-1$ polynomials of $x$, there is a unique set that satisfies
\begin{align}
\frac{1}{2\pi \mathsf{i}} \oint_{[\alpha_{2j-1}, \alpha_{2j}]} \frac{L_i(x)}{\sqrt{\sigma(x)}}\, dx= \delta_{i,j}\,,
\qquad
i,j=1,2,\ldots,h\,.
\label{basis_L}
\end{align}
In the present case, by integrations $\oint_{[\alpha_{2i-1}, \alpha_{2i}]}dx_1/2\pi \mathsf{i}\sqrt{\sigma(x_1)}$ $(i=1, \ldots, h)$, 
which do not contain $x_2$ inside the contours, 
of the equation \eqref{sd_n1_1_01} we find
\begin{align}
\frac{R_2^{(0)}(x_1, x_2)}{2q(x_1)} &= 
-\frac12 \partial_{x_2} \sqrt{\sigma(x_2)} \sum_{i=1}^{h} C_i(x_2)\, L_i(x_1)\,,
\nonumber\\
C_i(x_2) &:= \frac{1}{2\pi \mathsf{i}} 
\oint_{(x_2 \not\in) [\alpha_{2i-1}, \alpha_{2i}]} \frac{dx_1}{(x_1-x_2)\sqrt{\sigma(x_1)}}\,,
\label{rem_coeff}
\end{align}
and then
\begin{align}
F_{2}^{(0)}(x_1, x_2)\, dx_1 dx_2=
\frac12 \partial_{x_2} dS_{z_2}(z_1)\, dx_2
-\frac{dx_1 dx_2}{2\left(x_1-x_2\right)^2}\,,
\label{annulus_sol}
\end{align}
is obtained. Here
\begin{align}
dS_{z_2}(z_1)=
\frac{\sqrt{\sigma(x_2)}}{\sqrt{\sigma(x_1)}}
\left(\frac{1}{x_1-x_2}-\sum_{i=1}^{h}C_i(x_2)\, L_i(x_1)\right)dx_1\,,
\label{third_diff}
\end{align}
is the unique third kind differential on the spectral curve with 
the residue $+1$ (resp. $-1$) at $z_1=z_2$ (resp. $z_1= \overline{z}_2$, 
the conjugate point of $z_2$ as $x(\overline{z}_2)=x(z_2)$), and vanishing all the $A$-periods, where $z \in {\IP}^1$ is a variable, defined near a branch point of $\sqrt{\sigma(x)}$, which parametrizes the spectral curve as $x=x(z), y=y(z)$.

\subsubsection{Separated SD equation for general \texorpdfstring{$\nn$}{n} and topological recursion}\label{subsubsec:tr_w}

The connected part of the separated SD equation \eqref{key_sp_sd} 
for general $\nn \ge 2$ yields
\begin{align}
0&=
g_s
\Biggl(p(x_1)\, F_{\nn+1}(x_1,\bm{x}_{I})
+2 p(x_1)\, F_1(x_1) F_{\nn}(\bm{x}_{I})
\nonumber\\
&\ \ \qquad
+p(x_1) \sum_{\substack{I_1 \cup I_2=I \backslash \{1\} \\ I_1,I_2 \ne \emptyset}}
F_{|I_1|+1}(x_1, \bm{x}_{I_1}) F_{|I_2|+1}(x_1, \bm{x}_{I_2})
\Biggr)_{\mathrm{irreg}(x_1)}
\nonumber\\
&\ \
+g_s \sum_{i=2}^{\nn} \partial_{x_i} 
\frac{\left(p(x_1)\, F_{\nn-1}(\bm{x}_{I \backslash \{i\}})\right)_{\mathrm{irreg}(x_1)}
-\left(p(x_i)\, F_{\nn-1}(\bm{x}_{I \backslash \{1\}})\right)_{\mathrm{irreg}(x_i)}}{x_1-x_i}
+\widetilde{C}_{\nn-1}(\bm{x}_{I \backslash \{1\}})\,,
\label{sp_sd_con_gen}
\end{align}
where $I_1=\{i_1,\ldots,i_{|I_1|}\}$, $I_2=\{i_{{|I_1|}+1},\ldots,i_{\nn-1}\}$ 
are disjoint subsets of $I \backslash \{1\}=\{2,\ldots,\nn\}$, and
$\bm{x}_{I_1}=\{x_{i_1}, \ldots, x_{i_{|I_1|}}\}$, 
$\bm{x}_{I_2}=\{x_{i_{|I_1|+1}}, \ldots, x_{i_{\nn-1}}\}$, 
and $\widetilde{C}_{\nn-1}(\bm{x}_{I \backslash \{1\}})$ is a function of 
$\bm{x}_{I \backslash \{1\}}$.
By applying the expansion \eqref{pert_exp_f}, 
the SD equation \eqref{sp_sd_con_gen} yields
\begin{align}
F_{\nn}^{(g)}(\bm{x}_{I})&=
\frac{(-1)}{2F_1^{(0)}(x_1)}
\Biggl[F_{\nn+1}^{(g-1)}(x_1,\bm{x}_{I})
+\mathop{\sum_{g_1+g_2=g}}_{I_1 \cup I_2=I \backslash \{1\}}^{\textrm{no (0,1)}}
F_{|I_1|+1}^{(g_1)}(x_1, \bm{x}_{I_1}) F_{|I_2|+1}^{(g_2)}(x_1, \bm{x}_{I_2})
\nonumber\\
&\hspace{14em}
+\sum_{i=2}^{\nn} 
\frac{F_{\nn-1}^{(g)}(\bm{x}_{I \backslash \{i\}})}{(x_1-x_i)^2}\Biggr]
+\frac{R_{\nn}^{(g)}(\bm{x}_{I})}{2 p(x_1) F_1^{(0)}(x_1)}\,,
\label{sp_sd_con_gen_pert}
\end{align}
where ``no (0,1)'' in the summation means that it does not contain 
$F_1^{(0)}(x_1)$, and
\begin{align}
&
R_{\nn}^{(g)}(\bm{x}_{I})
\nonumber
\\
&=
\Biggl(p(x_1)\, F_{\nn+1}^{(g-1)}(x_1,\bm{x}_{I})
+p(x_1) \mathop{\sum_{g_1+g_2=g}}_{I_1 \cup I_2=I \backslash \{1\}}
F_{|I_1|+1}^{(g_1)}(x_1, \bm{x}_{I_1}) F_{|I_2|+1}^{(g_2)}(x_1, \bm{x}_{I_2})
\Biggr)_{\mathrm{reg}(x_1)}
\nonumber
\\
&\ \
+\sum_{i=2}^{\nn} \partial_{x_i} 
\frac{\left(p(x_1)\, F_{\nn-1}^{(g)}(\bm{x}_{I \backslash \{i\}})\right)_{\mathrm{reg}(x_1)}+
\left(p(x_i)\, F_{\nn-1}^{(g)}(\bm{x}_{I \backslash \{1\}})\right)_{\mathrm{irreg}(x_i)}}{x_1-x_i}
+C_{\nn-1}^{(g)}(\bm{x}_{I \backslash \{1\}})\,.
\end{align}
Here $C_{\nn-1}^{(g)}(\bm{x}_{I \backslash \{1\}})$ is a function of $\bm{x}_{I \backslash \{1\}}$. 
Note that the SD equation \eqref{sp_sd_con_gen_pert} can be 
regarded to include not only the SD equations for $\nn\ge 2, g \ge 0$, but also
the SD equations \eqref{sd_pert_n1h} for $\nn=1, g\ge 1$.

For the $(h+1)$-cut spectral curve 
$y=F_1^{(0)}(x)=q(x) \sqrt{\sigma(x)}/p(x)$ in \eqref{planar_f_g}, 
solutions to the SD equation \eqref{sp_sd_con_gen_pert} such that
\begin{itemize}
\item
the amplitude $F_{\nn}^{(g)}(\bm{x}_{I})$ on the left hand side of \eqref{sp_sd_con_gen_pert} has no poles in $x_1$ away from the branch cuts of the spectral curve, 
which lie along $\mathcal{C}:=[\alpha_1, \alpha_2] \cup \cdots \cup 
[\alpha_{2h+1}, \alpha_{2h+2}]$,
\end{itemize}
are obtained \cite{Eynard:2004mh} (see also \cite{Brini:2010fc}), 
because $R_{\nn}^{(g)}(\bm{x}_{I})$ does not have poles on $\mathcal{C}$ in $x_1$, as 
\begin{align}
F_{\nn}^{(g)}(\bm{x}_{I})&=
-
\mathop{\mathrm{Res}}_{x_0=x_1}\, \frac{dx_0\, dS_{z_0}(z_1)}{dx_1}\,
F_{\nn}^{(g)}(x_0,\bm{x}_{I \backslash \{1\}})
\nonumber\\
&=
\frac{1}{2\pi \mathrm{i}}\oint_{\mathcal{C}}
\frac{dx_0\, dS_{z_0}(z_1)}{dx_1}\,
F_{\nn}^{(g)}(x_0,\bm{x}_{I \backslash \{1\}})
\nonumber\\
&=
\frac{(-1)}{2\pi \mathrm{i}}\oint_{\mathcal{C}}
\frac{dx_0\, dS_{z_0}(z_1)}{2F_1^{(0)}(x_0)\, dx_1}
\Biggl[F_{\nn+1}^{(g-1)}(x_0,x_0,\bm{x}_{I \backslash \{1\}})
\nonumber\\
&\hspace{10.5em}
+\mathop{\sum_{g_1+g_2=g}}_{I_1 \cup I_2=I \backslash \{1\}}^{\textrm{no (0,1)}}
\mathcal{F}_{|I_1|+1}^{(g_1)}(x_0, \bm{x}_{I_1}) 
\mathcal{F}_{|I_2|+1}^{(g_2)}(x_0, \bm{x}_{I_2})\Biggr]\,,
\label{sd_cut}
\end{align}
where $z_i \in {\IP}^1$ are variables on the spectral curve, 
defined via the maps $x_i=x(z_i)$ introduced below \eqref{third_diff}, 
$dS_{z_0}(z_1)$ is the third kind differential \eqref{third_diff}, and
\begin{align}
\mathcal{F}_{|J|+1}^{(g)}(x_0, \bm{x}_{J}):=
F_{|J|+1}^{(g)}(x_0, \bm{x}_{J})+
\frac{\delta_{|J|,1}\, \delta_{g,0}}{2\left(x_0-x_j\right)^2}\,,\qquad 
j \in J=\{i_1, \ldots, i_{|J|}\}\,.
\end{align}
Solutions to the equation \eqref{sd_cut} satisfy $\mathcal{F}_{|J|+1}^{(g)}(x(\overline{z}), \bm{x}_{J})
=-\mathcal{F}_{|J|+1}^{(g)}(x(z), \bm{x}_{J})$, and we see that 
the integrand of the equation \eqref{sd_cut} has no branches. 
As a result, we get
\begin{align}
F_{\nn}^{(g)}(\bm{x}_{I})&=
\sum_{k=1}^{2h+2}
\mathop{\mathrm{Res}}_{x_0=\alpha_k}\,
\frac{(-1)\, dx_0\, dS_{z_0}(z_1)}{2F_1^{(0)}(x_0)\, dx_1}
\Biggl[F_{\nn+1}^{(g-1)}(x_0,x_0,\bm{x}_{I \backslash \{1\}})
\nonumber\\
&\hspace{12.5em}
+\mathop{\sum_{g_1+g_2=g}}_{I_1 \cup I_2=I \backslash \{1\}}^{\textrm{no (0,1)}}
\mathcal{F}_{|I_1|+1}^{(g_1)}(x_0, \bm{x}_{I_1}) 
\mathcal{F}_{|I_2|+1}^{(g_2)}(x_0, \bm{x}_{I_2})\Biggr]\,.
\label{top_rec_c}
\end{align}
Furthermore, in terms of multi-differentials of variables $z_i$ as
\begin{align}
\begin{split}
\omega_{2}^{(0)}(z_1, z_2)
&=B(z_1, z_2)
:=F_{2}^{(0)}(x_1, x_2)\, dx_1 dx_2 + \frac{dx_1 dx_2}{\left(x_1-x_2\right)^2}
\\
&=
\frac12 \partial_{x_2} dS_{z_2}(z_1)\, dx_2
+\frac{dx_1 dx_2}{2\left(x_1-x_2\right)^2}\,,
\\
\omega_{\nn}^{(g)}(z_1, \ldots, z_n)&=F_{\nn}^{(g)}(x_1, \ldots, x_n)\, dx_1 \cdots dx_{\nn}
\ \ \textrm{for}\ \ (g,\nn)\ne (0,2)\,,
\label{mdiffW}
\end{split}
\end{align}
the equation \eqref{top_rec_c} is written as the CEO topological recursion in \cite{Chekhov:2006vd,Eynard:2007kz} for the spectral curve \eqref{sp_curve_g}:%
\footnote{
Note the extra factor $2$ in the denominator of the recursion in the variables $z_i$, which arises from the double cover structure of the spectral curve.}
\begin{align}
\omega_{\nn}^{(g)}(\bm{z}_{I})&=
\sum_{k=1}^{2h+2}
\mathop{\mathrm{Res}}_{z_0=\xi_k}\,
K(z_0, z_1)
\Biggl[\omega_{\nn+1}^{(g-1)}(z_0,\overline{z}_0,\bm{z}_{I \backslash \{1\}})
\nonumber\\
&\hspace{10em}
+\mathop{\sum_{g_1+g_2=g}}_{I_1 \cup I_2=I \backslash \{1\}}^{\textrm{no (0,1)}}
\omega_{|I_1|+1}^{(g_1)}(z_0, \bm{z}_{I_1}) 
\omega_{|I_2|+1}^{(g_2)}(\overline{z}_0, \bm{z}_{I_2})\Biggr]\,,
\label{top_rec_gen}
\end{align}
where $\xi_k$ are branch points as $\alpha_k=x(\xi_k)$, and
\begin{align}
K(z_0, z_1):=
\frac{\int^{z=z_0}_{z=\overline{z}_0}B(z, z_1)}{4\omega_1^{(0)}(z_0)}
=\frac{\frac12 \int^{z=z_0}_{z=\overline{z}_0}B(z, z_1)}{\omega_1^{(0)}(z_0)-\omega_1^{(0)}(\overline{z}_0)}\,,
\end{align}
is the recursion kernel.

\begin{rem}\label{rem:kernel_expansion}
Solutions to the recursion \eqref{top_rec_c} are shown to be 
expressed in terms of kernel differentials \cite{Bouchard:2007ys},
\begin{align}
\chi_k^{(p)}(x):=&\,
\mathop{\mathrm{Res}}_{x_0=\alpha_k}
\left(\frac{dS_{z_0}(z)}{F_1^{(0)}(x_0)}\frac{dx_0}{(x_0-\alpha_k)^p}\right)
\nonumber\\
=&\,
\frac{dx}{(p-1)!\sqrt{\sigma(x)}}
\frac{\partial^{p-1}}{\partial x_0^{p-1}}\bigg|_{x_0=\alpha_k}
\frac{1}{M(x_0)}
\left(\frac{1}{x-x_0}-\sum_{i=1}^{h}C_i(x_0)\, L_i(x)\right),
\quad
p\ge 1\,,
\label{kernel_diff}
\end{align}
as
\begin{align}
F_{\nn}^{(g)}(x_1, \ldots, x_n)\, dx_1 \cdots dx_{\nn}=
\sum_{k=1}^{2h+2} \sum_{p_1,\ldots,p_{\nn} \ge 1}
C_{k;p_1,\ldots,p_{\nn}}^{(g)}\, 
\chi_{k}^{(p_1)}(x_1) \cdots \chi_{k}^{(p_{\nn})}(x_{\nn})\,,
\label{fg_kernel_exp}
\end{align}
for $2g+\nn \ge 3$, 
where the coefficients $C_{k;p_1,\ldots,p_{\nn}}^{(g)}$ do not depend on $x_1,\ldots,x_{\nn}$, and the sums contain only finitely many terms. 
From the expression \eqref{fg_kernel_exp} we see that the amplitudes are expanded around $x_i=\infty$ as \eqref{psdw},
\begin{align}
F_{\nn}^{(g)}(x_1, \ldots, x_n)=
\sum_{\ell_1,\ldots,\ell_{\nn} \ge 1}
x_1^{-\ell_1-1} \cdots x_{\nn}^{-\ell_{\nn}-1}\, f_{\ell_1,\ldots,\ell_{\nn}}^{(g)}\,,
\label{fg_x_exp}
\end{align}
for spectral curves with even number of branch points,
where the coefficients $f_{\ell_1,\ldots,\ell_{\nn}}^{(g)}$ do not depend on $x_1,\ldots,x_{\nn}$.
\end{rem}

\subsection{Examples}\label{subsec:ex_even}

We exemplify several $W^{(3)}$-type Hamiltonians, and their spectral curves obtained from the equation \eqref{sp_curve_determine}. 
Input parameters included in the Hamiltonian \eqref{gen_hamiltonian} are summarized in Remark \ref{rem:hamiltonian_param}, i.e., 
$\mu$ parameters $\kappa_k$ ($k=1, 2, \ldots, \mu$), 
$s+3$ parameters $\tau_k$ ($k=0, 1, \ldots, s+2$) and 
$r-\mu+1$ parameters $p_k$ ($k=\mu, \mu+1, \ldots, r$), 
and the Hamiltonian also contains constrained parameters $\kappa_{\mu+1}$ and $\kappa_{\mu+2}$ in \eqref{condition_coeff_cc_2}.

\subsubsection{Basic-type discrete DT (pure gravity) model}

As the first example, 
consider the one-cut spectral curve \eqref{discrete_DT_sim_sp_curve}:
\begin{align}
y=\frac{q(x)}{p(x)} \sqrt{\sigma(x)}\,,
\quad
p(x)=\kappa\, x^3\,,
\quad
q(x)= \frac{\kappa}{2} \left(x - \gamma \right),
\quad
\sigma(x)=x \left(x - \alpha\right),
\label{discrete_DT_sim_sp_curve_ex}
\end{align}
of the basic-type discrete DT model \cite{Watabiki:1993ym} reviewed in Section \ref{subsubsec:discrete_pureDT}.
In this case,
\begin{align}
r=\mu=3\,,\ \ 
s=1\,,
\end{align}
and we set
\begin{align}
\kappa_1=\kappa_3=\kappa\,,\ \ \kappa_2=0\,,\ \
\tau_0=0\,,\ \
\tau_1=\kappa\,,\ \ \tau_2=-\frac12\,,\ \ \tau_3=\frac{\kappa}{2}\,,\ \
p_3=\kappa\,.
\end{align}
Then, $\kappa_4=\kappa_5=0$ by \eqref{condition_coeff_cc_2}, and 
the Hamiltonian \eqref{gen_hamiltonian} yields \eqref{basic_h} as
\begin{align}
-\mathcal{H}_{\mathrm{basic}}&=
g_s^{-1} \sum_{k=1}^{3} k \, \kappa_{k}\, \Psi(k)
+2\sum_{k=1}^{3} \tau_k
\sum_{\ell \ge 1} \ell\, \Psd(\ell+k-2)\Psi(\ell)
\nonumber\\
&\ \
+ g_s \kappa \sum_{\ell, \ell' \ge 1} 
\left(
\left(\ell+\ell'-1\right)
\Psd(\ell)\Psd(\ell')\Psi(\ell+\ell'-1)
+ 
\ell \ell'\, \Psd(\ell+\ell'+1) \Psi(\ell) \Psi(\ell')
\right)
\nonumber\\
&=
g_s^{-1} \kappa\, \Psi(1) + 3g_s^{-1} \kappa\, \Psi(3)
+ 2\kappa \sum_{\ell \ge 1} \left(\ell+1\right) \Psd(\ell)\Psi(\ell+1)
\nonumber\\
&\ \
-\sum_{\ell \ge 1} \ell \left(\Psd(\ell)-\kappa \Psd(\ell+1)\right)\Psi(\ell)
\nonumber\\
&\ \ 
+ g_s \kappa \sum_{\ell,\ell' \ge 1} 
\left(
\left(\ell+\ell'-1\right) \Psd(\ell)\Psd(\ell')\Psi(\ell+\ell'-1)
+
\ell \ell'\, \Psd(\ell+\ell'+1)\Psi(\ell)\Psi(\ell')
\right),
\label{sim_hamiltonian}
\end{align}
and the spectral curve \eqref{discrete_DT_sim_sp_curve_ex} is actually obtained from the equation \eqref{sp_curve_determine}.

\subsubsection{\texorpdfstring{$(2,2m-1)$}{(2,2m-1)} minimal discrete DT model}

Next, let us provide a Hamiltonian for a one-cut spectral curve
\begin{align}
y=q(x) \sqrt{\sigma(x)}\,,
\quad
q(x)= \tau_{m+1}\, x^{m-1} + \sum_{k=0}^{m-2} q_k\, x^k\,,
\quad
\sigma(x)=\left(x - \alpha_1\right)\left(x - \alpha_2\right).
\label{minimal_DT_sp_curve}
\end{align}
In this case, 
\begin{align}
r=\mu=0\,,\ \ 
s=m-1\ (m \ge 1)\,,
\end{align}
and we set
\begin{align}
\tau_{0}=1\,,\ \ \tau_1=0\,,\ \ \tau_2=-1\,,\ \
p_0=1\,.
\end{align}
Then, $\kappa_1=0$, $\kappa_2=1$ by \eqref{condition_coeff_cc_2}, and 
the Hamiltonian \eqref{gen_hamiltonian} yields
\begin{align}
-\mathcal{H}_{\mathrm{MDDT}}^{(m)}&=
2g_s^{-1}\, \Psi(2)
+ \Psi(1)^2
+2\sum_{k=0}^{m+1} \tau_k
\sum_{\ell \ge 1} \ell\, \Psd(\ell+k-2)\Psi(\ell)
\nonumber\\
&\ \
+g_s \sum_{\ell, \ell' \ge 1} 
\left(
\left(\ell+\ell'+2\right)
\Psd(\ell)\Psd(\ell')\Psi(\ell+\ell'+2)
+
\ell \ell'\, \Psd(\ell+\ell'-2) \Psi(\ell) \Psi(\ell')
\right)
\nonumber\\
&=
g_s \sum_{\ell, \ell' \ge 0} 
\left(
\left(\ell+\ell'+2\right)
\widehat{\Psi}^{\dagger}(\ell) \widehat{\Psi}^{\dagger}(\ell') \Psi(\ell+\ell'+2)
+
\ell \ell'\, 
\widehat{\Psi}^{\dagger}(\ell+\ell'-2) \widehat{\Psi}(\ell) \widehat{\Psi}(\ell')
\right)
+ \mathcal{H}_0\,,
\label{mddt_hamiltonian}
\end{align}
which describes 
the $(2,2m-1)$ minimal discrete DT model \cite{Watabiki:1993ym} (see also footnote \ref{ft:dt} in Section \ref{subsubsec:discrete_pureDT}). 
Here
\begin{align}
\begin{split}
&
\widehat{\Psi}^{\dagger}(\ell)=\Psd(\ell)+g_s^{-1}\delta_{\ell,0}\ \
\textrm{for}\ \ \ell \ge 0\,,
\\
&
\widehat{\Psi}(\ell)=\Psi(\ell) + g_s^{-1}\frac{\tau_{\ell}}{\ell}\ \
\textrm{for}\ \ \ell \ge 1\,,\quad
\tau_{\ell}=0\ \ \textrm{for}\ \ \ell \ge m+2\,,
\end{split}
\end{align}
are introduced, and $\mathcal{H}_0$ is explicitly expressed in terms of the string creation operators $\Psd(\ell)$.
The spectral curve \eqref{minimal_DT_sp_curve} is obtained from the equation \eqref{sp_curve_determine}, and 
$q_k$ ($0, \ldots, m-2$), $\alpha_1$ and $\alpha_2$ are determined as functions of 
the parameters $\tau_{\ell}$ ($\ell=3, \ldots, m+1$).

\begin{rem}\label{rem:mddt_matrix}
For the minimal discrete DT model, the Laplace-transformed string operator $g_s^{-1}x^{-1}+ \Psdw(x)$ in \eqref{psdw} with the shift $g_s^{-1}x^{-1}$ is identified with the generating function of traces of rank $N$ matrix $M$:
\begin{align}
\omega(x;M):=\mathrm{Tr}\frac{1}{x-M}
=N\, x^{-1}+\sum_{\ell \ge 1} x^{-\ell-1}\, \mathrm{Tr} M^{\ell}\,,
\label{resolvent_op}
\end{align}
in the hermitian matrix model with a polynomial potential,
\begin{align}
Z_{\mathrm{MM}}=\int dM\, \e^{-\frac{2}{g_s} \mathrm{Tr} V(M)}\,,\qquad
V(M)=
\frac{1}{2} M^{2} - \sum_{\ell=3}^{m+1} \frac{\tau_{\ell}}{\ell} M^{\ell}\,.
\end{align}
This follows from the fact that 
the separated SD equation \eqref{key_sp_sd} for amplitudes $f_{\nn}(\bm{x}_{I})$ agrees with the SD equation (loop equation)
for $\nn$-point resolvents in the matrix model (see e.g., \cite{Eynard:2004mh}):
\begin{align}
\left\langle \prod_{i=1}^{\nn} \omega(x_i;M) \right\rangle_{\mathrm{MM}}
=
\frac{1}{Z_{\mathrm{MM}}}\,
\int dM\, \e^{-\frac{2}{g_s} \mathrm{Tr} V(M)}
\prod_{i=1}^{\nn} \omega(x_i;M)\,.
\end{align}
\end{rem}

\subsubsection{Penner model}

We here provide a Hamiltonian for the spectral curve \cite{Manabe:2015kbj} of the Penner model \cite{Penner1988} (see Remark \ref{rem:penner} in the following):
\begin{align}
y=\frac{1}{p(x)} \sqrt{\sigma(x)}\,,
\quad
p(x)=1-x\,,\ \
\sigma(x)=x^2 + 4\mu x-4\mu\,.
\label{penner_sp_curve}
\end{align}
In this case,
\begin{align}
r=1\,,\ \ 
\mu=0\,,\ \ 
s=0\,,
\end{align}
and we set
\begin{align}
p_0=1\,,\ \ p_1=-1\,.
\end{align}
Then, $\kappa_1=2\tau_{0}\tau_1+\tau_{0}^2$, $\kappa_2=\tau_{0}^2$ by \eqref{condition_coeff_cc_2}, and 
the Hamiltonian \eqref{gen_hamiltonian} yields
\begin{align}
-\mathcal{H}_{\mathrm{penner}}&=
g_s^{-1} \sum_{k=1,2} k \, \kappa_{k}\, \Psi(k)
+
\tau_{0} \Psi(1)^2
+2\sum_{k=0,1,2} \tau_k \sum_{\ell \ge 1}
\ell\, \Psd(\ell+k-2)\Psi(\ell)
\nonumber\\
&\ \
+g_s \sum_{k=0,1} p_k
\sum_{\ell, \ell' \ge 1}
\left(\ell+\ell'-k+2\right)
\Psd(\ell)\Psd(\ell')\Psi(\ell+\ell'-k+2)
\nonumber\\
&\ \
+g_s \sum_{k=0,1} p_k \sum_{\ell, \ell' \ge 1}
\ell \ell'\, \Psd(\ell+\ell'+k-2) \Psi(\ell) \Psi(\ell')\,.
\label{penner_hamiltonian}
\end{align}
The equation \eqref{sp_curve_determine} is
\begin{align}
q(x)^2 \sigma(x)=\tau_2^2 x^2 + 2 \tau_1 \tau_2 x
+ \left(\tau_1 + \tau_{0}\right)^2 + 2 \tau_{0} \tau_2\,,
\end{align}
and by setting
\begin{align}
\tau_{0}=-2\mu\,,\ \ \tau_1=2\mu\,,\ \ \tau_2=1\,,
\end{align}
we obtain the spectral curve \eqref{penner_sp_curve}.

\begin{rem}\label{rem:penner}
The Penner model is a matrix model with potential given by
\begin{align}
V(x)=-x - \log\left(1-x\right).
\end{align}
The partition function of the Penner model gives the generating function of the virtual Euler characteristics of moduli spaces of curves of genus $g$ with $n$ marked points \cite{Penner1988}:
\begin{align}
\chi(\mathcal{M}_{g,n})=
\frac{(-1)^n (2g+n-3)! (2g-1)}{(2g)! n!}\, B_{2g}\,,
\end{align}
where $B_{2g}$ are Bernoulli numbers defined by
\begin{align}
\frac{x}{\e^x-1}=\sum_{k \ge 0} \frac{B_k}{k!}\, x^k
=-\frac12\, x + \sum_{g \ge 0} \frac{B_{2g}}{(2g)!}\, x^{2g}\,.
\end{align}
We remark that the parameter $\mu$ in the spectral curve \eqref{penner_sp_curve} corresponds to the 't Hooft parameter and serves as the fugacity for the number of marked points.
\end{rem}

\subsubsection{4D \texorpdfstring{$\mathcal{N}=2$}{N=2} \texorpdfstring{$SU(2)$}{SU(2)} gauge theory with \texorpdfstring{$N_f=4$}{Nf=4}}
\label{subsubsec:nf4_sw}

Let us provide a Hamiltonian for the two-cut spectral curve in \cite{Gaiotto:2009we,Eguchi:2009gf,Kozcaz:2010af}, which is the Seiberg-Witten curve in 4D $\mathcal{N}=2$ $SU(2)$ gauge theory with $N_f=4$ hypermultiplets, as
\begin{align}
y=\frac{m_0}{p(x)} \sqrt{\sigma(x)}\,,
\quad
p(x)=x \left(x-1\right)\left(x-\zeta\right),
\quad
\sigma(x)=x^4 + S_1 x^3 + S_2 x^2 + S_3 x + S_4\,,
\label{su2_nf4_sw}
\end{align}
where
\begin{align}
\begin{split}
S_1&=-\frac{(\zeta-1)m_2^2 +2\zeta m_2m_3+(1+\zeta)(m_0^2+U)}{m_0^2}\,,
\qquad
S_4=\frac{\zeta^2m_1^2}{m_0^2}\,,
\\
S_2&=\frac{\zeta (m_0^2+m_1^2-m_3^2+2m_2m_3) + (\zeta-1)\zeta m_2^2+ \zeta^2 (2m_2+m_3)m_3+ (1+\zeta)^2U}{m_0^2}\,,
\\
S_3&=-\frac{\zeta (m_1^2-m_3^2)+ \zeta^2 (m_1^2+2m_2m_3+m_3^2)+ \zeta(1+\zeta)U}{m_0^2}\,.
\end{split}
\end{align}
Here $\zeta$ is the UV coupling parameter, $m_0, m_1, m_2, m_3$ are the masses of the hypermultiplets, and $U$ is the (quantum) Coulomb branch parameter.
In this case, 
\begin{align}
r=3\,,\ \ 
\mu=1\,,\ \ 
s=1\,,
\end{align}
and we set
\begin{align}
\tau_3=m_0\,,\ \
p_1=\zeta\,,\ \ p_2=-1-\zeta\,,\ \ p_3=1\,.
\end{align}
Then, $\kappa_2=2\tau_{0}\tau_{1}/\zeta+\tau_{0}^2(1+\zeta)/\zeta^2$, 
$\kappa_3=\tau_{0}^2/\zeta$ by \eqref{condition_coeff_cc_2}, 
and the Hamiltonian \eqref{gen_hamiltonian} yields
\begin{align}
-\mathcal{H}^{SU(2)}_{N_f=4}&=
g_s^{-1}
\sum_{k=1}^{3} k \, \kappa_{k}\, \Psi(k)
+
\tau_{0} \Psi(1)^2
+2\sum_{k=0}^{3} \tau_k
\sum_{\ell \ge 1}
\ell\, \Psd(\ell+k-2)\Psi(\ell)
\nonumber\\
&\ \
+g_s \sum_{k=1}^{3} p_k
\sum_{\ell, \ell' \ge 1}
\left(\ell+\ell'-k+2\right)
\Psd(\ell)\Psd(\ell')\Psi(\ell+\ell'-k+2)
\nonumber\\
&\ \
+g_s \sum_{k=1}^{3} p_k
\sum_{\ell, \ell' \ge 1}
\ell \ell'\, 
\Psd(\ell+\ell'+k-2) \Psi(\ell) \Psi(\ell')\,.
\label{sw_su2_nf4_hamiltonian}
\end{align}
The equation \eqref{sp_curve_determine} is
\begin{align}
q(x)^2 \sigma(x)=&\
\left(\tau_0 x^{-1} + \tau_1 + \tau_2 x + m_0 x^2\right)^2
+ 2 m_0\, f\, x \left(x-1\right)\left(x-\zeta\right)
\nonumber\\
&
- x \left(x-1\right)\left(x-\zeta\right)
\left(\kappa_1 x^{-1} + \left(\frac{2\tau_{0}\tau_{1}}{\zeta}+\frac{\tau_{0}^2(1+\zeta)}{\zeta^2}\right)x^{-2}
+\frac{\tau_{0}^2}{\zeta} x^{-3}\right),
\nonumber\\
=&\
m_0^2\, x^4 + 2m_0 \left(f + \tau_2\right) x^3
+
\left(-2 m_{0} f \left(1+\zeta\right) + 2 \tau_{1} m_{0}+\tau_{2}^{2}-\kappa_1 \right) x^{2}
\nonumber\\
&
+ \frac{\left(2 \zeta^{3} m_{0} f+\kappa_1  \,\zeta^{3}+2 \zeta^{2} \tau_{0} m_{0}+2 \tau_{1} \tau_{2} \zeta^{2}+\kappa_1  \,\zeta^{2}-\tau_{0}^{2} \zeta -2 \tau_{0} \tau_{1} \zeta -\tau_{0}^{2}\right) x}{\zeta^{2}}
\nonumber\\
&
-\frac{\kappa_1  \,\zeta^{3}-\zeta^{2} \tau_{0}^{2}-2 \tau_{0} \tau_{1} \zeta^{2}-2 \tau_{0} \tau_{2} \zeta^{2}-\tau_{1}^{2} \zeta^{2}-\tau_{0}^{2} \zeta -2 \tau_{0} \tau_{1} \zeta -\tau_{0}^{2}}{\zeta^{2}}\,,
\label{sp_eq_sw_nf4}
\\
f:=&\
\lim_{g_s \to 0} g_s \lim_{T \to \infty} \left\langle\mathrm{vac}\Big| \e^{-T \mathcal{H}_{N_f=4}^{SU(2)}}\,
\Psd(1) \Big|\mathrm{vac}\right\rangle.
\end{align}
Here we change the parameters 
$\kappa_1, \tau_0, \tau_1, \tau_2, \tau_3(=m_0)$ to $m_0, m_1, m_2, m_3$ by%
\footnote{\label{ft:sw_red}
One of the parameters in $\kappa_1, \tau_0, \tau_1, \tau_2, \tau_3$ is redundant with respect to the parameters $m_0, m_1, m_2, m_3$, and the choices \eqref{param_rel_sw_su2_4} are not unique for obtaining the Seiberg-Witten curve \eqref{su2_nf4_sw} (see also Remark \ref{rem:hamiltonian_param} in Section \ref{subsec:w3_h}).}
\begin{align}
\begin{split}
\kappa_1&=
\left(m_{0}+m_{1}+m_{2}+m_{3}\right)
\left(m_{0}-m_{1}+m_{2}+m_{3}\right)\zeta\,,
\\
\tau_{0}&=0\,,
\ \
\tau_{1}=\left(m_{0}+m_{2}+m_{3}\right)\zeta\,,
\ \
\tau_{2}=
-\left(m_{0}+m_{2}\right)
-\left(m_{0}+m_{3}\right)\zeta\,.
\label{param_rel_sw_su2_4}
\end{split}
\end{align}
Then, by the equation \eqref{sp_eq_sw_nf4}, we obtain the spectral curve \eqref{su2_nf4_sw}, where 
the Coulomb branch parameter $U$ is introduced by
\begin{align}
2 m_0\, f=
-\left(\zeta +1\right) U
+\left(m_0+m_2\right)^2
+\left(m_{0}-m_{2}\right) 
\left(m_{0} + m_{2} + 2 m_{3}\right)\zeta\,,
\label{su2_nf4_coulomb}
\end{align}
and determined by the $A$-period of the curve:
\begin{align}
a = \oint_{\mathcal{A}} y\, dx\,.
\end{align}

\begin{rem}\label{rem:nf4_sw}
The amplitudes $F_{\nn}^{(g)}(\bm{x}_{I})$, determined by the Hamiltonian \eqref{sw_su2_nf4_hamiltonian}, provide the perturbative expansion coefficients of the instanton partition function with $\nn$ half-BPS simple-type surface defects, in the self-dual $\Omega$-background with parameter $g_s=\hbar$ 
\cite{Alday:2009fs,Dimofte:2010tz,Kozcaz:2010af}.
\end{rem}

\section{Hamiltonians for spectral curves with odd number of branch points}
\label{sec:ham_ai}

In this section, instead of the spectral curve \eqref{sp_curve_g} we focus on 
the following class of spectral curves with $(h+1)$-cut and odd number of branch points as
\begin{align}
y=M(x) \sqrt{\sigma(x)}\,,
\quad M(x)=\frac{q(x)}{p(x)}\,,
\quad \sigma(x)=\prod_{k=1}^{2h+1} (x - \alpha_k)\,,
\label{sp_curve_g_c}
\end{align}
where $p(x)$ and $q(x)$ are polynomials of $x$ defined in \eqref{pq_def}.
We will associate a two-reduced $W^{(3)}$-type Hamiltonian with the spectral curve \eqref{sp_curve_g_c}.

\subsection{Two-reduced \texorpdfstring{$W^{(3)}$}{W(3)}-type Hamiltonian and amplitudes}

The spectral curve \eqref{cont_DT_sp_curve} of the continuum pure DT model is an example of the above class of spectral curves with $r=\mu=0$, $s=1$ and $h=0$.
Similar to Section \ref{sec:ham_tr}, by consulting the Hamiltonian \eqref{cont_pure_dt_ham}, we associate the following two-reduced $W^{(3)}$-type Hamiltonian $\mathcal{H}^{red}$ with the spectral curve \eqref{sp_curve_g_c} (see Appendix \ref{app:rev_sd_ai} for a detailed construction of $\mathcal{H}^{red}$). 
A remarkable point, compared to the $W^{(3)}$-type Hamiltonian $\mathcal{H}$ in \eqref{gen_hamiltonian}, 
is the appearance of a three-string annihilation operator $\Psi(1)^2 \Psi(2)$.

\begin{Def}\label{def:hamiltonian_red}
A two-reduced $W^{(3)}$-type Hamiltonian $\mathcal{H}^{red}$ (see Figs. \ref{fig:hamiltonian} and \ref{fig:hamiltonian_extra}) is
\begin{align}
\begin{split}
-\mathcal{H}^{red}&=
2 g_s^{-1}
\sum_{k=1}^{\mu+1} k \, \kappa_{k}\, \Psi(2k)
+
\frac{g_s}{8} \left(2p_0 \Psi(4) + p_1 \Psi(2)\right)
+ \tau_{0}\, \Psi(1) \Psi(2)
\\
&\ \
+ \frac{g_s p_0}{8} \Psi(1)^2 \Psi(2)
+2\sum_{k=0}^{s+1} \tau_k
\sum_{\ell \ge 1}
\ell\, \Psd(\ell+2k-3)\Psi(\ell)
\\
&\ \
+g_s \sum_{k=\mu}^{r} p_k
\sum_{\ell, \ell' \ge 1}
\left(\ell+\ell'-2k+4\right)
\Psd(\ell)\Psd(\ell')\Psi(\ell+\ell'-2k+4)
\\
&\ \
+\frac{g_s}{4} \sum_{k=\mu}^{r} p_k
\sum_{\ell, \ell' \ge 1}
\ell \ell'\, 
\Psd(\ell+\ell'+2k-4) \Psi(\ell) \Psi(\ell')\,,
\label{gen_hamiltonian_ai}
\end{split}
\end{align}
where $g_s$ is the string coupling constant.
\begin{figure}[t]
\centering
\includegraphics[width=65mm]{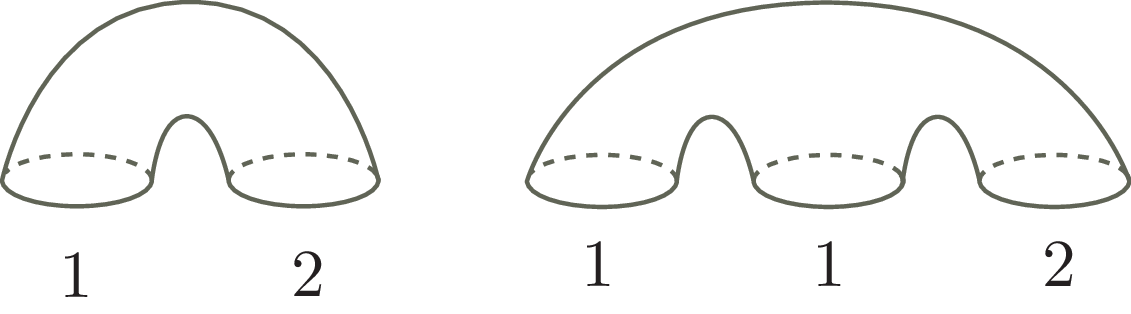}
\caption{Graphical representation of the term $\Psi(1) \Psi(2)$ in $\mathcal{H}^{red}$ instead of $\Psi(1)^2$ in Fig. \ref{fig:hamiltonian} and 
the additional term $\Psi(1)^2 \Psi(2)$ in $\mathcal{H}^{red}$ which is not included among the building blocks in Fig. \ref{fig:hamiltonian}.}
\label{fig:hamiltonian_extra}
\end{figure}
The parameters $\kappa_k$ ($k=1, 2, \ldots, \mu$), 
$\tau_k$ ($k=0, 1, \ldots, s+1$) and 
$p_k$ ($k=0,1, \mu, \mu+1, \ldots, r$, $p_{\mu} \ne 0$, $p_r \ne 0$) are independent parameters, 
whereas $\kappa_{\mu+1}$ is the constrained parameter determined by the planar SD equation in Section \ref{subsec:n1_sd_red} (see also \eqref{nu0_ai_app}) as
\begin{align}
\kappa_{\mu+1}=\frac{\tau_0^2}{p_{\mu}}\,.
\label{nu0_ai}
\end{align}
Here we impose a condition
\begin{align}
s+1 \ge r\,,
\label{condition_param_gen_ai}
\end{align}
for obtaining the spectral curve \eqref{sp_curve_g_c} from the planar SD equation in Section \ref{subsec:n1_sd_red} (see also \eqref{condition_param_c}).
\end{Def}

\begin{rem}\label{rem:hamiltonian_param_ai}
The input data for the two-reduced $W^{(3)}$-type 
Hamiltonian \eqref{gen_hamiltonian_ai} admit
a similar interpretation to that given in Remark \ref{rem:hamiltonian_param} for 
the $W^{(3)}$-type Hamiltonian \eqref{gen_hamiltonian}.
The analysis of the planar SD equation in Section \ref{subsec:n1_sd_red} will clarify the constrained parameter \eqref{nu0_ai} and the condition \eqref{condition_param_gen_ai}.
Here we note that the Hamiltonian \eqref{gen_hamiltonian_ai} contains $\mu$ redundant parameters in relation to the determination of the spectral curve \eqref{sp_curve_g_c}, just as mentioned in Remark \ref{rem:hamiltonian_param}.
\end{rem}

Here, we note that the Hamiltonian \eqref{gen_hamiltonian_ai} satisfies the no big-bang condition $\mathcal{H}^{red} |\mathrm{vac}\rangle=0$ given in \eqref{no_bigbang_red}. 
For the $\nn$-point amplitude $\sff_{\nn}(\bm{x}_{I})$ defined in \eqref{n_pt_amp_v} and determined by $\mathcal{H}^{red}$, we obtain the following proposition.

\begin{prop}\label{prop:sd_eq_red}
The SD equation \eqref{sd_eq_gen_red} yields
\begin{align}
0&=
g_s \sum_{i=1}^{\nn} \partial_{x_i} \left(
\left(p(x_i)\, 
\sff_{\nn+1}(x_i, \bm{x}_{I}\right)_{\mathrm{irreg}(x_i)}
\right)
+ \sum_{i=1}^{\nn} \partial_{x_i} \left( 
\left(2 \sfDel(x_i)\,  \sff_{\nn}(\bm{x}_{I})\right)_{\mathrm{irreg}(x_i)}\right)
\nonumber\\
&\ \
+ \sum_{i=1}^{\nn} \partial_{x_i} \left(
\left(
g_s^{-1} \sfK_{-}(x_i)
+\frac{g_s p_0}{16} x_i^{-2}+\frac{g_s p_1}{16} x_i^{-1}
\right)
\sff_{\nn-1}(\bm{x}_{I \backslash \{i\}})\right)
\nonumber\\
&\ \
+ g_s \sum_{1 \le i < j \le \nn} \partial_{x_i}\partial_{x_j}
\biggl[\biggl(\left(x_i^{-1/2} p(x_i)\, \sff_{\nn-1}(\bm{x}_{I \backslash \{j\}})\right)_{\mathrm{irreg}(x_i)}
\nonumber\\
&\hspace{10em}
-\left(x_j^{-1/2} p(x_j)\, \sff_{\nn-1}(\bm{x}_{I \backslash \{i\}})\right)_{\mathrm{irreg}(x_j)}\biggr)
/\left(x_i^{1/2}-x_j^{1/2}\right)\biggr]
\nonumber\\
&\ \
+ \sum_{1 \le i < j \le \nn} \partial_{x_i}\partial_{x_j}
\frac{\left(x_i^{-1/2}\sfDel(x_i)\right)_{\mathrm{irreg}(x_i)}
-\left(x_j^{-1/2}\sfDel(x_j)\right)_{\mathrm{irreg}(x_j)}}
{x_i^{1/2}-x_j^{1/2}}
\, \sff_{\nn-2}(\bm{x}_{I \backslash \{i,j\}})
\nonumber\\
&\ \
-\frac{g_s p_0}{8} 
\sum_{i=1}^{\nn} \sum_{\substack{1 \le j <k \le \nn \\ j,k \ne i}}
x_i^{-2}x_j^{-3/2}x_k^{-3/2}\,
\sff_{\nn-3}(\bm{x}_{I \backslash \{i,j,k\}})\,,
\label{key_sd_red}
\end{align}
where $f(x)_{\mathrm{irreg}(x)}:=F(x)-F(x)_{\mathrm{reg}(x)}$ for a function $F(x)$ of $x$, and 
$\mathrm{reg}(x)$ denotes the regular part of the Laurent series of $x^{1/2}$ at $x=0$.
Here $p(x)$ is defined in \eqref{pq_def}, and 
\begin{align}
\sfDel(x):=\sum_{k=-1}^{s}\tau_{k+1}\, x^{k+1/2}\,,
\qquad
\sfK_{-}(x):=\sum_{k=1}^{\mu+1} \kappa_k\, x^{-k}\,.
\label{delta_k_def_ai}
\end{align}
\end{prop}
\begin{proof}
By
\begin{align}
&
\left[-\mathcal{H}^{red}, \Psd(\ell)\right]
\nonumber\\
&=
g_s\, \ell \sum_{k=\mu}^{r} p_k\, 
\sum_{\substack{\ell',\ell''\ge 1 \\ \ell'+\ell''=\ell+2k-4}}
\Psd(\ell')\Psd(\ell'')
+ 2\ell \sum_{k=-1}^{s}\tau_{k+1}\, \Psd(\ell+2k-1)
+ g_s^{-1}\, \ell \sum_{k=1}^{\mu+1} \kappa_{k} \delta_{\ell,2k}
\nonumber\\
&\ \
+ \frac{g_s}{8} \left(2p_0 \delta_{\ell,4} + p_1 \delta_{\ell,2}\right)
+\frac{g_s}{2} \sum_{k=\mu}^{r} p_k
\sum_{\ell' \ge 1} \ell \ell'\,
\Psd(\ell+\ell'+2k-4) \Psi(\ell')
\nonumber\\
&\ \
+\frac12 \tau_{0}\, \ell \left(3-\ell\right) \Psi(3-\ell)
+\frac{g_s p_0}{8} \left(2\delta_{\ell,1} \Psi(1)\Psi(2)
+ \delta_{\ell,2} \Psi(1)^2\right),
\label{com_H1_red}
\end{align}
we obtain
\begin{align}
&
\Psd(\ell_1) \cdots \Psd(\ell_{i-1})
\left[-\mathcal{H}^{red}, \Psd(\ell_i)\right]
\Psd(\ell_{i+1})\cdots \Psd(\ell_{\nn})
\nonumber
\\
&=
\Psd(\ell_1)\cdots \breve{\Psi}^{\dagger}(\ell_i)\cdots \Psd(\ell_{\nn})
\nonumber\\
&\ \
\times
\Biggl(
g_s\, \ell_i \sum_{k=\mu}^{r} p_k\, 
\sum_{\substack{\ell,\ell'\ge 1 \\ \ell+\ell'=\ell_i+2k-4}}
\Psd(\ell)\Psd(\ell')
+ 2 \ell_i \sum_{k=-1}^{s}\tau_{k+1}\, \Psd(\ell_i+2k-1)
\nonumber\\
&\hspace{2.5em}
+ g_s^{-1}\, \ell_i \sum_{k=1}^{\mu+1} \kappa_{k} \delta_{\ell_i,2k}
+\frac{g_s}{8} \left(2p_0 \delta_{\ell_i,4} + p_1 \delta_{\ell_i,2}\right)
\nonumber\\
&\hspace{2.5em}
+ \frac{g_s}{2} \sum_{k=\mu}^{r} p_k
\sum_{\ell \ge 1} \ell_i \ell\,
\Psd(\ell_i+\ell+2k-4)\Psi(\ell)
\nonumber\\
&\hspace{2.5em}
+\frac12 \tau_{0}\, \ell_i \left(3-\ell_i\right) \Psi(3-\ell_i)
+\frac{g_s p_0}{8} \left(2\delta_{\ell_i,1} \Psi(1)\Psi(2)
+ \delta_{\ell_i,2} \Psi(1)^2\right)
\Biggr)
\nonumber\\
&\ \
+ \sum_{j=i+1}^{\nn} 
\Psd(\ell_1)\cdots \breve{\Psi}^{\dagger}(\ell_i)\cdots 
\breve{\Psi}^{\dagger}(\ell_j) \cdots \Psd(\ell_{\nn})
\Biggl(
\frac{g_s}{2} \sum_{k=\mu}^{r} p_k \ell_i \ell_j\,
\Psd(\ell_i+\ell_j+2k-4)
\nonumber\\
&\hspace{10em}
+\frac12 \tau_{0}\, \delta_{\ell_i+\ell_j,3}\, \ell_i\, \ell_j
+\frac{g_s p_0}{4}
\left(\delta_{\ell_i+\ell_j,2}\Psi(2) + \delta_{\ell_i+\ell_j,3}\Psi(1)\right)
\Biggr)
\nonumber\\
&\ \
+\frac{g_s p_0}{4} \sum_{i+1 \le j<k \le \nn} 
\delta_{\ell_i+\ell_j+\ell_k,4}\, 
\Psd(\ell_1)\cdots \breve{\Psi}^{\dagger}(\ell_i)\cdots 
\breve{\Psi}^{\dagger}(\ell_j)\cdots 
\breve{\Psi}^{\dagger}(\ell_k) \cdots \Psd(\ell_{\nn})\,,
\end{align}
and the SD equation \eqref{sd_eq_gen_red} yields \eqref{key_sd_red}.
\end{proof}

From behaviors for $x_1 \to \infty$:
\begin{align}
\begin{split}
&
p(x_1) \sff_{\nn+1}(x_1,x_1, \bm{x}_{I \backslash \{1\}})=O(x_1^{r-3})\,,\ \
\sfDel(x_1) \sff_{\nn}(\bm{x}_{I})=O(x_1^{s-1})\,,
\\
&
p(x_1)\sff_{\nn-1}(\bm{x}_{I \backslash \{j\}})=O(x_1^{r-3/2})\ \textrm{for}\ j\ne 1\,,
\ \
\sfDel(x_1) = O(x_1^{s+1/2})\,,
\end{split}
\end{align}
and an equation
\begin{align}
\frac{x_1^{-1/2} f(x_1)-x_i^{-1/2} f(x_i)}{x_1^{1/2}-x_i^{1/2}}
= \frac{\left(1+x_1^{-1/2}x_i^{1/2}\right) f(x_1)
-\left(x_1^{1/2}x_i^{-1/2}+1\right) f(x_i)}{x_1-x_i}\,,
\end{align}
for a function $f(x)$ of $x$, 
Proposition \ref{prop:sd_eq_red} implies the following proposition.

\begin{prop}\label{prop:sd_eq_sep_red}
Symmetric solutions $\sff_{\nn}(\bm{x}_{I})$ in the variables $\bm{x}_{I}$ of the separated SD equation
\begin{align}
0&=
g_s \left(p(x_1)\, \sff_{\nn+1}(x_1, \bm{x}_{I})\right)_{\mathrm{irreg}(x_1)}
+
\left(2 \sfDel(x_1)\, \sff_{\nn}(\bm{x}_{I})\right)_{\mathrm{irreg}(x_1)}
\nonumber\\
&\ \
+ \left(g_s^{-1} \sfK_{-}(x_1)
+\frac{g_s p_0}{16} x_1^{-2}+\frac{g_s p_1}{16} x_1^{-1}\right)
\sff_{\nn-1}(\bm{x}_{I \backslash \{1\}})
\nonumber\\
&\ \
+g_s \sum_{i=2}^{\nn} \partial_{x_i}
\frac{
\left(x_1^{-1/2}x_i^{1/2} p(x_1)\, \sff_{\nn-1}(\bm{x}_{I \backslash \{i\}})\right)_{\mathrm{irreg}(x_1)}
-\left(p(x_i)\, \sff_{\nn-1}(\bm{x}_{I \backslash \{1\}})\right)_{\mathrm{irreg}(x_i)}}{x_1-x_i}
\nonumber\\
&\ \
+ \sum_{i=2}^{\nn} \partial_{x_i}
\frac{\left(x_1^{-1/2}x_i^{1/2}\sfDel(x_1)\right)_{\mathrm{irreg}(x_1)}
-\sfDel(x_i)_{\mathrm{irreg}(x_i)}}{x_1-x_i}
\, \sff_{\nn-2}(\bm{x}_{I \backslash \{1,i\}})
\nonumber\\
&\ \
+ \frac{g_s p_0}{8} \sum_{2 \le i <j \le \nn} x_1^{-1} x_i^{-3/2} x_j^{-3/2}
\sff_{\nn-3}(\bm{x}_{I \backslash \{1,i,j\}})
+ C_{\nn-1}(\bm{x}_{I \backslash \{1\}})\,,
\label{key_sp_sd_red}
\end{align}
are also solutions to the equation \eqref{key_sd_red}, 
where $C_{\nn-1}(\bm{x}_{I \backslash \{1\}})$ is 
a function of $\bm{x}_{I \backslash \{1\}}$.
\end{prop}

\subsection{Topological recursion}

\begin{Def}\label{def:free_en_red}
(Connected) amplitudes $\sfF_{\nn}(\bm{x}_{I})$, $\nn \ge 1$ are%
\footnote{
The shift term $\Omega_2(x_1,x_2)$ in the annulus amplitude motivates the inclusion of the three-string annihilation operator $\Psi(1)^2 \Psi(2)$ in the two-reduced $W^{(3)}$-type Hamiltonian $\mathcal{H}^{red}$ in \eqref{gen_hamiltonian_ai}. This is shown in Section \ref{subsubsec:n3_sp_sd_odd}.
In the pure DT model reviewed in Section \ref{subsec:dt_rev}, this shift naturally appears in the continuum limit (see e.g., \cite{FMW2025a} for details).}
\begin{align}
\begin{split}
&
\sfF_1(x) = \sff_1(x) + g_s^{-1} \frac{\sfDel(x)}{p(x)}\,,
\\
&
\sfF_2(x_1,x_2) = \sff^{\mathrm{con}}_2(x_1,x_2) + \Omega_2(x_1,x_2)\,,
\quad
\Omega_2(x_1,x_2):=
\frac{1}{4x_1^{1/2}x_2^{1/2}\left(x_1^{1/2}+x_2^{1/2}\right)^2}\,,
\\
&
\sfF_{\nn}(\bm{x}_{I}) =\sff_{\nn}^{\mathrm{con}}(\bm{x}_{I})
\ \textrm{for}\ \nn \ge 2\,,
\label{free_en_red}
\end{split}
\end{align}
where $\sff_{\nn}^{\mathrm{con}}(\bm{x}_{I})$ is the connected part of 
$\sff_{\nn}(\bm{x}_{I})$. 
We consider the perturbative expansion of the connected amplitudes as
\begin{align}
\sff_{\nn}^{\mathrm{con}}(\bm{x}_{I})
=\sum_{g \ge 0} g_s^{2g-2+\nn} \sff_{\nn}^{(g)}(\bm{x}_{I})\,,
\quad
\sfF_{\nn}(\bm{x}_{I})=\sum_{g \ge 0} g_s^{2g-2+\nn} \sfF_{\nn}^{(g)}(\bm{x}_{I})\,.
\label{pert_exp_f_red}
\end{align}
\end{Def}

\subsubsection{\texorpdfstring{$\nn=1$}{n=1} separated SD equation and disk amplitude}\label{subsec:n1_sd_red}

When $\nn=1$, the separated SD equation \eqref{key_sp_sd_red} is
\begin{align}
0&=
\left(
g_s p(x)\, \sff_{2}(x, x)
+ 2 \sfDel(x)\, \sff_{1}(x)\right)_{\mathrm{irreg}(x)}
+ g_s^{-1} \sfK_{-}(x)
+\frac{g_s}{16} \left(p_0 x^{-2}+p_1 x^{-1}\right)
+ C_{0}
\nonumber\\
&=
\left(
g_s p(x)\, \sfF_{2}(x, x)
+ g_s p(x)\, \sff_{1}(x)^2
+ 2 \sfDel(x)\, \sff_{1}(x)\right)_{\mathrm{irreg}(x)}
+ g_s^{-1} \sfK_{-}(x)\,,
\label{sp_sd_red_1}
\end{align}
where $C_{0}=0$ is used, following from the asymptotic behavior as $x \to \infty$.
The form of this equation is same as the equation \eqref{sd_n1} except 
the form of $\sfDel(x)$, and we obtain
\begin{align}
0&=
\left(\sff_1^{(0)}(x) + \frac{\sfDel(x)}{p(x)}\right)^2
- \frac{q(x)^2}{p(x)^2}\, \sigma(x)
=
\sfF_1^{(0)}(x)^2 
- \frac{q(x)^2}{p(x)^2}\, \sigma(x)\,,
\label{sd_pert_n10_odd}
\\
0&=
\left(p(x)\, \sfF_{2}^{(g-1)}(x,x)
+p(x) \sum_{g_1+g_2=g} \sfF_1^{(g_1)}(x)\, \sfF_1^{(g_2)}(x)
\right)_{\mathrm{irreg}(x)}
\quad \textrm{for}\ g \ge 1\,,
\label{sd_pert_n1h_red}
\end{align}
where
\begin{align}
q(x)^2 \sigma(x)=
\sfDel(x)^2 
+p(x)\left(2 \sfDel(x)\,  \sff_{1}^{(0)}(x)
+p(x)\, \sff_1^{(0)}(x)^2\right)_{\mathrm{reg}(x)} 
- p(x)\, \sfK_{-}(x)\,.
\label{sp_curve_determine_red}
\end{align}
As a result, 
$(h+1)$-cut solution ($0 \le h \le s$) of the disk amplitude
\begin{align}
\sfF_1^{(0)}(x)=M(x) \sqrt{\sigma(x)}\,,\ \
M(x)=\frac{q(x)}{p(x)}=
\frac{\sum_{k=0}^{s-h} q_k\, x^k}{\sum_{k=\mu}^{r} p_k\, x^{k}}\,,\ \
\sigma(x)=\prod_{k=1}^{2h+1} (x - \alpha_k)\,,
\label{planar_f_g_red}
\end{align}
is obtained from the equation \eqref{sp_curve_determine_red}, 
which yields $2s+3$ conditions 
on the coefficients of powers of $x$, by
\begin{align*}
&
\sfDel(x)^2 = \left(\sum_{k=-1}^{s}\tau_{k+1}\, x^{k+1/2}\right)^2
= \tau_{s+1}^2\, x^{2s+1} + \cdots
+ \tau_0^2\, x^{-1}\,,
\\
&
p(x)\left(2 \sfDel(x)\,  \sff_{1}^{(0)}(x)
+p(x)\, \sff_1^{(0)}(x)^2\right)_{\mathrm{reg}(x)} = 
\left(\sum_{k=\mu}^{r} p_k\, x^{k}\right)
\left(\sum_{k=0}^{s-1} c_k\, x^k\right)
\\
&\hspace{17.4em}
=p_r c_{s-1} x^{r+s-1} + \cdots,
\\
&
p(x) \sfK_{-}(x) = 
\left(\sum_{k=\mu}^{r} p_k\, x^{k}\right)
\left(\sum_{k=1}^{\mu+1} \kappa_k\, x^{-k}\right)
= p_r \kappa_1 x^{r-1} + \cdots 
+ p_{\mu}\kappa_{\mu+1}\, x^{-1}\,,
\end{align*}
where $c_k$ $(k=0,1, \ldots, s-1)$ are certain constants. 
They determine the parameter $\kappa_{\mu+1}$ as in \eqref{nu0_ai}, 
the $s$ constants $c_k$, and $s+2$ parameters $q_k$ and $\alpha_k$ out of the total $s+h+2$ parameters. 
Here $2s+1 > r+s-1$, i.e., the condition $s+1 \ge r$ in \eqref{condition_param_gen_ai} is imposed to determine the spectral curve $y=\sfF_1^{(0)}(x)$, and in particular, 
\begin{align}
q_{s-h}=\tau_{s+1}\,,
\end{align}
is determined. 
The remaining $h$ parameters are determined by the $h$ $A$-periods in \eqref{periods_c}.

\subsubsection{\texorpdfstring{$\nn=2$}{n=2} separated SD equation and annulus amplitude}

When $\nn=2$, the separated SD equation \eqref{key_sp_sd_red} is
\begin{align}
0&=
\left(
g_s p(x_1)\, \sff_{3}(x_1, x_1, x_2)
+
2 \sfDel(x_1)\, \sff_{2}(x_1, x_2)\right)_{\mathrm{irreg}(x_1)}
\nonumber\\
&\ \
+ 
\left(g_s^{-1} \sfK_{-}(x_1)
+ \frac{g_s}{16}\left(p_0 x_1^{-2}+p_1 x_1^{-1}\right)\right) \sff_{1}(x_2)
\nonumber\\
&\ \
+g_s \partial_{x_2}
\frac{
\left(x_1^{-1/2}x_2^{1/2} p(x_1)\, \sff_{1}(x_1)\right)_{\mathrm{irreg}(x_1)}
-\left(p(x_2)\, \sff_{1}(x_2)\right)_{\mathrm{irreg}(x_2)}}{x_1-x_2}
\nonumber\\
&\ \
+ \partial_{x_2}
\frac{\left(x_1^{-1/2}x_2^{1/2}\sfDel(x_1)\right)_{\mathrm{irreg}(x_1)}
-\sfDel(x_2)_{\mathrm{irreg}(x_2)}}{x_1-x_2}
+ C_{1}(x_2)
\nonumber\\
&=
\bigl(
g_s p(x_1)\left(\sff_{3}^{\mathrm{con}}(x_1, x_1, x_2)
+ 2\sff_{2}^{\mathrm{con}}(x_1, x_2) \sff(x_1)
+ \sff_{2}(x_1, x_1) \sff_1(x_2)\right)
\nonumber\\
&\ \ \ \
+ 2 \sfDel(x_1)\left(\sff_{2}^{\mathrm{con}}(x_1,x_2)+\sff_1(x_1)\sff_1(x_2)\right)
\bigr)_{\mathrm{irreg}(x_1)}
\nonumber\\
&\ \
+ 
\left(g_s^{-1} \sfK_{-}(x_1)
+\frac{g_s}{16}\left(p_0 x_1^{-2}+p_1 x_1^{-1}\right)\right) \sff_{1}(x_2)
\nonumber\\
&\ \
+g_s \partial_{x_2}
\frac{
\left(x_1^{-1/2}x_2^{1/2} p(x_1)\, \sfF_{1}(x_1)\right)_{\mathrm{irreg}(x_1)}
-\left(p(x_2)\, \sfF_{1}(x_2)\right)_{\mathrm{irreg}(x_2)}}{x_1-x_2}
+ C_{1}(x_2)\,,
\label{sd_n2_red}
\end{align}
and, by the $\nn=1$ separated SD equation \eqref{sp_sd_red_1} we obtain
\begin{align}
0&=
g_s
\left(
p(x_1)\left(\sfF_{3}(x_1, x_1, x_2)
+ 2\sff_{2}^{\mathrm{con}}(x_1, x_2)
\sfF_1(x_1)\right)\right)_{\mathrm{irreg}(x_1)}
\nonumber\\
&\ \
+g_s \partial_{x_2}
\frac{
\left(p(x_1)\, \sfF_{1}(x_1)\right)_{\mathrm{irreg}(x_1)}
-\left(p(x_2)\, \sfF_{1}(x_2)\right)_{\mathrm{irreg}(x_2)}}{x_1-x_2}
\nonumber\\
&\ \
+g_s \partial_{x_2}
\frac{
\left(\left(x_1^{-1/2}x_2^{1/2}-1\right) p(x_1)\, \sfF_{1}(x_1)
\right)_{\mathrm{irreg}(x_1)}}{x_1-x_2}
+ C_{1}(x_2)\,.
\label{sd_n2_1_red}
\end{align}
By
\begin{align}
\partial_{x_2} \frac{x_1^{-1/2}x_2^{1/2} - 1}{x_1-x_2}
=\frac{1}{2x_1^{1/2}x_2^{1/2}\left(x_1^{1/2}+x_2^{1/2}\right)^2}
=2 \Omega_2(x_1, x_2)\,,
\label{omega_id}
\end{align}
where $\Omega_2(x_1, x_2)$ is 
defined in \eqref{pert_exp_f_red}, the equation \eqref{sd_n2_1_red} yields
\begin{align}
0&=
g_s
p(x_1)\left(\sfF_{3}(x_1, x_1, x_2)
+ 2\sfF_{2}(x_1, x_2) \sfF_1(x_1)\right)
+g_s \partial_{x_2}
\frac{p(x_1)\, \sfF_{1}(x_1)-p(x_2)\, \sfF_{1}(x_2)}{x_1-x_2}
\nonumber\\
&\ \
- g_s
\left(
p(x_1)\left(\sfF_{3}(x_1, x_1, x_2)
+ 2\sff_{2}^{\mathrm{con}}(x_1, x_2) \sfF_1(x_1)\right)\right)_{\mathrm{reg}(x_1)}
\nonumber\\
&\ \
-g_s \partial_{x_2}
\frac{
\left(x_1^{-1/2}x_2^{1/2} p(x_1)\, \sfF_{1}(x_1)\right)_{\mathrm{reg}(x_1)}
-\left(p(x_2)\, \sfF_{1}(x_2)\right)_{\mathrm{reg}(x_2)}}{x_1-x_2}
+ C_{1}(x_2)\,.
\label{sd_n2_2_red}
\end{align}
This equation is almost the same form as the equation \eqref{sd_n2_1}, and actually we can show that the annulus amplitude of the same form as \eqref{annulus_sol}, such that
\begin{itemize}
\item
the amplitude $\sfF_{2}^{(0)}(x_1,x_2)$ has no poles in $x_1$ away from the branch cuts of the spectral curve, 
which lie along $[\alpha_1, \alpha_2] \cup \cdots \cup 
[\alpha_{2h+1}, \infty]$,
\end{itemize}
is obtained.

\subsubsection{\texorpdfstring{$\nn=3$}{n=3} separated SD equation}
\label{subsubsec:n3_sp_sd_odd}

When $\nn=3$, the separated SD equation \eqref{key_sp_sd_red} is
\begin{align}
0&=
\left(
g_s p(x_1)\, \sff_{4}(x_1, x_1, x_2, x_3)
+
2 \sfDel(x_1)\, \sff_{3}(x_1, x_2, x_3)\right)_{\mathrm{irreg}(x_1)}
\nonumber\\
&\ \
+ 
\left(g_s^{-1} \sfK_{-}(x_1)
+ \frac{g_s}{16}\left(p_0 x_1^{-2} + p_1 x_1^{-1}\right)\right)
\sff_{2}(x_2,x_3)
\nonumber\\
&\ \
+g_s \sum_{i=2,3} \partial_{x_i}
\frac{
\left(x_1^{-1/2}x_i^{1/2} p(x_1)\, \sff_{2}(\bm{x}_{I \backslash \{i\}})\right)_{\mathrm{irreg}(x_1)}
-\left(p(x_i)\, \sff_{2}(x_2,x_3)\right)_{\mathrm{irreg}(x_i)}}{x_1-x_i}
\nonumber\\
&\ \
+ \sum_{i=2,3} \partial_{x_i}
\frac{\left(x_1^{-1/2}x_i^{1/2}\sfDel(x_1)\right)_{\mathrm{irreg}(x_1)}
-\sfDel(x_i)_{\mathrm{irreg}(x_i)}}{x_1-x_i}
\, \sff_{1}(\bm{x}_{I \backslash \{1,i\}})
\nonumber\\
&\ \
+ \frac{g_s p_0}{8} x_1^{-1} x_2^{-3/2} x_3^{-3/2}
+ C_{2}(x_2,x_3)\,.
\label{sd_n3_red}
\end{align}
Using the separated SD equations \eqref{sp_sd_red_1} for 
$\nn=1$ and \eqref{sd_n2_red} for $\nn=2$, we obtain
\begin{align}
0&=
\bigl(
g_s p(x_1)\left(\sff_{4}^{\mathrm{con}}(x_1, x_1, x_2, x_3)
+ 2\sff_{3}^{\mathrm{con}}(x_1, x_2, x_3) \sff_1(x_1)
+ 2\sff_{2}^{\mathrm{con}}(x_1, x_2) \sff_2^{\mathrm{con}}(x_1, x_3)\right)
\nonumber\\
&\ \ \ \
+ 2 \sfDel(x_1) \sff_{3}^{\mathrm{con}}(x_1, x_2, x_3)
\bigr)_{\mathrm{irreg}(x_1)}
\nonumber\\
&\ \
+g_s \sum_{(i,j)=(2,3), (3,2)} \partial_{x_i}
\frac{
\left(x_1^{-1/2}x_i^{1/2} p(x_1)\, \sff_{2}^{\mathrm{con}}(x_1,x_j)\right)_{\mathrm{irreg}(x_1)}
-\left(p(x_i)\, \sff_{2}^{\mathrm{con}}(x_2,x_3)\right)_{\mathrm{irreg}(x_i)}}{x_1-x_i}
\nonumber\\
&\ \
+ \frac{g_s p_0}{8} x_1^{-1} x_2^{-3/2} x_3^{-3/2}
+ \widetilde{C}_{2}(x_2,x_3)
\nonumber\\
&=
g_s
\left(p(x_1)\left(\sfF_{4}(x_1, x_1, x_2, x_3)
+ 2\sfF_{3}(x_1, x_2, x_3) \sfF_1(x_1)
+ 2\sff_{2}^{\mathrm{con}}(x_1, x_2) \sff_2^{\mathrm{con}}(x_1, x_3)\right)
\right)_{\mathrm{irreg}(x_1)}
\nonumber\\
&\ \
+g_s \sum_{(i,j)=(2,3), (3,2)} \partial_{x_i}
\frac{
\left(p(x_1)\, \sff_{2}^{\mathrm{con}}(x_1,x_j)\right)_{\mathrm{irreg}(x_1)}
-\left(p(x_i)\, \sff_{2}^{\mathrm{con}}(x_2,x_3)
\right)_{\mathrm{irreg}(x_i)}}{x_1-x_i}
\nonumber\\
&\ \
+g_s \sum_{(i,j)=(2,3), (3,2)} \partial_{x_i}
\frac{
\left(\left(x_1^{-1/2}x_i^{1/2}-1\right) p(x_1)\, \sff_{2}^{\mathrm{con}}(x_1,x_j)\right)_{\mathrm{irreg}(x_1)}}{x_1-x_i}
\nonumber\\
&\ \
+ \frac{g_s p_0}{8} x_1^{-1} x_2^{-3/2} x_3^{-3/2}
+ \widetilde{C}_{2}(x_2,x_3)\,,
\label{sd_n3_1_red}
\end{align}
where $\widetilde{C}_{2}(x_2,x_3)$ is a function of $x_2$ and $x_3$.
By \eqref{omega_id}, and%
\footnote{
Equation \eqref{omega_id_2} can be shown by using \eqref{omega_id} as
\begin{align*}
&
2x_1^k \Omega_2(x_1,x_2) \Omega_2(x_1,x_3)
+ \sum_{(i,j)=(2,3), (3,2)} \partial_{x_i}
\frac{x_1^k \Omega_2(x_1,x_j)-x_i^k\Omega_2(x_2,x_3)}{x_1-x_i}
\\
&=\frac12 \partial_{x_2}\partial_{x_3}
\Biggl(\frac{x_1^{k-1}}{(x_1^{1/2}+x_2^{1/2})(x_1^{1/2}+x_3^{1/2})}
\\
&\hspace{6em}
+ \sum_{(i,j)=(2,3), (3,2)}
\frac{1}{x_1-x_i}
\left(\frac{-x_1^k}{x_1^{1/2}(x_1^{1/2}+x_j^{1/2})}
+\frac{x_i^k}{x_i^{1/2}(x_i^{1/2}+x_j^{1/2})}\right)\Biggr)\,,
\end{align*}
and simplifying this equation.
}
\begin{align}
&
\frac{p_0}{8} x_1^{-1} x_2^{-3/2} x_3^{-3/2}
+
\sum_{k=3}^{r} \frac{p_k}{8}
\sum_{\ell=0}^{k-3}
\sum_{\substack{\ell_2, \ell_3 \ge 0 \\ \ell_2+\ell_3=2\ell}}
(-1)^{\ell_2}(\ell_2+1)(\ell_3+1) x_1^{k-3-\ell}
x_2^{(\ell_2-1)/2} x_3^{(\ell_3-1)/2}
\nonumber\\
&=
2p(x_1) \Omega_2(x_1,x_2) \Omega_2(x_1,x_3)
+ \sum_{(i,j)=(2,3), (3,2)} \partial_{x_i}
\frac{p(x_1)\Omega_2(x_1,x_j)-p(x_i)\Omega_2(x_2,x_3)}{x_1-x_i}\,,
\label{omega_id_2}
\end{align}
the equation \eqref{sd_n3_1_red} yields
\begin{align}
0&=
g_s p(x_1)\left(\sfF_{4}(x_1, x_1, x_2, x_3)
+ 2\sfF_{3}(x_1, x_2, x_3) \sfF_1(x_1)
+ 2\sfF_{2}(x_1, x_2) \sfF_2(x_1, x_3)\right)
\nonumber\\
&\ \
+g_s \sum_{(i,j)=(2,3), (3,2)} \partial_{x_i}
\frac{
p(x_1)\, \sfF_{2}(x_1,x_j)-p(x_i)\, \sfF_{2}(x_2,x_3)
}{x_1-x_i}
\nonumber\\
&\ \
-g_s
\left(p(x_1)\left(\sfF_{4}(x_1, x_1, x_2, x_3)
+ 2\sfF_{3}(x_1, x_2, x_3) \sfF_1(x_1)
+ 2\sff_{2}^{\mathrm{con}}(x_1, x_2) \sff_2^{\mathrm{con}}(x_1, x_3)\right)
\right)_{\mathrm{reg}(x_1)}
\nonumber\\
&\ \
-g_s \sum_{(i,j)=(2,3), (3,2)} \partial_{x_i}
\frac{
\left(x_1^{-1/2}x_i^{1/2} p(x_1)\, \sff_{2}^{\mathrm{con}}(x_1,x_j)\right)_{\mathrm{reg}(x_1)}
-\left(p(x_i)\, \sff_{2}^{\mathrm{con}}(x_2,x_3)\right)_{\mathrm{reg}(x_i)}}
{x_1-x_i}
\nonumber\\
&\ \
- g_s \sum_{k=3}^{r} \frac{p_k}{8}
\sum_{\ell=0}^{k-3}
\sum_{\substack{\ell_2, \ell_3 \ge 0 \\ \ell_2+\ell_3=2\ell}}
(-1)^{\ell_2}(\ell_2+1)(\ell_3+1) x_1^{k-3-\ell}
x_2^{(\ell_2-1)/2} x_3^{(\ell_3-1)/2}
+ \widetilde{C}_{2}(x_2,x_3)\,.
\label{sd_n3_2_red}
\end{align}

\subsubsection{Separated SD equation for general \texorpdfstring{$\nn$}{n} and topological recursion}

For $\nn \ge 4$, the connected part of 
the separated SD equation \eqref{key_sp_sd_red} yields
\begin{align}
0&=
g_s p(x_1)\left(\sfF_{\nn+1}(x_1,\bm{x}_{I})
+\sum_{I_1 \cup I_2=I \backslash \{1\}}
\sfF_{|I_1|+1}(x_1, \bm{x}_{I_1}) \sfF_{|I_2|+1}(x_1, \bm{x}_{I_2})\right)
\nonumber\\
&\ \
+g_s \sum_{i=2}^{\nn} \partial_{x_i} 
\frac{p(x_1)\, \sfF_{\nn-1}(\bm{x}_{I \backslash \{i\}})
-p(x_i)\, \sfF_{\nn-1}(\bm{x}_{I \backslash \{1\}})}{x_1-x_i}
\nonumber\\
&\ \
-g_s \left(
p(x_1)\, \sfF_{\nn+1}(x_1,\bm{x}_{I})
+p(x_1) \sum_{I_1 \cup I_2=I \backslash \{1\}}
\widetilde{\sfF}_{|I_1|+1}(x_1, \bm{x}_{I_1}) 
\widetilde{\sfF}_{|I_2|+1}(x_1, \bm{x}_{I_2})
\right)_{\mathrm{reg}(x_1)}
\nonumber\\
&\ \
-g_s \sum_{i=2}^{\nn} \partial_{x_i} 
\frac{\left(x_1^{-1/2} x_i^{1/2} p(x_1)\, \sfF_{\nn-1}(\bm{x}_{I \backslash \{i\}})
\right)_{\mathrm{reg}(x_1)}
-\left(p(x_i)\, \sfF_{\nn-1}(\bm{x}_{I \backslash \{1\}})\right)_{\mathrm{reg}(x_i)}}{x_1-x_i}
\nonumber\\
&\ \
+\widetilde{C}_{\nn-1}(\bm{x}_{I \backslash \{1\}})\,,
\label{sp_sd_con_gen_red}
\end{align}
where 
$\widetilde{C}_{\nn-1}(\bm{x}_{I \backslash \{1\}})$ is a function of 
$\bm{x}_{I \backslash \{1\}}$, and 
$\widetilde{\sfF}_2=\sff_{2}^{\mathrm{con}}$ and 
$\widetilde{\sfF}_{\nn}=\sfF_{\nn}$ for $\nn \ne 2$.
By applying the perturbative expansion \eqref{pert_exp_f_red}, 
the SD equations \eqref{sp_sd_red_1} for $\nn=1$,
\eqref{sd_n2_2_red} for $\nn=2$, 
\eqref{sd_n3_2_red} for $\nn=3$, and 
\eqref{sp_sd_con_gen_red} for $\nn \ge 4$ yield 
perturbative SD equations for $g \ge 0$, $\nn \ge 1$ with $2g-2+\nn \ge 0$ as
\begin{align}
\sfF_{\nn}^{(g)}(\bm{x}_{I})&=
\frac{(-1)}{2\sfF_1^{(0)}(x_1)}
\Biggl[\sfF_{\nn+1}^{(g-1)}(x_1,\bm{x}_{I})
+\mathop{\sum_{g_1+g_2=g}}_{I_1 \cup I_2=I \backslash \{1\}}^{\textrm{no (0,1)}}
\sfF_{|I_1|+1}^{(g_1)}(x_1, \bm{x}_{I_1}) \sfF_{|I_2|+1}^{(g_2)}(x_1, \bm{x}_{I_2})
\nonumber\\
&\hspace{14em}
+\sum_{i=2}^{\nn} 
\frac{\sfF_{\nn-1}^{(g)}(\bm{x}_{I \backslash \{i\}})}{(x_1-x_i)^2}\Biggr]
+\frac{R_{\nn}^{(g)}(\bm{x}_{I})}{2 p(x_1) \sfF_1^{(0)}(x_1)}\,.
\label{sp_sd_con_gen_pert_red}
\end{align}
Here
\begin{align}
&
R_{\nn}^{(g)}(\bm{x}_{I})
\nonumber
\\
&=
\Biggl(p(x_1)\, \sfF_{\nn+1}^{(g-1)}(x_1,\bm{x}_{I})
+p(x_1) \mathop{\sum_{g_1+g_2=g}}_{I_1 \cup I_2=I \backslash \{1\}}
\widetilde{\sfF}_{|I_1|+1}^{(g_1)}(x_1, \bm{x}_{I_1}) 
\widetilde{\sfF}_{|I_2|+1}^{(g_2)}(x_1, \bm{x}_{I_2})
\Biggr)_{\mathrm{reg}(x_1)}
\nonumber
\\
& \ \
+\sum_{i=2}^{\nn} \partial_{x_i} 
\frac{\left(x_1^{-1/2} x_i^{1/2} p(x_1)\, \widetilde{\sfF}_{\nn-1}^{(g)}(\bm{x}_{I \backslash \{i\}})\right)_{\mathrm{reg}(x_1)}+
\left(p(x_i)\, \widetilde{\sfF}_{\nn-1}^{(g)}(\bm{x}_{I \backslash \{1\}})\right)_{\mathrm{irreg}(x_i)}}{x_1-x_i}
\nonumber\\
&\ \
+\delta_{g,0}\delta_{\nn,3}
\sum_{k=3}^{r} \frac{p_k}{8}
\sum_{\ell=0}^{k-3}
\sum_{\substack{\ell_2, \ell_3 \ge 0 \\ \ell_2+\ell_3=2\ell}}
(-1)^{\ell_2}(\ell_2+1)(\ell_3+1) x_1^{k-3-\ell}
x_2^{(\ell_2-1)/2} x_3^{(\ell_3-1)/2}
+C_{\nn-1}^{(g)}(\bm{x}_{I \backslash \{1\}})\,,
\end{align}
where $C_{\nn-1}^{(g)}(\bm{x}_{I \backslash \{1\}})$ is a function of $\bm{x}_{I \backslash \{1\}}$, and
$\widetilde{\sfF}_2^{(0)}=\sff_2^{(0)}$ and 
$\widetilde{\sfF}_{\nn}^{(g)}=\sfF_{\nn}^{(g)}$ for $(g,\nn) \ne (0,2)$.

Consider the $(h+1)$-cut spectral curve 
$y={\sfF}_1^{(0)}(x)=q(x) \sqrt{\sigma(x)}/p(x)$ in \eqref{planar_f_g_red} 
as input of the SD equation \eqref{sp_sd_con_gen_pert_red}. 
By exactly the same argument as Section \ref{subsubsec:tr_w}, by assuming that
\begin{itemize}
\item
the amplitude ${\sfF}_{\nn}^{(g)}(\bm{x}_{I})$ on the left hand side of \eqref{sp_sd_con_gen_pert_red} has no poles in $x_1$ away from the branch cuts 
of the spectral curve,
which lie along $\mathcal{C}:=[\alpha_1, \alpha_2] \cup \cdots \cup 
[\alpha_{2h+1}, \infty]$,
\end{itemize}
the amplitudes are obtained as solutions to the CEO topological recursion as in \eqref{top_rec_c} or \eqref{top_rec_gen}, where the residues are, instead, taken at the branch points $\alpha_k$ ($k=1, 2, \ldots, 2h+1$).

\begin{rem}\label{rem:kernel_expansion_ai}
As mentioned in Remark \ref{rem:kernel_expansion}, the solutions ${\sfF}_{\nn}^{(g)}(\bm{x}_{I})$ obtained via the topological recursion are 
expressed in terms of kernel differentials \eqref{kernel_diff}. 
They are expanded around $x_i=\infty$ as in \eqref{psdw_v},
\begin{align}
{\sfF}_{\nn}^{(g)}(\bm{x}_{I})=
\sum_{\ell_1,\ldots,\ell_{\nn} = 1,3,5,\ldots}
x_1^{-\ell_1/2-1} \cdots x_{\nn}^{-\ell_{\nn}/2-1}\, \sff_{\ell_1,\ldots,\ell_{\nn}}^{(g)}\,,
\label{fg_x_exp_ai}
\end{align}
for spectral curves with odd number of branch points,
where the coefficients $\sff_{\ell_1,\ldots,\ell_{\nn}}^{(g)}$ do not depend on $x_1,\ldots,x_{\nn}$.
Here, note that the sums are restricted to odd positive integers.
\end{rem}

\subsection{Examples}\label{subsec:ex_odd}

We exemplify several two-reduced $W^{(3)}$-type Hamiltonians, and their spectral curves obtained from the equation \eqref{sp_curve_determine_red}.
The Hamiltonian \eqref{gen_hamiltonian_ai} contains 
$\mu$ parameters $\kappa_k$ ($k=1, 2, \ldots, \mu$), 
$s+2$ parameters $\tau_k$ ($k=0, 1, \ldots, s+1$) and 
$r-\mu+1$ parameters $p_k$ ($k=\mu, \mu+1, \ldots, r$) (as well as $p_0, p_1$) as input parameters, 
and $\kappa_{\mu+1}=\tau_0^2/p_{\mu}$ in \eqref{nu0_ai} as a constrained parameter.

\subsubsection{\texorpdfstring{$(2,2m-1)$}{(2,2m-1)} minimal continuum DT model}\label{subsubsec:mcdt}

Consider a one-cut spectral curve
\begin{align}
y=q(x) \sqrt{\sigma(x)}\,,
\quad
q(x)= \tau_m\, x^{m-1} + \sum_{k=0}^{m-2} q_k\, x^k\,,
\quad
\sigma(x)=x - \alpha\,.
\label{minimal_c_DT_sp_curve}
\end{align}
In this case, 
\begin{align}
r=\mu=0\,,\ \ 
s=m-1\ (m \ge 1)\,,
\end{align}
and we set
\begin{align}
p_0=1\,,\ \ 
p_1=0\,.
\end{align}
Then, the parameter $\kappa_1$ is determined by \eqref{nu0_ai} as $\kappa_1=\tau_0^2$, and the Hamiltonian \eqref{gen_hamiltonian_ai} yields
\begin{align}
-\mathcal{H}_{\mathrm{MDT}}^{(m)}&=
\frac{g_s}{4} \Psi(4)
+ \frac{g_s}{8} 
\left(\Psi(1)+4g_s^{-1}\tau_0\right)^2\Psi(2)
+2\sum_{k=0}^{m} \tau_k 
\sum_{\ell \ge 1} \ell\, \Psd(\ell+2k-3)\Psi(\ell)
\nonumber\\
&\ \
+ g_s \sum_{\ell, \ell' \ge 1}
\left(\ell+\ell'+4\right) \Psd(\ell)\Psd(\ell')\Psi(\ell+\ell'+4)
+\frac{g_s}{4}
\sum_{\ell, \ell' \ge 1}
\ell \ell'\, \Psd(\ell+\ell'-4) \Psi(\ell) \Psi(\ell')
\nonumber\\
&=
\frac{g_s}{4} \Psi(4)
+ \frac{g_s}{8} \widehat{\Psi}(1)^2 \Psi(2)
+ g_s \sum_{\ell, \ell' \ge 1}
\left(\ell+\ell'+4\right) \Psd(\ell)\Psd(\ell')\Psi(\ell+\ell'+4)
\nonumber\\
&\ \
+\frac{g_s}{4}
\sum_{\ell, \ell' \ge 1}
\ell \ell'\, \Psd(\ell+\ell'-4) \widehat{\Psi}(\ell) \widehat{\Psi}(\ell')
+ \mathcal{H}_0\,,
\label{cont_dt_hamiltonian}
\end{align}
where $\widehat{\Psi}(2\ell)=\Psi(2\ell)$ and
\begin{align}
\widehat{\Psi}(2\ell+1)
=\Psi(2\ell+1) + g_s^{-1}\frac{4\tau_{\ell}}{2\ell+1}\,,\quad
\tau_{\ell}=0\ \textrm{for}\ \ell \ge m+1\,,
\end{align}
are introduced, and $\mathcal{H}_0$ is explicitly expressed in terms of the string creation operators $\Psd(\ell)$.
This Hamiltonian describes 
the $(2,2m-1)$ minimal continuum DT ($m$-th multicritical DT) model 
\cite{Gubser:1993vx,Watabiki:1993ym,Ambjorn:1996ne} (see also footnote \ref{ft:cont_dt} in Section \ref{subsubsec:conti_pureDT}), 
and the spectral curve \eqref{minimal_c_DT_sp_curve} is obtained from the equation \eqref{sp_curve_determine_red}, 
where $q_k$ ($0, \ldots, m-2$) and $\alpha$ are determined as functions of 
the parameters $\tau_{\ell}$ ($\ell=0,1, \ldots, m$).

Here, it is intriguing to consider a specialization of the parameters $\tau_k$ as
\begin{align}
\begin{split}
&
\tau_m=1\,,
\quad
\tau_{m-2\ell+1}=0\,,
\quad
\ell=1, \ldots, \left\lfloor\frac{m+1}{2}\right\rfloor,
\\
&
\tau_{m-2k}=
\frac{\left(2m-1\right)(2m-2k-3)!!}{k!\, (2m-4k-1)!!}
\left(-\frac{\mu}{8}\right)^k,
\quad
k=1, \ldots, \left\lfloor\frac{m}{2}\right\rfloor,
\label{minimal_c_DT_sp}
\end{split}
\end{align}
which leads to the so-called ``conformal background'', with the cosmological constant $\mu$, studied in \cite{Moore:1991ir,Gubser:1993vx}. 
Then it is shown in \cite{FMW2025b} that the polynomial $q(x)$ and the parameter $\alpha$ in \eqref{minimal_c_DT_sp_curve} yield
\begin{align}
\begin{split}
&
\alpha=-\sqrt{\mu}\left(1- \delta_{m,1}\right),
\\
&
q(x)=\frac{(-1)^{m-1}\mu^{(2m-1)/4}}{2^{m-3/2}\sqrt{x+\sqrt{\mu}}}\,
\sin\left(\frac{2m-1}{2}\arccos\left(-\frac{x}{\sqrt{\mu}}\right)\right)
\\
&\hspace{1.85em}=
\frac{\mu^{(2m-1)/4}}{2^{m-3/2}\sqrt{x+\sqrt{\mu}}}\,
T_{2m-1}\biggl(\frac{\sqrt{x+\sqrt{\mu}}}{\sqrt{2}\, \mu^{1/4}}\biggr)\,,
\label{minimal_c_DT_sp_curve_FZZT}
\end{split}
\end{align}
where $T_{n}(x)$ is the Chebyshev polynomial of the first kind defined by $T_n(\cos \theta)=\cos(n\theta)$.
This spectral curve arises from the leading density of eigenvalues in the matrix model dual to the $(2,2m-1)$ minimal string in the semi-classical expansion (see e.g., \cite{Saad:2019lba,Mertens:2020hbs}). 
We also consider an alternative specialization of the parameters $\tau_k$ as
\begin{align}
\tau_0=0\,,
\quad
\tau_k=
\frac{(-4\pi^2)^{k-1}}{(2k-1)!}
\prod_{i=1}^{k-1} \left(1-\left(\frac{2i-1}{2m-1}\right)^2\right),\quad
k=1, \ldots, m\,.
\end{align}
We see that this just changes $x \to x - \sqrt{\mu}$ and the overall normalization of the polynomial $q(x)$ in \eqref{minimal_c_DT_sp_curve_FZZT}, 
and specializes $\sqrt{\mu}=(2m-1)^2/(2^3 \pi^2)$ as \cite{Seiberg:2003nm} (see also \cite{Gregori:2021tvs,Fuji:2023wcx}), 
\begin{align}
\begin{split}
&
\alpha=0\,,
\\
&
q(x)
=\frac{(-1)^{m-1}}{2\pi \sqrt{x}}\,
T_{2m-1}\biggl(\frac{2\pi\sqrt{x}}{2m-1}\biggr)
=\sum_{k=0}^{m-1}\frac{(-4\pi^2 x)^k}{(2k+1)!}
\prod_{i=1}^k \left(1-\left(\frac{2i-1}{2m-1}\right)^2\right).
\label{minimal_c_DT_sp_curve_chebyshev}
\end{split}
\end{align}
As shown in \cite{Saad:2019lba}, 
in the limit $m\to\infty$, the polynomial $q(x)$ in \eqref{minimal_c_DT_sp_curve_chebyshev} yields $q(x)=\frac{1}{2\pi\sqrt{x}}\sin(2\pi\sqrt{x})$ 
which gives the spectral curve arising from the Jackiw-Teitelboim (JT) gravity 
\cite{Stanford:2017thb}.
In the context of hyperbolic geometry, 
the CEO topological recursion for this spectral curve is shown to be the Laplace dual of Mirzakhani's recursion \cite{Mirzakhani:2006fta} for the Weil-Petersson volumes of moduli spaces of hyperbolic bordered Riemann surfaces \cite{Eynard:2007fi}.

\begin{rem}\label{rem:minimal_grav}
When $m=1$ (resp. $m=2$), the $(2,1)$ (resp. $(2,3)$) minimal continuum DT model describes the topological gravity (resp. pure gravity). 
Actually, when $m=1$, by setting 
$\tau_0=0, \tau_1=1$, we obtain the spectral curve 
$y=\sqrt{x}$, given by \eqref{minimal_c_DT_sp_curve_FZZT} or \eqref{minimal_c_DT_sp_curve_chebyshev} for $m=1$, of the topological gravity. 
And when $m=2$, by setting $\tau_0=-3\mu/8, \tau_1=0, \tau_2=1$, 
the Hamiltonian \eqref{cont_dt_hamiltonian} yields the Hamiltonian $\mathcal{H}_{\mathrm{conti.DT}}^{\mathrm{pure}}$ of the continuum pure DT model in \eqref{cont_pure_dt_ham}, and 
the equations in \eqref{minimal_c_DT_sp_curve_FZZT} yield the spectral curve 
$y=(x -\sqrt{\mu}/2) \sqrt{x + \sqrt{\mu}}$ as given in 
\eqref{cont_DT_sp_curve}.
Mathematically, the amplitudes $\sfF_{\nn}^{(g)}(\bm{x}_{I})$ of the topological gravity (resp. pure gravity) give the $\psi$-class intersection numbers on the moduli space of stable curves of genus $g$ with $n$ marked points 
\cite{Witten:1990hr,Kontsevich:1992ti}
(resp. a WKB solution for the isomonodromy system for the Painlev\'e I equation \cite{Bergere:2013qba,Iwaki:2015xnr}). 
\end{rem}

\subsubsection{A supersymmetric analogue of Section \ref{subsubsec:mcdt}} 

Next, we provide a Hamiltonian for a one-cut spectral curve
\begin{align}
y=\frac{q(x)}{x} \sqrt{\sigma(x)}\,,
\quad
q(x)= \tau_{m}\, x^{m-1} + \sum_{k=0}^{m-2} q_k\, x^k\,,
\quad
\sigma(x)=x - \alpha\,.
\label{super_minimal_c_DT_sp_curve}
\end{align}
In this case, 
\begin{align}
r=1\,,\ \ 
\mu=1\,,\ \ 
s=m-1\ (m \ge 1)\,,
\end{align}
and we set
\begin{align}
p_0=0\,,\ \
p_1=1\,.
\end{align}
Then the parameter $\kappa_2$ in \eqref{nu0_ai} is $\kappa_2=\tau_0^2$, and 
the Hamiltonian \eqref{gen_hamiltonian_ai} yields
\begin{align}
-\mathcal{H}_{\mathrm{SMDT}}^{(m)}&=
4g_s^{-1} \tau_0^2 \Psi(4) + 2g_s^{-1} \kappa_1 \Psi(2) 
+ \frac{g_s}{8} \Psi(2) + \tau_0 \Psi(1)\Psi(2)
\nonumber\\
&\ \
+2\sum_{k=0}^{m} \tau_k
\sum_{\ell \ge 1} \ell\, \Psd(\ell+2k-3)\Psi(\ell)
\nonumber\\
&\ \
+g_s \sum_{\ell, \ell' \ge 1}
\left(\ell+\ell'+2\right)
\Psd(\ell)\Psd(\ell')\Psi(\ell+\ell'+2)
+\frac{g_s}{4} \sum_{\ell, \ell' \ge 1}
\ell \ell'\, 
\Psd(\ell+\ell'-2) \Psi(\ell) \Psi(\ell')
\,.
\label{super_dt_hamiltonian}
\end{align}
The spectral curve \eqref{super_minimal_c_DT_sp_curve} is obtained from 
the equation \eqref{sp_curve_determine_red}, 
where $q_k$ ($0, \ldots, m-2$) and $\alpha$ are determined as functions of $\kappa_1$ and $\tau_{\ell}$ ($\ell=0, \ldots, m$).

Now we consider a specialization of the parameters $\kappa_1$ and $\tau_k$ as
\begin{align}
\kappa_1=\tau_0=0\,,
\quad
\tau_k=
\frac{(-4\pi^2)^{k-1}}{(2k-2)!}
\prod_{i=1}^{k-1} \left(1-\left(\frac{2i-1}{2m-1}\right)^2\right),
\quad
k=1, \ldots, m\,,
\end{align}
such that the polynomial $q(x)$ and the parameter $\alpha$ in \eqref{super_minimal_c_DT_sp_curve} yield
\begin{align}
\begin{split}
&
\alpha=0\,,
\\
&
q(x)
=(-1)^{m-1}\,
U_{2m-2}\biggl(\frac{2\pi\sqrt{x}}{2m-1}\biggr)
=\sum_{k=0}^{m-1}\frac{(-4\pi^2 x)^k}{(2k)!}
\prod_{i=1}^k \left(1-\left(\frac{2i-1}{2m-1}\right)^2\right),
\label{super_minimal_c_DT_sp_curve_chebyshev}
\end{split}
\end{align}
where $U_{n}(x)$ is the Chebyshev polynomial of the second kind defined by 
$U_{n}(\cos \theta)\sin \theta=\sin((n+1) \theta)$.
This specialization is discussed in \cite{Fuji:2023wcx} as a supersymmetric analogue%
\footnote{
This supersymmetric analogue model is expected to be related to the $(2,4m-4)$ minimal superstring \cite{Seiberg:2003nm}.
}
of the spectral curve \eqref{minimal_c_DT_sp_curve_chebyshev}.
In particular, for $m=1$ on the one hand, 
we obtain the spectral curve $y=\sqrt{x}/x$, 
with $q(x)=1$, known as the Bessel curve \cite{Do:2016odu} which is 
associated to the Br\'ezin-Gross-Witten tau function \cite{Norbury:2017eih} of the KdV hierarchy. 
On the other hand, in the limit $m\to\infty$, 
we obtain the spectral curve with $q(x)=\cos(2\pi\sqrt{x})$, 
and the CEO topological recursion for this spectral curve is shown to be the Laplace dual of Stanford-Witten's recursion for the super Weil-Petersson volumes of moduli spaces of hyperbolic bordered super Riemann surfaces \cite{Stanford:2019vob,Norbury:2020vyi}.

\subsubsection{4D \texorpdfstring{$\mathcal{N}=2$}{N=2} pure \texorpdfstring{$SU(2)$}{SU(2)} gauge theory}
\label{subsubsec:pure_sw}

Finally, let us provide a Hamiltonian for the two-cut spectral curve in \cite{Gaiotto:2009hg,Gaiotto:2009ma},
\begin{align}
y=\frac{\Lambda}{x^2} \sqrt{\sigma(x)}\,,
\quad
\sigma(x)=x^3 + \frac{u}{\Lambda^2} x^2 + x\,,
\label{pure_su2_sw}
\end{align}
which describes
the Seiberg-Witten curve in 4D $\mathcal{N}=2$ pure $SU(2)$ gauge theory, where $\Lambda$ is the dynamical scale parameter and $u$ parametrizes the Coulomb branch.
In this case, 
\begin{align}
r=\mu=2\,,\ \ 
s=1\,,
\end{align}
and we set
\begin{align}
\kappa_1 = \kappa_2 = 0\,,\ \
\tau_0 = 0\,,\ \
\tau_1 = \tau_2 = \Lambda\,,\ \
p_0=p_1=0\,,\ \
p_2=1\,.
\end{align}
Then, the parameter $\kappa_3$ is determined in \eqref{nu0_ai} as $\kappa_3=0$, and the Hamiltonian \eqref{gen_hamiltonian_ai} is
\begin{align}
-\mathcal{H}_{\mathrm{pure}}^{SU(2)}&=
2\Lambda \sum_{k=1,2}
\sum_{\ell \ge 1} \ell\, \Psd(\ell+2k-3)\Psi(\ell)
\nonumber\\
&\ \
+g_s \sum_{\ell, \ell' \ge 1}
\left(\ell+\ell'\right) \Psd(\ell)\Psd(\ell')\Psi(\ell+\ell')
+\frac{g_s}{4} \sum_{\ell, \ell' \ge 1}
\ell \ell'\, \Psd(\ell+\ell') \Psi(\ell) \Psi(\ell')\,.
\label{pure_sw_hamiltonian}
\end{align}
The equation \eqref{sp_curve_determine_red} which determines a spectral curve is
\begin{align}
\Lambda^2 \sigma(x)&=
\Lambda^2 x^3 + 2 \Lambda \left(f + \Lambda\right)x^2
+\Lambda^2 x\,,
\end{align}
where
\begin{align}
f:=\lim_{g_s \to 0} g_s \lim_{T \to \infty} \left\langle\mathrm{vac}\Big| \e^{-T \mathcal{H}_{\mathrm{pure}}^{SU(2)}}\,
\Psd(1) \Big|\mathrm{vac}\right\rangle.
\end{align}
As a result, the spectral curve \eqref{pure_su2_sw} is obtained, 
where the parameter $u = 2\Lambda \left(f + \Lambda\right)$ is determined by the $A$-period
\begin{align}
a = \oint_{\mathcal{A}} y\, dx\,.
\end{align}
As noted in Remark \ref{rem:nf4_sw} of Section \ref{subsubsec:nf4_sw}, 
the amplitudes $\sfF_{\nn}^{(g)}(\bm{x}_{I})$, defined by 
the Hamiltonian \eqref{pure_sw_hamiltonian}, 
encode the perturbative expansion of the instanton partition function with $\nn$ half-BPS simple-type surface defects, 
in the self-dual $\Omega$-background with $g_s=\hbar$ \cite{Alday:2009fs,Dimofte:2010tz,Kozcaz:2010af} (see also \cite{Awata:2010bz}).


\acknowledgments{
We would like to thank Hiroaki Kanno, Kento Osuga, and Shintarou Yanagida for valuable comments.
This work was supported by JSPS KAKENHI Grant Numbers JP23K22388 and JP25K07278.
}

\appendix
\section{Spectral curve as a planar SD equation}
\label{app:sp_to_sd}

This appendix presents the detailed computations carried out in the background in order to derive the Hamiltonians in Definitions \ref{def:hamiltonian} and \ref{def:hamiltonian_red}.

\subsection{Even number of branch points}
\label{app:rev_sd}

We here discuss a reverse construction of a ``planar SD equation'' by rewriting the spectral curve \eqref{sp_curve_g} with even number of branch points.
Consulting the derivation of the spectral curve \eqref{sp_curve_basic} of the basic-type discrete DT model from the planar SD equation \eqref{sp_sd_sim_n1_planar}, 
we rewrite the spectral curve \eqref{sp_curve_g} as
\begin{align}
y= f(x) + \frac{\Delta_{\nu}(x)}{p(x)}\,,
\label{sp_rw}
\end{align}
then the functions $f(x)$ and $\Delta_{\nu}(x)$ obey
\begin{align}
\begin{split}
&
0= p(x)\, f(x)^2 + 2 \Delta_{\nu}(x)\, f(x) + \widetilde{K}(x)\,,
\\
&
\widetilde{K}(x):=\frac{1}{p(x)} \left(\Delta_{\nu}(x)^2 - q(x)^2\, \sigma(x)\right).
\label{key_relation}
\end{split}
\end{align}
Here we assume that
\begin{itemize}
\item[\textbf{1.}]
the function $f(x)$ is identified with a disk amplitude $f_1^{(0)}(x)$ and behaves at $x=\infty$ as
\begin{align}
f(x)=O(x^{-2})
\quad (x \to \infty)\,,
\label{disk_id}
\end{align}
\item[\textbf{2.}]
$\Delta_{\nu}(x)$ is a Laurent polynomial in $x$ as 
\begin{align}
\Delta_{\nu}(x)=\sum_{k=-\nu}^{s+1}\tau_{k+1}\, x^k\,,
\quad \nu \ge 0\,,
\label{delta_def}
\end{align}
where we fix $\nu=1$ in Section \ref{sec:ham_tr} as in \eqref{branch_condition_c} 
and define $\Delta(x)=\Delta_{1}(x)$,
\item[\textbf{3.}]
$\widetilde{K}(x)$ is a Laurent polynomial in $x$ expanded as
\begin{align}
\widetilde{K}(x) = \frac{1}{p(x)} \left(\Delta_{\nu}(x)^2 - q(x)^2\, \sigma(x)\right)
=\sum_{k=-\mu-2\nu}^{2s-r+1} \wkp_k\, x^k\,,
\label{k_laurent}
\end{align}
where we require $\widetilde{K}(x) \neq 0$ so that the spectral curve has branch cuts, and
\begin{align}
2s + 2\nu + 1 \ge r - \mu \ge 0\,,
\label{condi_num}
\end{align}
is imposed,
\item[\textbf{4.}]
(parameter ansatz for reverse construction): 
$s+h+3$ parameters $\alpha_k$ ($k=1, \ldots, 2h+2$) in $\sigma(x)$ and 
$q_k$ ($k=0, 1, \ldots, s-h$) in $q(x)$ of the spectral curve are determined by 
the parameters $p_k$ ($k=\mu, \mu+1, \ldots, r$) in $p(x)$, 
$\tau_k$ ($k=-\nu+1, -\nu+2, \ldots, s+2$) in $\Delta_{\nu}(x)$, 
$\wkp_k$ in \eqref{k_laurent}, and $h$ $A$-periods
\begin{align}
P_i=
\frac{1}{2 \pi \mathsf{i}} \oint_{[\alpha_{2i-1}, \alpha_{2i}]} 
M(x)\sqrt{\sigma(x)}\, dx\,,
\qquad
i=1,2,\ldots, h\,.
\label{periods_c_app}
\end{align}
\end{itemize}
Assumption \textbf{3} provides $2s+2\nu+3$ conditions for the coefficients of $x^{k}$ $(k=0,1,\ldots,2s+2\nu+2)$ in
\begin{align}
\begin{split}
&
\left(\sum_{k=0}^{s+\nu+1}\tau_{k-\nu+1}\, x^{k}\right)^2
-x^{2\nu}
\left(\sum_{k=0}^{s-h} q_k\, x^{k}\right)^2
\prod_{k=1}^{2h+2}(x - \alpha_k)
\\
&=
\left(\sum_{k=0}^{r-\mu} p_{k+\mu}\, x^{k}\right)
\left(\sum_{k=0}^{2s-r+\mu+2\nu+1} \wkp_{k-\mu-2\nu}\, x^{k}\right).
\label{condition_coeff}
\end{split}
\end{align}
In particular, these conditions determine $2\nu$ parameters $\wkp_k$ 
($k=-\mu-2\nu, -\mu-2\nu+1, \ldots, -\mu-1$) in terms of the parameters 
$p_k$ ($k=\mu, \mu+1, \ldots, r$) in $p(x)$ and 
$\tau_k$ ($k=-\nu+1, -\nu+2, \ldots, s+2$) in $\Delta_{\nu}(x)$:
\begin{align}
\wkp_k=\kappa_{-k}
:=\kappa_{-k}(p_{\mu}, \ldots, p_r, \tau_{-\nu+1}, \ldots, \tau_{s+2})\,,
\label{determined_kp}
\end{align}
by the conditions arising from the identity in $x$:
\begin{align}
\sum_{\substack{k,\ell=1 \\ k+\ell \le 2\nu+1}}^{s+\nu+2}
\tau_{k-\nu}\tau_{\ell-\nu}\, x^{k+\ell}
=
\mathop{\sum_{k=1}^{r-\mu+1}\ \sum_{\ell=1}^{2\nu}}_{k+\ell \le 2\nu+1}
p_{k+\mu-1} \kappa_{\mu+2\nu-\ell+1}\, x^{k+\ell}\,.
\label{condition_coeff_c}
\end{align}
The remaining $2s+3$ conditions in \eqref{condition_coeff} and $h$ $A$-periods
\eqref{periods_c_app} determine the $s+h+3$ parameters $\alpha_k$ ($k=1, \ldots, 2h+2$) and $q_k$ ($k=0, 1, \ldots, s-h$), and in particular $q_{s-h}=\tau_{s+2}$ is determined.
The number of the remaining parameters 
$\wkp_k$ ($k=-\mu, -\mu+1, \ldots, 2s-r+1$) is $2s-r+\mu+2$ and 
the number of the remaining conditions is $s$, and we further assume the following.
\begin{itemize}
\item[\textbf{5.}]
The $s$ remaining conditions determine 
$\wkp_k$ ($k=0, 1, \ldots, s-1$), and they give 
the regular part of $-2 \Delta_{\nu}(x) f(x)$ as
\begin{align}
\sum_{k=0}^{s-1} \wkp_{k}\, x^k
=
-\left(2 \Delta_{\nu}(x)\,  f(x)\right)_{\mathrm{reg}(x)}
+ \sum_{k=0}^{s-1} \gamma_k\, x^k\,,
\label{kappa_assume}
\end{align}
where $\mathrm{reg}(x)$ denotes the regular part of the Laurent series of $x$ at $x=0$, and 
$\gamma_k$ are constants determined below by the asymptotic condition \eqref{disk_id} stated in Assumption \textbf{1}. 
Here a condition $2s-r+1 \ge s-1$ for the number of parameters $\wkp_{k}$, i.e.,
\begin{align}
s+2 \ge r\,,
\label{condition_param}
\end{align}
is imposed.%
\footnote{The condition \eqref{condition_param} ensures that $2s-r+\mu+2$ (the number of the remaining $\wkp_k$) 
is greater than or equal to $s$ (the number of the remaining conditions) 
so as not to be an over-constrained system.}
Then, we see that, instead of the condition $2s-r+\mu+2\nu+1\ge 0$ in \eqref{condi_num}, it is enough to impose
\begin{align}
s \ge 1\ \ \ \ \textrm{or}\ \ \ \
\mu+2\nu \ge 1\,.
\label{branch_condition_c_gen}
\end{align}
\end{itemize}
By Assumptions \textbf{3} and \textbf{5}, we have
\begin{align}
\widetilde{K}(x) =-\left(2 \Delta_{\nu}(x)\,  f(x)\right)_{\mathrm{reg}(x)}
+ \sum_{k=0}^{s-1} \gamma_k\, x^k 
+ \sum_{k=s}^{2s-r+1} \wkp_k\, x^{k}
+ \sum_{k=1}^{\mu+2\nu} \kappa_k\, x^{-k}\,,
\label{k_eq}
\end{align}
where $\kappa_k=\wkp_{-k}$ ($k=\mu+1, \mu+2, \ldots, \mu+2\nu$) are parameters determined in \eqref{determined_kp}, and
$\kappa_k:=\wkp_{-k}$ ($k=1, 2, \ldots, \mu$) are considered as independent parameters. 
By the asymptotic condition $f(x)=O(x^{-2})$ as $x \to \infty$ in \eqref{disk_id}, 
the equation \eqref{key_relation}:
\begin{align}
\begin{split}
0&= p(x)\, f(x)^2 + 
\left(2 \Delta_{\nu}(x)\, f(x) -\left(2 \Delta_{\nu}(x)\,  f(x)\right)_{\mathrm{reg}(x)}\right)
\\
&\ \
+ \sum_{k=0}^{s-1} \gamma_k\, x^k 
+ \sum_{k=s}^{2s-r+1} \wkp_k\, x^{k}
+ \sum_{k=1}^{\mu+2\nu} \kappa_k\, x^{-k}\,,
\end{split}
\end{align}
implies that $\gamma_k=0$ and $\wkp_k=0$ for $k \ge r-3$,
and by the condition \eqref{condition_param} 
all the parameters $\wkp_k$ ($k=s, \ldots, 2s-r+1$) in \eqref{k_eq} are vanished, and 
\begin{align}
\sum_{k=0}^{r-4} \gamma_k\, x^k =
-\left(p(x)\, f(x)^2\right)_{\mathrm{reg}(x)}\,.
\label{k_eq2}
\end{align} 
As a result, we have
\begin{align}
\widetilde{K}(x) =
-\left(2 \Delta_{\nu}(x)\,  f(x) + p(x)\, f(x)^2\right)_{\mathrm{reg}(x)}
+ K_{-,\nu}(x)\,,
\label{k_eq3}
\end{align}
where 
\begin{align}
K_{-,\nu}(x)=\sum_{k=1}^{\mu+2\nu} \kappa_k\, x^{-k}\,.
\label{k_input_2_nu}
\end{align} 
Then, the equation \eqref{key_relation} yields a ``planar SD equation'' for the function $f(x)$:
\begin{align}
0= 
\left(p(x)\, f(x)^2 + 2 \Delta_{\nu}(x)\, f(x)\right)_{\mathrm{irreg}(x)}
+ K_{-,\nu}(x)\,,
\label{planar_sd}
\end{align}
where $F(x)_{\mathrm{irreg}(x)}:=F(x)-F(x)_{\mathrm{reg}(x)}$ is defined for a function $F(x)$ of $x$. 
This generalizes the planar SD equation \eqref{sp_sd_sim_n1_planar} of the basic-type discrete DT model. 
Note that, in Section \ref{sec:ham_tr}, we fix $\nu=1$ as in \eqref{branch_condition_c} 
and define $K_{-}(x)=K_{-,1}(x)$.

Similar to the discussion in Section \ref{subsubsec:discrete_pureDT} or \ref{subsec:n1_sd}, 
we find that the planar SD equation \eqref{planar_sd} for 
the disk amplitude $f(x)=f_{1}^{(0)}(x)$ is derived from a $W^{(3)}$-type Hamiltonian $\mathcal{H}_{\nu}$:
\begin{align}
\begin{split}
-\mathcal{H}_{\nu}&=
g_s^{-1}
\sum_{k=1}^{\mu+2\nu} k \, \kappa_{k}\, \Psi(k)
+
\sum_{k=0}^{\nu-1} \tau_{-k}
\sum_{\ell=1}^{k+1} \ell \left(k+2-\ell\right)
\Psi(\ell) \Psi(k+2-\ell)
\\
&\ \
+2\sum_{k=-\nu+1}^{s+2} \tau_k
\sum_{\ell \ge 1}
\ell\, \Psd(\ell+k-2)\Psi(\ell)
\\
&\ \
+g_s \sum_{k=\mu}^{r} p_k
\sum_{\ell, \ell' \ge 1}
\left(\ell+\ell'-k+2\right)
\Psd(\ell)\Psd(\ell')\Psi(\ell+\ell'-k+2)
\\
&\ \
+g_s \sum_{k=\mu}^{r} p_k
\sum_{\ell, \ell' \ge 1}
\ell \ell'\, 
\Psd(\ell+\ell'+k-2) \Psi(\ell) \Psi(\ell')\,,
\label{gen_hamiltonian_nu}
\end{split}
\end{align}
where $g_s$ is the string coupling constant, and the conditions in \eqref{condition_coeff_c}, as well as those in  
\eqref{condition_param} and \eqref{branch_condition_c_gen}, are imposed. 
Here we adopt the parameters $p_k$ as coupling constants of three-string interaction terms, as motivated by the example in Section \ref{subsubsec:discrete_pureDT}.
In Section \ref{sec:ham_tr}, we fix $\nu=1$ (see Remark \ref{rem:nu} in the following) for simplicity. 
Following the discussion in Section \ref{subsec:tr}, 
it is straightforward to see that, instead of \eqref{free_en_def} with $\nu=1$, introducing
\begin{align}
F_1(x) = f_1(x) + g_s^{-1} \frac{\Delta_{\nu}(x)}{p(x)}\,,
\end{align}
the CEO topological recursion gives the perturbative amplitudes defined by the Hamiltonian $\mathcal{H}_{\nu}$, where the $\nu$ dependence only appears in the disk amplitude.

\begin{rem}\label{rem:nu}
For the condition \eqref{branch_condition_c_gen}, it is further 
enough to fix $\nu$ as
\begin{align}
\nu = 1\,,
\label{branch_condition_c}
\end{align}
and the conditions in \eqref{condition_coeff_c} determines $\wkp_{-\mu-1}=\kappa_{\mu+1}$ and $\wkp_{-\mu-2}=\kappa_{\mu+2}$ as
\begin{align}
&
\tau_{0}^2=p_{\mu} \kappa_{\mu+2}\,,\ \
2 \tau_{0} \tau_{1}=p_{\mu} \kappa_{\mu+1} + p_{\mu+1} \kappa_{\mu+2}\,,
\nonumber\\
& \Longleftrightarrow\ \
\kappa_{\mu+1}=
\frac{2 \tau_{0} \tau_{1}}{p_{\mu}}-\frac{\tau_{0}^2p_{\mu+1}}{p_{\mu}^2}\,,\ \
\kappa_{\mu+2}=\frac{\tau_{0}^2}{p_{\mu}}\,.
\label{condition_coeff_cc_1}
\end{align}
In Section \ref{sec:ham_tr}, we fix $\nu = 1$ as above, and only the condition \eqref{condition_param} needs to be imposed.
\end{rem}

\subsection{Odd number of branch points}
\label{app:rev_sd_ai}

We repeat a similar discussion in Appendix \ref{app:rev_sd} for the spectral curve \eqref{sp_curve_g_c} with odd number of branch points.
We rewrite this spectral curve, 
by \eqref{sp_rw}, i.e., $y= f(x) + \sfDel_{\nu}(x)/p(x)$, 
as \eqref{key_relation}:
\begin{align}
\begin{split}
&
0= p(x)\, f(x)^2 + 2 \sfDel_{\nu}(x)\, f(x) + \widetilde{\sfK}(x)\,,
\\
&
\widetilde{\sfK}(x):=\frac{1}{p(x)} \left(\sfDel_{\nu}(x)^2 - q(x)^2\, \sigma(x)\right).
\label{key_relation_c}
\end{split}
\end{align} 
We assume that
\begin{itemize}
\item[\textbf{1.}]
the function $f(x)$ is identified with a disk amplitude $\sff_{1}^{(0)}(x)$ and 
behaves at $x=\infty$ as
\begin{align}
f(x)=O(x^{-3/2})
\quad (x \to \infty)\,,
\label{f_asympt}
\end{align}
instead of \eqref{disk_id},
\item[\textbf{2.}]
$\sfDel_{\nu}(x)$ is expanded as
\begin{align}
\sfDel_{\nu}(x)&=\sum_{k=-\nu}^{s}\tau_{k+1}\, x^{k+1/2}\,,
\quad \nu \ge 0\,,
\label{delta_def_c}
\end{align}
where we fix $\nu=1$ in Section \ref{sec:ham_ai} as in \eqref{nu_odd_fix} 
and define $\sfDel(x)=\sfDel_{1}(x)$,
\item[\textbf{3.}]
$\widetilde{\sfK}(x)$ is expanded as
\begin{align}
\widetilde{\sfK}(x) = \frac{1}{p(x)} \left(\sfDel_{\nu}(x)^2 - q(x)^2\, \sigma(x)\right)
=\sum_{k=-\mu-2\nu+1}^{2s-r} \wkp_k\, x^k\,,
\label{k_laurent_c}
\end{align}
where the condition
\begin{align}
2s + 2\nu - 1 \ge r - \mu \ge 0\,,
\label{condi_num_odd}
\end{align}
is imposed for $\widetilde{\sfK}(x) \neq 0$,
\item[\textbf{4.}]
(parameter ansatz for reverse construction): 
$s+h+2$ parameters $\alpha_k$ ($k=1, \ldots, 2h+1$) in $\sigma(x)$ and 
$q_k$ ($k=0, 1, \ldots, s-h$) in $q(x)$ of the spectral curve are determined by 
the parameters $p_k$ ($k=\mu, \mu+1, \ldots, r$) in $p(x)$,  
$\tau_k$ ($k=-\nu+1, -\nu+2, \ldots, s+1$) in $\sfDel_{\nu}(x)$, 
and $\wkp_k$ in \eqref{k_laurent_c}, and $h$ $A$-periods \eqref{periods_c_app}.
\end{itemize}
Assumption \textbf{3} provides $2s+2\nu+1$ conditions for the coefficients of $x^{k}$ $(k=-1,0,\ldots,2s+2\nu-1)$ in
\begin{align}
\begin{split}
&
\left(\sum_{k=0}^{s+\nu}\tau_{k-\nu+1}\, x^{k+1/2}\right)^2
-x^{2\nu-1}
\left(\sum_{k=0}^{s-h} q_k\, x^{k}\right)^2
\prod_{k=1}^{2h+1}(x - \alpha_k)
\\
&=
\left(\sum_{k=0}^{r-\mu} p_{k+\mu}\, x^{k}\right)
\left(\sum_{k=0}^{2s-r+\mu+2\nu-1} \wkp_{k-\mu-2\nu+1}\, x^{k}\right),
\label{condition_coeff_x}
\end{split}
\end{align}
and determines $2\nu-1$ parameters $\wkp_k$ 
($k=-\mu-2\nu+1, -\mu-2\nu+2, \ldots, -\mu-1$):
\begin{align}
\wkp_k=\kappa_{-k}
:=\kappa_{-k}(p_{\mu}, \ldots, p_r, \tau_{-\nu+1}, \ldots, \tau_{s+1})\,,
\label{determined_kp_c}
\end{align}
by the conditions arising from the identity in $x$:
\begin{align}
\sum_{\substack{k,\ell=1 \\ k+\ell \le 2\nu}}^{s+\nu+1}
\tau_{k-\nu}\tau_{\ell-\nu}\, x^{k+\ell}
=
\mathop{\sum_{k=1}^{r-\mu+1}\ \sum_{\ell=1}^{2\nu-1}}_{k+\ell \le 2\nu}
p_{k+\mu-1} \kappa_{\mu+2\nu-\ell}\, x^{k+\ell}\,.
\label{condition_coeff_xc}
\end{align}
The remaining $2s+2$ conditions in \eqref{condition_coeff_x} and $h$ $A$-periods \eqref{periods_c_app} determine the $s+h+2$ parameters $\alpha_k$ ($k=1, \ldots, 2h+1$) and $q_k$ ($k=0, 1, \ldots, s-h$), and in particular $q_{s-h}=\tau_{s+1}$ is determined.
The number of the remaining parameters 
$\wkp_k$ ($k=-\mu, -\mu+1, \ldots, 2s-r$) is $2s-r+\mu+1$ and 
the number of the remaining conditions is $s$, and 
we make the same assumption as in \eqref{kappa_assume}, i.e.,
\begin{itemize}
\item[\textbf{5.}]
the $s$ remaining conditions determine $\wkp_k$ ($k=0, 1, \ldots, s-1$) as
\begin{align}
\sum_{k=0}^{s-1} \wkp_{k}\, x^k
=
-\left(2 \sfDel_{\nu}(x)\,  f(x)\right)_{\mathrm{reg}(x)}
+ \sum_{k=0}^{s-1} \gamma_k\, x^k,
\end{align}
where $\mathrm{reg}(x)$ denotes the regular part of the Laurent series 
of $x^{1/2}$ at $x=0$, and 
$\gamma_k$ are constants determined by asymptotic condition \eqref{f_asympt} stated in Assumption \textbf{1}.
Here a condition $2s-r \ge s-1$ for the number of parameters $\wkp_{k}$, i.e.,
\begin{align}
s+1 \ge r,
\label{condition_param_c}
\end{align}
is imposed. 
\end{itemize}
By Assumptions \textbf{3} and \textbf{5},
\begin{align}
\widetilde{\sfK}(x) =-\left(2 \sfDel_{\nu}(x)\,  f(x)\right)_{\mathrm{reg}(x)}
+ \sum_{k=0}^{s-1} \gamma_k\, x^k 
+ \sum_{k=s}^{2s-r} \wkp_k\, x^{k}
+ \sum_{k=1}^{\mu+2\nu-1} \kappa_k\, x^{-k}\,,
\label{k_eq_c}
\end{align}
where $\kappa_k=\wkp_{-k}$ ($k=\mu+1, \mu+2, \ldots, \mu+2\nu-1$) are parameters determined in \eqref{determined_kp_c}, and
$\kappa_k=\wkp_{-k}$ ($k=1, 2, \ldots, \mu$) are free parameters. 
By the asymptotic condition $f(x)=O(x^{-3/2})$ as $x \to \infty$ 
in \eqref{f_asympt}, the equation \eqref{key_relation_c}:
\begin{align}
\begin{split}
0&= p(x)\, f(x)^2 + 
\left(2 \sfDel_{\nu}(x)\, f(x) -
\left(2 \sfDel_{\nu}(x)\,  f(x)\right)_{\mathrm{reg}(x)}\right)
\\
&\quad
+ \sum_{k=0}^{s-1} \gamma_k\, x^k 
+ \sum_{k=s}^{2s-r} \wkp_k\, x^{k}
+ \sum_{k=1}^{\mu+2\nu-1} \kappa_k\, x^{-k}\,,
\end{split}
\end{align}
implies that $\gamma_k=0$ and $\wkp_k=0$ for $k \ge r-2$,
and by the condition \eqref{condition_param_c} 
all the parameters $\wkp_k$ ($k=s, \ldots, 2s-r$) in \eqref{k_eq_c} are vanished, and 
\begin{align}
\sum_{k=0}^{r-3} \gamma_k\, x^k =
-\left(p(x)\, f(x)^2\right)_{\mathrm{reg}(x)}\,.
\label{k_eq2_c}
\end{align} 
We then obtain
\begin{align}
\widetilde{\sfK}(x) =
-\left(2 \sfDel_{\nu}(x)\,  f(x) + p(x)\, f(x)^2\right)_{\mathrm{reg}(x)}
+ \sfK_{-,\nu}(x)\,,
\label{k_eq3_c}
\end{align}
where 
\begin{align}
\sfK_{-,\nu}(x)&=\sum_{k=1}^{\mu+2\nu-1} \kappa_k\, x^{-k}\,,
\end{align}
and a ``planar SD equation'', from \eqref{key_relation_c}, with the same form as \eqref{planar_sd}, i.e., 
\begin{align}
0= \left(p(x)\, f(x)^2 + 2 \sfDel_{\nu}(x)\, f(x)\right)_{\mathrm{irreg}(x)}
+ \sfK_{-,\nu}(x)\,.
\label{planar_sd_odd}
\end{align}
In Section \ref{sec:ham_ai}, we fix $\nu=1$ as in \eqref{nu_odd_fix} 
and define $\sfK_{-}(x)=\sfK_{-,1}(x)$.

Similar to the $W^{(3)}$-type Hamiltonian $\mathcal{H}_{\nu}$ in \eqref{gen_hamiltonian_nu}, the planar SD equation \eqref{planar_sd_odd} for the disk amplitude $f(x)=\sff_{1}^{(0)}(x)$ is derived from a two-reduced $W^{(3)}$-type Hamiltonian $\mathcal{H}_{\nu}^{red}$:
\begin{align}
\begin{split}
-\mathcal{H}^{red}_{\nu}&=
2 g_s^{-1}
\sum_{k=1}^{\mu+2\nu-1} k \, \kappa_{k}\, \Psi(2k)
+
\frac{g_s}{8} \left(2p_0 \Psi(4) + p_1 \Psi(2)\right)
\\
&\ \
+ \frac14 \sum_{k=0}^{\nu-1} \tau_{-k}
\sum_{\ell=1}^{2k+2} \ell \left(2k+3-\ell\right)
\Psi(\ell) \Psi(2k+3-\ell)
+ \frac{g_s p_0}{8} \Psi(1)^2 \Psi(2)
\\
&\ \
+2\sum_{k=-\nu+1}^{s+1} \tau_k
\sum_{\ell \ge 1}
\ell\, \Psd(\ell+2k-3)\Psi(\ell)
\\
&\ \
+g_s \sum_{k=\mu}^{r} p_k
\sum_{\ell, \ell' \ge 1}
\left(\ell+\ell'-2k+4\right)
\Psd(\ell)\Psd(\ell')\Psi(\ell+\ell'-2k+4)
\\
&\ \
+\frac{g_s}{4} \sum_{k=\mu}^{r} p_k
\sum_{\ell, \ell' \ge 1}
\ell \ell'\, 
\Psd(\ell+\ell'+2k-4) \Psi(\ell) \Psi(\ell')\,,
\label{gen_hamiltonian_ai_nu}
\end{split}
\end{align}
where $g_s$ is the string coupling constant, and the conditions in \eqref{condition_coeff_xc} are imposed as well as the conditions 
\eqref{condi_num_odd} and \eqref{condition_param_c}.
Here, as indicated by the motivated example in Section \ref{subsubsec:conti_pureDT}, 
the parameters $p_k$ are adopted as coupling constants of three-string interaction terms, 
and two tadpole operators $\Psi(2)$ and $\Psi(4)$, as well as 
a three-string annihilation operator $\Psi(1)^2 \Psi(2)$, 
are also introduced, in comparison with the $W^{(3)}$-type Hamiltonian $\mathcal{H}_{\nu}$ in \eqref{gen_hamiltonian_nu}.
It is, then, shown that the CEO topological recursion gives the perturbative amplitudes defined by the Hamiltonian $\mathcal{H}^{red}_{\nu}$, where the $\nu$ dependence only appears in the disk amplitude.
In Section \ref{sec:ham_ai}, we fix $\nu=1$ (see Remark \ref{rem:nu_ai} in the following) for simplicity.

\begin{rem}\label{rem:nu_ai}
Taking into account of the condition \eqref{condition_param_c}, 
instead of the condition \eqref{condi_num_odd} it is enough to fix $\nu$ as
\begin{align}
\nu = 1\,,
\label{nu_odd_fix}
\end{align}
and $\wkp_{-\mu-1}=\kappa_{\mu+1}$ is determined by \eqref{condition_coeff_xc} as
\begin{align}
\kappa_{\mu+1}=\frac{\tau_0^2}{p_{\mu}}\,.
\label{nu0_ai_app}
\end{align}
In Section \ref{sec:ham_ai}, we fix $\nu = 1$ as above, and only the condition \eqref{condition_param_c} needs to be imposed.
\end{rem}

\section{Decoupling scheme in 4D \texorpdfstring{$\mathcal{N}=2$}{N=2} \texorpdfstring{$SU(2)$}{SU(2)} gauge theories}
\label{app:decoupling}

In this appendix, we present Hamiltonians of 4D $\mathcal{N}=2$ $SU(2)$ gauge theories with $N_f=0,1,2,3,4$ hypermultiplets, from the point of view of decoupling these hypermultiplets as described in Fig. \ref{fig:degen_sw2} \cite{Gaiotto:2009hg,Gaiotto:2009ma,Marshakov:2009gn,Eguchi:2009gf} (see also \cite{Awata:2010bz}).
\begin{figure}[t]
\centering
\includegraphics[width=160mm]{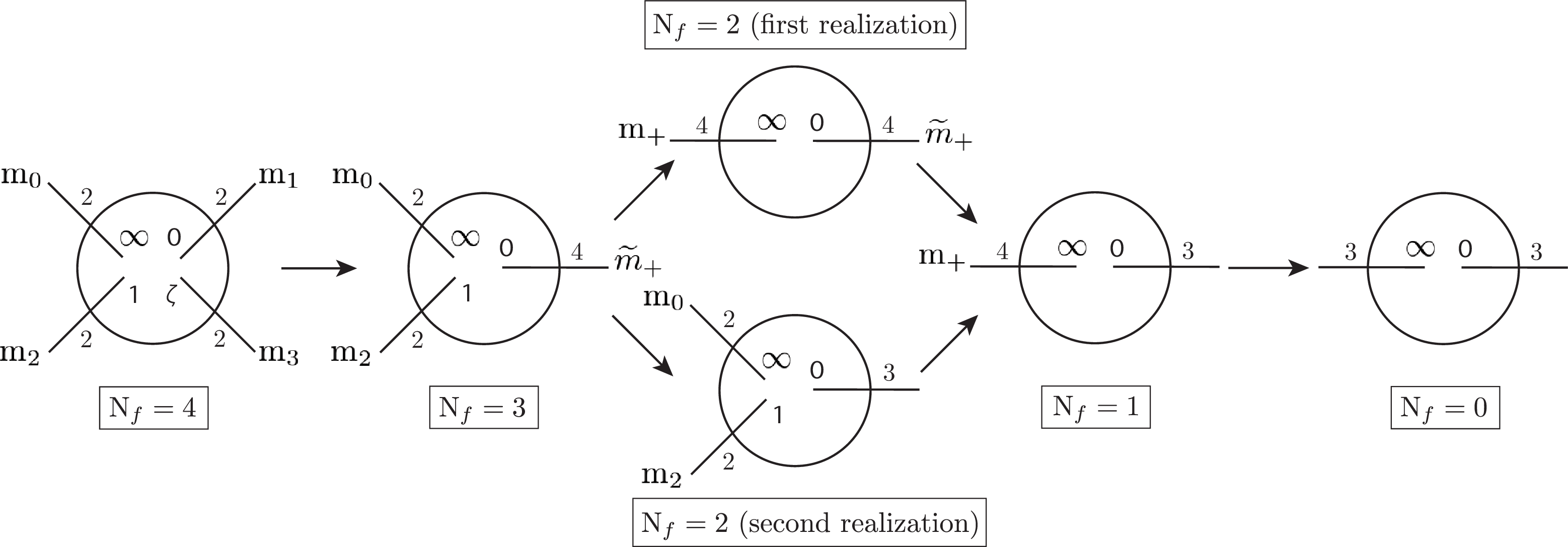}
\caption{Decoupling scheme of hypermultiplets, where the number assigned to each edge represents the degree of singularity of the quadratic differential $\phi(x)dx^2$ for Seiberg-Witten curves $y^2=\phi(x)$.}
\label{fig:degen_sw2}
\end{figure}

\subsection*{\underline{$N_f=4$}}

When $N_f=4$, the Hamiltonian \eqref{sw_su2_nf4_hamiltonian} is provided in Section \ref{subsubsec:nf4_sw}.
For convenience, we present it again here.

In this case, the Seiberg-Witten curve is given by \eqref{su2_nf4_sw} with 
$r=3$, $\mu=1$\, $s=1$, 
where the mass parameters $m_0, m_1, m_2, m_3$ are, respectively, associated with poles of degree 2 at $x=\infty, 0, 1, \zeta$ as
\begin{align}
\begin{split}
\frac{1}{2\pi \sqrt{-1}} \oint_{x=\infty} ydx &= m_0\,,
\quad
\frac{1}{2\pi \sqrt{-1}}\oint_{x=0} ydx = m_1\,,
\\
\frac{1}{2\pi \sqrt{-1}}\oint_{x=1} ydx &= m_2\,,
\quad
\frac{1}{2\pi \sqrt{-1}}\oint_{x=\zeta} ydx = m_3\,.
\end{split}
\label{pole_nf4}
\end{align}
The parameters for the $W^{(3)}$-type Hamiltonian \eqref{gen_hamiltonian} are taken as (see footnote \ref{ft:sw_red} for the parameter redundancy),
\begin{align}
\begin{split}
p_1&=\zeta\,,\ \ p_2=-1-\zeta\,,\ \ p_3=1\,,
\ \
\kappa_1=
\left(m_{0}+m_{1}+m_{2}+m_{3}\right)
\left(m_{0}-m_{1}+m_{2}+m_{3}\right)\zeta\,,
\\
\tau_{0}&=0\,,
\ \
\tau_{1}=\left(m_{0}+m_{2}+m_{3}\right)\zeta\,,
\ \
\tau_{2}=
-\left(m_{0}+m_{2}\right)
-\left(m_{0}+m_{3}\right)\zeta\,,
\ \
\tau_3=m_0\,.
\label{sw2_p_4}
\end{split}
\end{align}
Then, the Hamiltonian is
\begin{align}
\begin{split}
-\mathcal{H}^{SU(2)}_{N_f=4}&=
g_s^{-1} \kappa_{1}\, \Psi(1)
+2\sum_{k=1}^{3} \tau_k
\sum_{\ell \ge 1}
\ell\, \Psd(\ell+k-2)\Psi(\ell)
\\
&\ \
+g_s \sum_{k=1}^{3} p_k
\sum_{\ell, \ell' \ge 1}
\left(\ell+\ell'-k+2\right)
\Psd(\ell)\Psd(\ell')\Psi(\ell+\ell'-k+2)
\\
&\ \
+g_s \sum_{k=1}^{3} p_k
\sum_{\ell, \ell' \ge 1}
\ell \ell'\, 
\Psd(\ell+\ell'+k-2) \Psi(\ell) \Psi(\ell')\,,
\label{sw2_ham_4}
\end{split}
\end{align}
and the two-cut ansatz for the equation \eqref{sp_curve_determine} introduces the Coulomb branch parameter $U$ by
\begin{align}
\begin{split}
&
2 m_0 \lim_{g_s \to 0} g_s \lim_{T \to \infty} \left\langle\mathrm{vac}\Big| \e^{-T \mathcal{H}_{N_f=4}^{SU(2)}}\,
\Psd(1) \Big|\mathrm{vac}\right\rangle
\\
&=
-\left(\zeta +1\right) U
+\left(m_0+m_2\right)^2
+\left(m_{0}-m_{2}\right) 
\left(m_{0} + m_{2} + 2 m_{3}\right)\zeta\,.
\label{sw2_coulomb_4}
\end{split}
\end{align}
To discuss the decoupling of massive flavors in the following, 
it is useful to introduce
\begin{align}
m_{+} = m_2 + m_0\,,
\quad
m_{-} = m_2 - m_0\,,
\quad
\widetilde{m}_{+} = m_3 + m_1\,,
\quad
\widetilde{m}_{-} = m_3 - m_1\,.
\label{mass_redef}
\end{align}

\subsection*{\underline{$N_f=3$}}

When $N_f=3$, the Seiberg-Witten curve is obtained from the curve \eqref{su2_nf4_sw} for $N_f=4$ by the decoupling limit
\begin{align}
\widetilde{m}_{-} \to \infty\,,
\quad
\zeta \to 0\,,
\quad
\widetilde{m}_{-}\, \zeta=\Lambda_3\,
\ \textrm{(fixed)}\,,
\label{decouple_3}
\end{align}
as \cite{Eguchi:2009gf},
\begin{align}
y^2 = 
\frac{\Lambda_3^{2}}{4x^{4}}
-\frac{\widetilde{m}_{+}\Lambda_3}{x^{3} \left(x-1\right)}
-\frac{U+m_{2}\Lambda_3}{x^{2}\left(x-1\right)}
+\frac{m_{0}^{2}}{x \left(x-1\right)}
+\frac{m_{2}^{2}}{x \left(x-1\right)^{2}}
=:
\frac{m_0^2}{p(x)^2}\, \sigma(x)\,,
\label{sw2_curve_3}
\end{align}
with $r=3$, $\mu=2$, $s=1$, where $p(x)=x^{2}(x-1)$.
The mass parameters $m_0, m_2$ are, respectively, associated with poles of degree 2 at $x=\infty, 1$, and 
the mass parameter $\widetilde{m}_{+}$ is associated with a pole of degree 4 at $x=0$ as
\begin{align}
\frac{1}{2\pi \sqrt{-1}} \oint_{x=\infty} ydx = m_0\,,
\quad
\frac{1}{2\pi \sqrt{-1}}\oint_{x=1} ydx &= m_2\,,
\quad
\frac{1}{2\pi \sqrt{-1}}\oint_{x=0} ydx = \widetilde{m}_{+}\,.
\label{pole_nf3}
\end{align}
By taking the limit \eqref{decouple_3} for the parameters \eqref{sw2_p_4}, 
we here set
\begin{align}
\begin{split}
p_2&=-1\,,\ \ p_3=1\,,
\ \
\kappa_1=
\left(m_{0}+m_{2}+\widetilde{m}_{+}\right)\Lambda_3\,,
\ \
\kappa_2=0\,,
\\
\tau_{0}&=0\,,
\ \
\tau_{1}=\frac12\, \Lambda_3\,,
\ \
\tau_{2}=
-\left(m_{0}+m_{2}\right)
-\frac12\, \Lambda_3\,,
\ \
\tau_3=m_0\,,
\label{sw2_p_3}
\end{split}
\end{align}
for the $W^{(3)}$-type Hamiltonian \eqref{gen_hamiltonian},
and obtain
\begin{align}
\begin{split}
-\mathcal{H}^{SU(2)}_{N_f=3}&=
g_s^{-1} \kappa_{1}\, \Psi(1)
+2\sum_{k=1}^{3} \tau_k
\sum_{\ell \ge 1}
\ell\, \Psd(\ell+k-2)\Psi(\ell)
\\
&\ \
+g_s \sum_{k=2,3} p_k
\sum_{\ell, \ell' \ge 1}
\left(\ell+\ell'-k+2\right)
\Psd(\ell)\Psd(\ell')\Psi(\ell+\ell'-k+2)
\\
&\ \
+g_s \sum_{k=2,3} p_k
\sum_{\ell, \ell' \ge 1}
\ell \ell'\, 
\Psd(\ell+\ell'+k-2) \Psi(\ell) \Psi(\ell')\,.
\label{sw2_ham_3}
\end{split}
\end{align}
By introducing the Coulomb branch parameter $U$ by
\begin{align}
2 m_0 \lim_{g_s \to 0} g_s \lim_{T \to \infty} \left\langle\mathrm{vac}\Big| \e^{-T \mathcal{H}_{N_f=3}^{SU(2)}}\,
\Psd(1) \Big|\mathrm{vac}\right\rangle
=
-U + \left(m_0+m_2\right)^2
+\left(m_{0}-m_{2}\right) \Lambda_3\,,
\label{sw2_coulomb_3}
\end{align}
which is consistent with the decoupling limit of \eqref{sw2_coulomb_4}, 
the two-cut ansatz for the equation \eqref{sp_curve_determine} leads to the Seiberg-Witten curve \eqref{sw2_curve_3}.

\subsection*{\underline{$N_f=2$ (first realization)}}

When $N_f=2$, two different types of Seiberg-Witten curves can be found \cite{Gaiotto:2009hg,Gaiotto:2009ma}. 
The Seiberg-Witten curve of the first realization is obtained from the curve \eqref{sw2_curve_3} for $N_f=3$ by the decoupling limit
\begin{align}
m_{-} \to \infty\,,
\quad
\Lambda_3 \to 0\,,
\quad
m_{-}\, \Lambda_3=\Lambda_2^2\,
\ \textrm{(fixed)}\,,
\label{decouple_2_1}
\end{align}
and scalings
\begin{align}
x \to \frac{\Lambda_3}{\Lambda_2}\, x\,,
\quad
y \to \frac{\Lambda_2}{\Lambda_3}\, y\,,
\label{sw_scaling_nf2_1}
\end{align}
as \cite{Gaiotto:2009hg,Gaiotto:2009ma,Eguchi:2009gf},
\begin{align}
y^2 = 
\frac{\Lambda_2^{2}}{4x^{4}}+\frac{\widetilde{m}_{+}\Lambda_2}{x^{3}}
+\frac{u}{x^{2}}+\frac{m_{+}\Lambda_2}{x}+\frac{\Lambda_2^{2}}{4}
=:
\frac{\Lambda_2^2}{4x^4}\, \sigma(x)
\,,
\label{sw2_curve_2_1}
\end{align}
where $u=U+\Lambda_2^2/2$ is defined.
The mass parameters $m_{+}, \widetilde{m}_{+}$ are, respectively, associated with poles of degree 4 at $x=\infty, 0$ as
\begin{align}
\frac{1}{2\pi \sqrt{-1}} \oint_{x=\infty} ydx = m_{+}\,,
\quad
\frac{1}{2\pi \sqrt{-1}}\oint_{x=0} ydx = \widetilde{m}_{+}\,.
\label{pole_nf2_1}
\end{align}
In this case, $r=\mu=2$, $s=1$, and for the $W^{(3)}$-type Hamiltonian \eqref{gen_hamiltonian}, we set
\begin{align}
\begin{split}
p_2&=-1\,,
\ \
\kappa_1=
\left(m_{+}+\widetilde{m}_{+}\right)\Lambda_2\,,
\ \
\kappa_2=0\,,
\\
\tau_{0}&=0\,,
\ \
\tau_{1}=\frac12\, \Lambda_2\,,
\ \
\tau_{2}=-m_{+}\,,
\ \
\tau_3=-\frac12\, \Lambda_2\,,
\label{sw2_p_2_1}
\end{split}
\end{align}
which are obtained from \eqref{sw2_p_3} by the limit \eqref{decouple_2_1} and scalings%
\footnote{
These scaling relations are associated with the scalings \eqref{sw_scaling_nf2_1}. 
From the definition \eqref{n_pt_amp} of the one-point amplitude $f_{1}(x)$, 
the string creation operators $\Psd(\ell)$ are scaled as 
$\Psd(\ell) \to (\Lambda_3/\Lambda_2)^{\ell+1} \Psd(\ell)$, 
and the commutation relations \eqref{string_com_rel} imply that 
the string annihilation operators $\Psi(\ell)$ are scaled as 
$\Psi(\ell) \to (\Lambda_2/\Lambda_3)^{\ell+1} \Psi(\ell)$.
Furthermore, by $f_{1}(x) = g_s^{-1} f_{1}^{(0)}(x) + O(g_s^{0})$ $(g_s \to 0)$, 
the scaling of $y$ implies that the string coupling constant $g_s$ is scaled as 
$g_s \to (\Lambda_2/\Lambda_3) g_s$. 
As a result, the scaling relations \eqref{ham_scaling_nf2_1} are obtained from the Hamiltonian \eqref{gen_hamiltonian}.
}
\begin{align}
\kappa_1 \to \frac{\Lambda_3}{\Lambda_2}\, \kappa_1\,,
\quad
\tau_k \to \left(\frac{\Lambda_2}{\Lambda_3}\right)^{k-2} \tau_k\,.
\label{ham_scaling_nf2_1}
\end{align}
Then, the Hamiltonian is
\begin{align}
\begin{split}
-\mathcal{H}^{SU(2)}_{N_f=2^{(1)}}&=
g_s^{-1} \kappa_{1}\, \Psi(1)
+2\sum_{k=1}^{3} \tau_k
\sum_{\ell \ge 1}
\ell\, \Psd(\ell+k-2)\Psi(\ell)
\\
&\ \
-g_s
\sum_{\ell, \ell' \ge 1}
\left(
\left(\ell+\ell'\right)
\Psd(\ell)\Psd(\ell')\Psi(\ell+\ell')
+
\ell \ell'\, 
\Psd(\ell+\ell') \Psi(\ell) \Psi(\ell')\right),
\label{sw2_ham_2_1}
\end{split}
\end{align}
and the two-cut ansatz for the equation \eqref{sp_curve_determine} provides 
the Seiberg-Witten curve \eqref{sw2_curve_2_1}, 
where the Coulomb branch parameter $u$ is introduced by
\begin{align}
\Lambda_2 \lim_{g_s \to 0} g_s \lim_{T \to \infty} \left\langle\mathrm{vac}\Big| \e^{-T \mathcal{H}_{N_f=2^{(1)}}^{SU(2)}}\,
\Psd(1) \Big|\mathrm{vac}\right\rangle
=
u - m_{+}^2
+\frac12\, \Lambda_2\,.
\label{sw2_coulomb_2_1}
\end{align}

\subsection*{\underline{$N_f=2$ (second realization)}}

The Seiberg-Witten curve of the second realization is obtained from the curve \eqref{sw2_curve_3} for $N_f=3$ by the decoupling limit
\begin{align}
\widetilde{m}_{+} \to \infty\,,
\quad
\Lambda_3 \to 0\,,
\quad
\widetilde{m}_{+}\, \Lambda_3=\Lambda_2^2\,
\ \textrm{(fixed)}\,,
\label{decouple_2_2}
\end{align}
as \cite{Gaiotto:2009hg,Gaiotto:2009ma},
\begin{align}
y^2 = 
-\frac{\Lambda_2^2}{x^{3} \left(x-1\right)}
-\frac{U}{x^{2}\left(x-1\right)}
+\frac{m_{0}^{2}}{x \left(x-1\right)}
+\frac{m_{2}^{2}}{x \left(x-1\right)^{2}}
=:
\frac{m_0^2}{p(x)^2}\, \sigma(x)\,,
\label{sw2_curve_2_2}
\end{align}
where $p(x)=x^2(x-1)$. 
The mass parameters $m_0, m_2$ are, respectively, associated with poles of degree 2 at $x=\infty, 1$ as
\begin{align}
\frac{1}{2\pi \sqrt{-1}} \oint_{x=\infty} ydx = m_0\,,
\quad
\frac{1}{2\pi \sqrt{-1}}\oint_{x=1} ydx &= m_2\,.
\label{pole_nf2_2}
\end{align}
In this case, $r=3$, $\mu=2$, $s=1$, and for the $W^{(3)}$-type Hamiltonian \eqref{gen_hamiltonian}, we set
\begin{align}
p_2=-1\,,\ \ p_3=1\,,
\ \
\kappa_1=\Lambda_2^2\,,
\ \
\kappa_2=0\,,
\ \
\tau_{0}=\tau_{1}=0\,,
\ \
\tau_{2}=-m_{0}-m_{2}\,,
\ \
\tau_3=m_0\,,
\label{sw2_p_2_2}
\end{align}
which are obtained from \eqref{sw2_p_3} by the limit \eqref{decouple_2_2}.
Then, the Hamiltonian is
\begin{align}
\begin{split}
-\mathcal{H}^{SU(2)}_{N_f=2^{(2)}}&=
g_s^{-1} \Lambda_2^2\, \Psi(1)
+2\sum_{k=2,3} \tau_k
\sum_{\ell \ge 1}
\ell\, \Psd(\ell+k-2)\Psi(\ell)
\\
&\ \
+g_s \sum_{k=2,3} p_k
\sum_{\ell, \ell' \ge 1}
\left(\ell+\ell'-k+2\right)
\Psd(\ell)\Psd(\ell')\Psi(\ell+\ell'-k+2)
\\
&\ \
+g_s \sum_{k=2,3} p_k
\sum_{\ell, \ell' \ge 1}
\ell \ell'\, 
\Psd(\ell+\ell'+k-2) \Psi(\ell) \Psi(\ell')\,,
\label{sw2_ham_2_2}
\end{split}
\end{align}
and the two-cut solution of the planar SD equation is given by 
the Seiberg-Witten curve \eqref{sw2_curve_2_1}, 
where the Coulomb branch parameter $U$ is introduced by
\begin{align}
2 m_0 \lim_{g_s \to 0} g_s \lim_{T \to \infty} \left\langle\mathrm{vac}\Big| \e^{-T \mathcal{H}_{N_f=2^{(2)}}^{SU(2)}}\,
\Psd(1) \Big|\mathrm{vac}\right\rangle
=
-U + \left(m_0+m_2\right)^2,
\label{sw2_coulomb_2_2}
\end{align}
which is understood to be reduced from \eqref{sw2_coulomb_3} by the limit \eqref{decouple_2_2}.

\subsection*{\underline{$N_f=1$}}

When $N_f=1$, the Seiberg-Witten curve is obtained from the curves \eqref{sw2_curve_2_1} or \eqref{sw2_curve_2_2} for $N_f=2$.
Here, we take the curve \eqref{sw2_curve_2_1} of the first realization, and 
by considering the decoupling limit
\begin{align}
\widetilde{m}_{+} \to \infty\,,
\quad
\Lambda_2 \to 0\,,
\quad
\widetilde{m}_{+}\, \Lambda_2^2=\Lambda_1^3\,
\ \textrm{(fixed)}\,,
\label{decouple_1}
\end{align}
and scalings
\begin{align}
x \to \frac{\Lambda_1}{\Lambda_2}\, x\,,
\quad
y \to \frac{\Lambda_2}{\Lambda_1}\, y\,,
\end{align}
the Seiberg-Witten curve of $N_f=1$ \cite{Gaiotto:2009hg,Gaiotto:2009ma},
\begin{align}
y^2=
\frac{\Lambda_1^{2}}{x^{3}}+\frac{u}{x^{2}}+\frac{m_{+} \Lambda_1}{x}+
\frac{\Lambda_1^{2}}{4}
=
\frac{\Lambda_1^2}{4x^4}
\left(
x^4 + \frac{4m_{+}}{\Lambda_1}x^3
+ \frac{4u}{\Lambda_1^{2}}x^2 + 4x
\right),
\label{sw2_curve_1}
\end{align}
with $r=\mu=2$, $s=1$, is obtained. 
The mass parameter $m_{+}$ is associated with a pole of degree 4 at $x=\infty$ as
\begin{align}
\frac{1}{2\pi \sqrt{-1}} \oint_{x=\infty} ydx = m_{+}\,.
\label{pole_nf1}
\end{align}
In this case, for the $W^{(3)}$-type Hamiltonian \eqref{gen_hamiltonian}, we set
\begin{align}
p_2=-1\,,
\ \
\kappa_1=\Lambda_1^2\,,
\ \
\kappa_2=0\,,
\ \
\tau_{0}=\tau_{1}=0\,,
\ \
\tau_{2}=-m_{+}\,,
\ \
\tau_3=-\frac12\, \Lambda_1\,,
\label{sw2_p_1}
\end{align}
which are obtained from \eqref{sw2_p_2_1} by the limit \eqref{decouple_1} and scalings
\begin{align}
\kappa_1 \to \frac{\Lambda_1}{\Lambda_2}\, \kappa_1\,,
\quad
\tau_k \to \left(\frac{\Lambda_2}{\Lambda_1}\right)^{k-2} \tau_k\,.
\label{ham_scaling_nf1}
\end{align}
Then, the Hamiltonian is%
\footnote{
It is remarkable that the Hamiltonian \eqref{sw2_ham_1} at $\Lambda_1=0$ agrees with the Calogero-Sutherland Hamiltonian at $\beta=1$ in the collective coordinate representation \cite{Awata:1994xd}. 
In this case, the Seiberg-Witten curve \eqref{sw2_curve_1} reduces to 
$y^2=u/x^2$, which does not have any branch points. 
Here we note that the Hamiltonian formalism proposed in this paper is also applicable to such cases with no branch points.
}
\begin{align}
\begin{split}
-\mathcal{H}^{SU(2)}_{N_f=1}&=
g_s^{-1} \Lambda_{1}^2\, \Psi(1)
+2\sum_{k=2,3} \tau_k
\sum_{\ell \ge 1}
\ell\, \Psd(\ell+k-2)\Psi(\ell)
\\
&\ \
-g_s
\sum_{\ell, \ell' \ge 1}
\left(
\left(\ell+\ell'\right)
\Psd(\ell)\Psd(\ell')\Psi(\ell+\ell')
+
\ell \ell'\, 
\Psd(\ell+\ell') \Psi(\ell) \Psi(\ell')\right),
\label{sw2_ham_1}
\end{split}
\end{align}
and the Seiberg-Witten curve \eqref{sw2_curve_1} is obtained as the two-cut solution of the planar SD equation, where
the Coulomb branch parameter $u$ is introduced by
\begin{align}
\Lambda_1 \lim_{g_s \to 0} g_s \lim_{T \to \infty} \left\langle\mathrm{vac}\Big| \e^{-T \mathcal{H}_{N_f=1}^{SU(2)}}\,
\Psd(1) \Big|\mathrm{vac}\right\rangle
=
u - m_{+}^2\,.
\label{sw2_coulomb_1}
\end{align}

\subsubsection*{\underline{$N_f=0$}}

When $N_f=0$, the Seiberg-Witten curve \eqref{pure_su2_sw} is obtained from the curve \eqref{sw2_curve_1} for $N_f=1$ by the decoupling limit
\begin{align}
m_{+} \to \infty\,,
\quad
\Lambda_1 \to 0\,,
\quad
m_{+}\, \Lambda_1^3=\Lambda^4\,
\ \textrm{(fixed)}\,,
\label{decouple_0}
\end{align}
and scalings
\begin{align}
x \to \frac{\Lambda_1^2}{\Lambda^2}\, x\,,
\quad
y \to \frac{\Lambda^2}{\Lambda_1^2}\, y\,,
\end{align}
as \cite{Gaiotto:2009hg,Gaiotto:2009ma},
\begin{align}
y^2=
\frac{\Lambda^{2}}{x^{3}}+\frac{u}{x^{2}}+\frac{\Lambda^{2}}{x}
=
\frac{\Lambda^2}{x^4}
\left(x^3 + \frac{u}{\Lambda^2} x^2 + x\right),
\label{sw2_curve_0}
\end{align}
with $r=\mu=2$, $s=1$.
In this case, different from the above cases, 
the Hamiltonian which leads to the curve \eqref{sw2_curve_0} 
is not simply obtained by the decoupling limit of the $W^{(3)}$-type Hamiltonian $\mathcal{H}^{SU(2)}_{N_f=1}$ in \eqref{sw2_ham_1}.
Instead, by considering the two-reduced $W^{(3)}$-type Hamiltonian \eqref{gen_hamiltonian_ai}, as discussed in Section \ref{subsubsec:pure_sw}, we find the Hamiltonian \eqref{pure_sw_hamiltonian}:
\begin{align}
-\mathcal{H}_{\mathrm{pure}}^{SU(2)}&=
2\Lambda \sum_{k=1,2}
\sum_{\ell \ge 1} \ell\, \Psd(\ell+2k-3)\Psi(\ell)
\nonumber\\
&\ \
+g_s \sum_{\ell, \ell' \ge 1}
\left(
\left(\ell+\ell'\right) \Psd(\ell)\Psd(\ell')\Psi(\ell+\ell')
+\frac{1}{4} \ell \ell'\, \Psd(\ell+\ell') \Psi(\ell) \Psi(\ell')\right),
\label{sw2_ham_0}
\end{align}
which provides the curve \eqref{sw2_curve_0}.

\begin{rem}\label{rem:sw_pure}
Under the rescaling
\begin{align}
g_s \to 2g_s\,, 
\quad
\Psd(\ell) \to \frac12 \Psd(\ell)\,, 
\quad
\Psi(\ell) \to 2 \Psi(\ell)\,, 
\end{align}
the Hamiltonian \eqref{sw2_ham_0} can be regarded as a $W^{(3)}$-type Hamiltonian, as follows:
\begin{align}
-\widetilde{\mathcal{H}}_{\mathrm{pure}}^{SU(2)}&=
2\Lambda \sum_{k=1,3}
\sum_{\ell \ge 1} \ell\, \Psd(\ell+k-2)\Psi(\ell)
\nonumber\\
&\ \
+g_s \sum_{\ell, \ell' \ge 1}
\left(
\left(\ell+\ell'\right) \Psd(\ell)\Psd(\ell')\Psi(\ell+\ell')
+ \ell \ell'\, \Psd(\ell+\ell') \Psi(\ell) \Psi(\ell')\right),
\label{pure_sw2_ham_c}
\end{align}
with $r=\mu=2$, $s=1$, and
\begin{align}
\kappa_1=\kappa_2=\kappa_3=\kappa_4=0\,,
\quad
\tau_0=\tau_2=0\,,
\quad
\tau_1=\tau_3=\Lambda\,,
\quad
p_2=1\,.
\end{align}
We then obtain the two-cut spectral curve from \eqref{sp_curve_determine}, as follows:
\begin{align}
y^2=
\frac{\Lambda^2}{x^4}
+
\frac{2\Lambda\,(\widetilde{f} + \Lambda)}{x^2}
+
\Lambda^2\
=
\frac{\Lambda^2}{x^4}\,
\biggl(x^4 + \frac{2\,(\widetilde{f} + \Lambda)}{\Lambda} x^2 + 1\biggr),
\label{sw_pure_app}
\end{align}
where
\begin{align}
\widetilde{f}:=\lim_{g_s \to 0} g_s \lim_{T \to \infty} \left\langle\mathrm{vac}\Big| \e^{-T \widetilde{\mathcal{H}}_{\mathrm{pure}}^{SU(2)}}\,
\Psd(1) \Big|\mathrm{vac}\right\rangle.
\end{align}
Note that by $x \to x^{1/2}$ and $y \to x^{-1/2}y$, the spectral curve \eqref{sw_pure_app} yields the Seiberg-Witten curve \eqref{sw2_curve_0}.
\end{rem}

\bibliographystyle{JHEP}
\bibliography{references.bib}

\end{document}